\documentclass[a4paper,USenglish,cleveref, autoref, thm-restate]{lipics-v2021}

\nolinenumbers

\usepackage{todonotes}

\usepackage{thm-restate}

\usepackage[ruled]{algorithm}
\usepackage[noend]{algpseudocode}

\usepackage{commath}
\usepackage{csquotes}

\usepackage{tcolorbox}

\usepackage{nicefrac}

\newcommand{\robotDiameter}[1]{\ensuremath{\mathrm{diam}_{#1}(t)}}
\newcommand{\globalDiameter}{\ensuremath{\mathrm{diam}(t)}}
\newcommand{\chull}[1]{\ensuremath{\mathrm{hull}_{#1}^{t}}}
\def\lambdaContracting/{$\lambda$-contracting}
\def\lambdaPrimeContracting/{$\lambda'$-contracting}
\def\lambdaPrimePrimeContracting/{$\lambda''$-contracting}
\def\lambdaGathering/{$\lambda$-contracting gathering protocol}
\def\alphaBetaContracting/{$\left(\alpha,\beta\right)$-contracting}
\def\alphaBetaPrimeContracting/{$\left(\alpha',\beta'\right)$-contracting}
\def\alphaBetaPrimePrimeContracting/{$\left(\alpha'',\beta''\right)$-contracting}
\def\alphaBetaGathering/{$\left(\alpha,\beta\right)$-contracting gathering protocol}
\def\alphaBetaPrimeGathering/{$\left(\alpha',\beta'\right)$-contracting gathering protocol}
\def\lambdaPrimeGathering/{$\lambda'$-contracting gathering protocol}
\def\alphaCentered/{$\alpha$-centered}
\def\lambdaCentered/{$\lambda$-centered}
\def\clLambdaContracting/{collisionless $\lambda$-contracting}
\def\clLambdaPrimeContracting/{collisionless $\lambda'$-contracting}
\def\lambdaNearGathering/{$\lambda$-contracting near-gathering protocol}

\def\pCL{\ensuremath{\calP^{cl}}}
\def\pTau{\ensuremath{\calP_\tau}}
\def\pCLtauEpsilon{\ensuremath{\pCL(\calP, \tau, \varepsilon)}}
\def\pVTau{\ensuremath{\calP^{V+\tau/2}}}
\newcommand{\fp}[1]{\ensuremath{\mathrm{target}_{#1}^{\calP}(t)}}
\newcommand{\fpTau}[1]{\ensuremath{\mathrm{target}_{#1}^{\pTau}(t)}}
\newcommand{\fpCLtauEpsilon}[1]{\ensuremath{\mathrm{target}_{#1}^{\pCLtauEpsilon}(t)}}
\newcommand{\fpCL}[1]{\ensuremath{\mathrm{target}_{#1}^{\pCL}(t)}}
\newcommand{\fpVTau}[1]{\ensuremath{\mathrm{target}_{#1}^{\pVTau}(t)}}
\newcommand{\ellp}[1]{\ensuremath{\mathrm{collvec}_{#1}^{\calP}(t)}}
\newcommand{\ellpTau}[1]{\ensuremath{\mathrm{collvec}_{#1}^{\pTau}(t)}}

\newcommand{\alphaCenterP}[1]{\ensuremath{\alpha\text{-}\mathrm{center}_{#1}^{\calP}(t)}}

\newcommand{\lambdaCenterP}[1]{\ensuremath{\lambda\text{-}\mathrm{center}_{#1}^{\calP}(t)}}

\newcommand{\targetPointDist}[1]{\ensuremath{d_{#1} \cdot \varepsilon \cdot \nicefrac{2}{\tau} \cdot \big|\ellpTau{#1}\big| }}

\def\lambdaPrimeValuePTau{\ensuremath{\lambda \cdot \frac{\tau}{4 \cdot (V + \tau)}}}
\def\lambdaPrimeValuePcl{\ensuremath{\lambda \cdot \frac{\tau}{4 \cdot (V + \tau)} \cdot (1-\varepsilon)}}

\def\runningTimeLemmaCL{\ensuremath{\frac{32 \cdot \pi \cdot \Delta^2}{\lambda^2 \cdot \tau}}}
\def\runningTimeLemmaCLPrime{\ensuremath{\frac{32 \cdot \pi \cdot \Delta^2}{\lambda'^2 \cdot \tau}}}
\newcommand{\collisionPoints}[2]{\ensuremath{\mathrm{collisionPoints}^{#1}_{#2}(R_{#2}, t)}}

\def\baseSegment/{\ensuremath{S_{\alpha}}}
\def\betaHalfSegment/{\ensuremath{S_{\alpha} \bigl(\frac{\beta}{2}\bigr)}}
\def\mainSegment/{\ensuremath{S_{\alpha} \bigl(\frac{\alpha \cdot \beta^2}{4}\bigr)}}
\def\intermediateSegment/{\ensuremath{S_{\alpha} \bigl(\frac{\alpha \cdot \beta}{4}\bigr)}}
\def\baseSegmentC/{\ensuremath{S_{\lambda \cdot \tau}}}

\def\baseSegmentLambda/{\ensuremath{S_{\lambda}}}
\def\betaHalfSegmentLambda/{\ensuremath{S_{\lambda} \bigl(\nicefrac{1}{2}\bigr)}}
\def\mainSegmentLambda/{\ensuremath{S_{\lambda} \bigl(\nicefrac{\lambda}{4}\bigr)}}
\def\intermediateSegmentLambda/{\ensuremath{S_{\lambda} \bigl(\nicefrac{\lambda}{4}\bigr)}}
\def\betaCSegmentLamda/{\ensuremath{S_{\alpha \cdot \tau} \bigl(\frac{\beta}{2}\bigr)}}
\def\baseSegmentCLambda/{\ensuremath{S_{\alpha \cdot \tau}}}

\def\baseCap/{\ensuremath{\mathrm{HSC}_{\alpha}}}
\def\betaHalfCap/{\ensuremath{\mathrm{HSC}_{\alpha} \bigl(\frac{\beta}{2}\bigr)}}
\def\mainCap/{\ensuremath{\mathrm{HSC}_{\alpha} \bigl(\frac{\alpha \cdot \beta^2}{4}\bigr)}}
\def\intermediateCap/{\ensuremath{\mathrm{HSC}_{\alpha} \bigl(\frac{\alpha \cdot \beta}{4}\bigr)}}

\def\baseCapLambda/{\ensuremath{\mathrm{HSC}_{\lambda}}}
\def\betaHalfCapLambda/{\ensuremath{\mathrm{HSC}_{\lambda} \bigl(\frac{1}{2}\bigr)}}
\def\mainCapLambda/{\ensuremath{\mathrm{HSC}_{\lambda} \bigl(\frac{\lambda}{4}\bigr)}}
\def\intermediateCapLambda/{\ensuremath{\mathrm{HSC}_{\lambda} \bigl(\frac{\lambda}{4}\bigr)}}

\def\gtc/{\textsc{Go-To-The-Center}}
\def\gtcShort/{\textsc{GtC}}
\def\gtmd/{\textsc{Go-To-The-Middle-Of-The-Diameter}}
\def\gtmdShort/{\textsc{GtMD}}
\def\gtcdmb/{\textsc{Go-To-The-Center-Of-The-Diameter-MinBox}}
\def\gtcdmbShort/{\textsc{GtCDMB}}

\def\oblot/{\ensuremath{\mathcal{OBLOT}}}
\def\lumi/{\ensuremath{\mathcal{LUMI}}}

\newcommand{\fsync}{\textsc{$\mathcal{F}$sync}}
\newcommand{\ssync}{\textsc{$\mathcal{S}$sync}}
\newcommand{\async}{\textsc{$\mathcal{A}$sync}}
\def\LCM/{\textsc{LCM}}
\def\Look/{\textsc{Look}}
\def\Compute/{\textsc{Compute}}
\def\Move/{\textsc{Move}}

\def\gathering/{\textsc{Gathering}}
\def\nearGathering/{\textsc{Near-Gathering}}

\def\clContractingStrat/{\emph{collisionless-contracting-strategy}}
\def\nonclContractingStrat/{\emph{arbitrary-contracting-strategy}}

\def\calO{\mathcal{O}}
\def\calP{\mathcal{P}}

\def\cNG/{\ensuremath{c_{\mathrm{ng}}}}

\newcommand{\vubg}[1]{\ensuremath{\mathrm{UBG}^{V}(#1)}}
\newcommand{\tauUbg}[1]{\ensuremath{\mathrm{UBG}^{V+\tau}(#1)}}

\newtheorem{innercustomthm}{Theorem}
\newenvironment{customthm}[1]
{\renewcommand\theinnercustomthm{#1}\innercustomthm}
{\endinnercustomthm}

\title{A Unifying Approach to Efficient (Near)-Gathering of Disoriented Robots with Limited Visibility} %TODO Please add
\titlerunning{Efficient (Near)-Gathering of Disoriented Robots with Limited Visibility} %TODO optional, please use if title is longer than one line

\author{Jannik {Castenow}}{Heinz Nixdorf Institute \& Computer Science Deparment, Paderborn University, Fürstenallee 11, 33102 Paderborn, Germany  }{jannik.castenow@upb.de}{https://orcid.org/0000-0002-8585-4181}{}

\author{Jonas {Harbig}}{Heinz Nixdorf Institute \& Computer Science Deparment, Paderborn University, Fürstenallee 11, 33102 Paderborn, Germany }{jonas.harbig@upb.de}{}{}

\author{Daniel {Jung}}{Heinz Nixdorf Institute \& Computer Science Deparment, Paderborn University, Fürstenallee 11, 33102 Paderborn, Germany  }{jungd@hni.upb.de}{https://orcid.org/0000-0001-8270-8130}{}

\author{Peter {Kling}}{Department of Informatics, Universität Hamburg, Vogt-Kölln-Str. 30, 22527 Hamburg, Germany}{ peter.kling@uni-hamburg.de}{https://orcid.org/0000-0003-0000-8689}{}

\author{Till {Knollmann}}{Heinz Nixdorf Institute \& Computer Science Deparment, Paderborn University, Fürstenallee 11, 33102 Paderborn, Germany  }{tillk@mail.upb.de}{https://orcid.org/0000-0003-2014-4696}{}

\author{Friedhelm {Meyer auf der Heide}}{Heinz Nixdorf Institute \& Computer Science Deparment, Paderborn University, Fürstenallee 11, 33102 Paderborn, Germany }{fmadh@upb.de}{}{}

\funding{This work was partially supported by the German Research Foundation (DFG) under the project number ME 872/14-1.}

\authorrunning{J. Castenow et al.} %TODO mandatory. First: Use abbreviated first/middle names. Second (only in severe cases): Use first author plus 'et al.'

\Copyright{Jane Open Access and Joan R. Public} %TODO mandatory, please use full first names. LIPIcs license is "CC-BY";  http://creativecommons.org/licenses/by/3.0/


 %TODO mandatory; please add comma-separated list of keywords

\category{} %optional, e.g. invited paper

\relatedversion{}

\EventEditors{John Q. Open and Joan R. Access}
\EventNoEds{2}
\EventLongTitle{42nd Conference on Very Important Topics (CVIT 2016)}
\EventShortTitle{CVIT 2016}
\EventAcronym{CVIT}
\EventYear{2016}
\EventDate{December 24--27, 2016}
\EventLocation{Little Whinging, United Kingdom}
\EventLogo{}
\SeriesVolume{42}
\ArticleNo{23}
\begin{document}

\maketitle

\begin{abstract}
We consider a swarm of $n$ robots in a $d$-dimensional Euclidean space.
The robots are oblivious (no persistent memory), disoriented (no common coordinate system/compass), and have limited visibility (observe other robots up to a constant distance).
The basic formation task \gathering/ requires that all robots reach the same, not predefined position.
In the related \nearGathering/ task, they must reach distinct positions in close proximity such that every robot sees the entire swarm.
In the considered setting, \gathering/ can be solved in $\calO(n + \Delta^2)$ synchronous rounds both in two and three dimensions, where $\Delta$ denotes the initial maximal distance of two robots \cite{DBLP:journals/trob/AndoOSY99,DBLP:conf/sirocco/BraunCH20,DBLP:conf/spaa/DegenerKLHPW11}.

In this work, we formalize a key property of efficient \gathering/ protocols and use it to define \emph{\lambdaContracting/ protocols}.
Any such protocol gathers $n$ robots in the $d$-dimensional space in $\calO(\Delta^2)$ synchronous rounds.
Moreover, we prove a corresponding lower bound stating that any protocol in which robots move to target points inside the local convex hulls of their neighborhoods -- $\lambda$-contracting protocols have this property -- requires $\Omega(\Delta^2)$ rounds to gather all robots.
Among others, we prove that the $d$-dimensional generalization of the \gtcShort/-protocol \cite{DBLP:journals/trob/AndoOSY99} is \lambdaContracting/.
Remarkably, our improved and generalized runtime bound is independent of $n$ and~$d$.
%The independence of $d$ answers an open research question posed in \cite{DBLP:conf/sirocco/BraunCH20}.

We also introduce an approach to make any \lambdaContracting/ protocol collision-free (robots never occupy the same position) to solve \nearGathering/.
The resulting protocols maintain the runtime of $\Theta (\Delta^2)$ and work even in the semi-synchronous model. This yields the first \nearGathering/ protocols for disoriented robots and the first proven runtime bound.
In particular, combined with results from \cite{DBLP:journals/dc/FlocchiniPSV17} for robots with global visibility, we obtain the first protocol to solve \textsc{Uniform Circle Formation} (arrange the robots on the vertices of a regular $n$-gon) for oblivious, disoriented robots with limited visibility.
\end{abstract}

\section{Introduction}%
\label{sec:introduction}

Envision a huge swarm of $n$ robots spread in a $d$-dimensional Euclidean space that must solve a \emph{formation task} like \gathering/ (moving all robots to a single, not pre-determined point) or \textsc{Uniform-Circle} (distributing the robots over the vertices of a regular $n$-gon).
Whether and how efficiently a given task is solvable varies largely with the robots' capabilities (local vs.~global visibility, memory vs.~memory-less, communication capabilities, common orientation vs.~disorientation).
While classical results study what capabilities the robots need \emph{at least} to solve a given task, our focus lies on \emph{how fast} a given formation task can be solved assuming simple robots.
Specifically, we consider the \gathering/ problem and the related \nearGathering/ problem for oblivious, disoriented robots with a limited viewing range.

\gathering/ is the most basic formation task and a standard benchmark to compare robot models~\cite{DBLP:series/lncs/FlocchiniPS19}.
The robots must gather at the same, not predefined, position.
Whether or not \gathering/ is solvable depends on various robot capabilities.
It is easy to see that robots can solve \gathering/ in case they have unlimited visibility (can observe all other robots) and operate fully synchronously \cite{cohenConvergencePropertiesGravitational2005}.
However, as soon as the robots operate asynchronously, have only limited visibility, or do not agree on common coordinate systems, the problem gets much harder or even impossible to solve (see \Cref{section:relatedWork} for a comprehensive discussion).
A well-known protocol to solve \gathering/ of robots with limited visibility is the \gtc/ (\gtcShort/) protocol that moves each robot towards the center of the smallest enclosing circle of all observable robots~\cite{DBLP:journals/trob/AndoOSY99}.
\gtcShort/ gathers all robots in $\calO\left(n+ \Delta^2\right)$ synchronous rounds, where the \emph{diameter} $\Delta$ denotes the initial maximal distance of two robots~\cite{DBLP:conf/spaa/DegenerKLHPW11}.
The term $n$ upper bounds the number of rounds in which robots \emph{collide} (move to the same position), while $\Delta^2$ results from how quickly the global smallest enclosing circle shrinks.
Hence, \gtcShort/ not only forces the robots to collide in the final configuration but also incurs several collisions during \gathering/.
Such collisions are fine for point robots in theoretical models but a serious problem for physical robots that cannot occupy the same position.
This leads us to \nearGathering/, which requires the robots to move \emph{collision-free} to \emph{distinct} locations such that every robot can observe the entire swarm despite its limited visibility~\cite{DBLP:journals/dc/PagliPV15}.
Requiring additionally that, eventually, robots \emph{simultaneously} (within one round/epoch) \emph{terminate}, turns \nearGathering/ into a powerful subroutine for more complex formation tasks like \textsc{Uniform Circle}.
Once all robots see the entire swarm and are simultaneously aware of that, they can switch to the protocol of \cite{DBLP:journals/dc/FlocchiniPSV17} to build a uniform circle.
Although that protocol is designed for robots with a global view, we can use it here since solving \nearGathering/ grants the robots de facto a global view.
Note the importance of \emph{simultaneous} termination, as otherwise, some robots might build the new formation while others are still gathering, possibly disconnecting some robots from the swarm.

\medskip\noindent\textbf{\textsf{Robot Model.\;}}
We assume the standard \oblot/ model~\cite{DBLP:series/lncs/FlocchiniPS19} for oblivious, point-shaped robots in $\mathbb{R}^{d}$.
The robots are \emph{anonymous} (no identifiers), \emph{homogeneous} (all robots execute the same protocol), \emph{identical} (same appearance), \emph{autonomous} (no central control) and \emph{deterministic}.
Moreover, we consider disoriented robots with limited visibility.
Disorientation means that a robot observes itself at the origin of its local coordinate system, which can be arbitrarily rotated and inverted compared to other robots.
The disorientation is \emph{variable}, i.e., the local coordinate system might differ from round to round.
Limited visibility implies that each robot can observe other robots only up to a constant distance.
The robots do not have \emph{multiplicity detection}, i.e., robots observe only a single robot in case multiple robots are located at the same position.
Furthermore, time is divided into discrete \LCM/-cycles (\emph{rounds}) consisting of the operations \Look/, \Compute/ and \Move/.
During its \Look/ operation, a robot takes a snapshot of all visible robots, which is used in the following \Compute/ operation to compute a \emph{target point}, to which the robot moves in the \Move/ operation.
Moves are \emph{rigid} (a robot always reaches its target point) and depend solely on observations from the last \Look/ operation (robots are oblivious).
The time model can be fully synchronous (\fsync{}; all robots are active each round and operations are executed synchronously), semi-synchronous (\ssync{}; a subset of robots is active each round and operations are executed synchronously), or completely asynchronous (\async{}).
The \ssync{} and \async{} schedulings of the robots are \emph{fair}, i.e., each robot is activated infinitely often.
Time is measured in rounds in \fsync{} and in \emph{epochs} (the smallest number of rounds such that all robots finish one \LCM/-cycle) in \ssync{} or \async{}.

\medskip\noindent\textbf{\textsf{Results in a Nutshell.\;}}
For \gathering/ of oblivious, disoriented robots with limited visibility in $\mathbb{R}^d$, we introduce the class of \emph{\lambdaContracting/} protocols for a constant $\lambda \in \intoc{0,1}$.
For instance, the well-known \gtcShort/~\cite{DBLP:journals/trob/AndoOSY99} and several other \gathering/ protocols are \lambdaContracting/.
We prove that every \lambdaContracting/ protocol gathers a swarm of diameter $\Delta$ in $\calO(\Delta^2)$ rounds.
We also prove a matching lower bound for any protocol in which robots always move to points inside the convex hull of their neighbors, including themselves.
While our results for \gathering/ assume the \fsync{} model\footnote{%
	With the considered robot capabilities, \gathering/ is impossible in \ssync{} or \async{}~\cite{prencipeImpossibilityGatheringSet2007}.
}, for \nearGathering/ we also consider \ssync{}.
We show how to transform any \lambdaContracting/ protocol into a \emph{collision-free} \lambdaContracting/ protocol to solve \nearGathering/ while maintaining a runtime of $\calO(\Delta^2)$.

\subsection{Related Work}%
\label{section:relatedWork}

One important topic of the research area of distributed computing by mobile robots is \emph{pattern formation} problems, i.e., the question of which patterns can be formed by a swarm of robots and which capabilities are required.
For instance the \textsc{Arbitrary Pattern Formation} problem requires the robots to form an arbitrary pattern specified in the input
\cite{DBLP:conf/podc/DasFSY10,DBLP:conf/wdag/DieudonnePV10,DBLP:journals/tcs/FlocchiniPSW08,suzukiDistributedAnonymousMobile1999,DBLP:journals/tcs/YamashitaS10,DBLP:conf/podc/YamauchiUY16}.
The patterns \emph{point} and \emph{uniform circle} play an important role since these are the only two patterns that can be formed starting from \emph{any} input configuration due to their high symmetry \cite{suzukiDistributedAnonymousMobile1999}.
In the following, we focus on the pattern \emph{point}, more precisely on the \gathering/, \textsc{Convergence} and \nearGathering/ problems.
While \gathering/ requires that all robots move to a single (not predefined) point in finite time, \textsc{Convergence} demands that for all $\varepsilon > 0$, there is a point in time such that the maximum distance of any pair of robots is at most $\varepsilon$ and this property is maintained (the robots \emph{converge} to a single point).
\nearGathering/ is closely related to the \textsc{Convergence} problem by robots with limited visibility.
Instead of converging to a single point, \nearGathering/ is solved as soon as all robots are located at \emph{distinct} locations within a small area.
For a more comprehensive overview of other patterns and models, we refer to~\cite{DBLP:series/lncs/11340}.

\medskip\noindent\textbf{\textsf{Possibilities \& Impossibilities.\;}}
In the context of robots with \emph{unlimited visibility}, \gathering/ can be solved under the \fsync{} scheduler by disoriented and oblivious robots without multiplicity detection \cite{cohenConvergencePropertiesGravitational2005}.
Under the same assumptions, \gathering/ is impossible under the \ssync{} and \async{} schedulers \cite{prencipeImpossibilityGatheringSet2007}.
Multiplicity detection plays a crucial role: at least $3$ disoriented robots with multiplicity detection can be gathered in \async{} (and thus also \ssync{}) \cite{cieliebakDistributedComputingMobile2012}.
The case of $2$ robots remains impossible \cite{suzukiDistributedAnonymousMobile1999}.
Besides multiplicity detection, an agreement on one axis of the local coordinate systems also allows the robots to solve \gathering/ in \async{} \cite{bhagatFaulttolerantGatheringAsynchronous2016}.
\textsc{Convergence} requires less assumptions than \gathering/.
No multiplicity detection is needed for the \async{} scheduler~\cite{cohenConvergencePropertiesGravitational2005}.

Under the assumption of \emph{limited visibility}, disoriented robots without multiplicity detection can be gathered in \fsync{} \cite{DBLP:journals/trob/AndoOSY99} with the \gtcShort/ protocol that moves every robot towards the center of the smallest circle enclosing its neighborhood.
\gtcShort/ has also been generalized to three dimensions \cite{DBLP:conf/sirocco/BraunCH20}.
In \async{}, current solutions require more capabilities: \gathering/ can be achieved by robots with limited visibility that agree additionally on the axes and orientation of their local coordinate systems \cite{flocchiniGatheringAsynchronousRobots2005}.
It is open whether fewer assumptions are sufficient to solve \gathering/ of robots with limited visibility in \ssync{} or \async{}.
In \ssync{}, \textsc{Convergence}  can be solved even by disoriented robots with limited visibility without multiplicity detection \cite{DBLP:journals/trob/AndoOSY99}.
However, similar to \gathering/, it is still open whether disoriented robots with limited visibility can solve \textsc{Convergence} under the \async{} scheduler.
Recently, it could be shown that multiplicity detection suffices to solve \textsc{Convergence}  under the more restricted $k$-\async{} scheduler.
The constant $k$ bounds how often other robots can be activated within one \LCM/ cycle of a single robot \cite{katreniakConvergenceLimitedVisibility2011,kirkpatrickSeparatingBoundedUnbounded2021}.

The \nearGathering/ problem has been introduced in \cite{DBLP:conf/sirocco/PagliPV12,DBLP:journals/dc/PagliPV15} together with an algorithm to solve \nearGathering/ by
robots with limited visibility and agreement on one axis of their local coordinate systems under the \async{} scheduler.
An important tool to prevent collisions is a well-connected initial configuration, i.e., the initial configuration is connected concerning the \emph{connectivity range} which is by an additive constant smaller than the viewing range \cite{DBLP:conf/sirocco/PagliPV12,DBLP:journals/dc/PagliPV15}.
In an earlier work, \nearGathering/ has been used as a subroutine to solve \textsc{Arbitrary Pattern Formation} by robots with limited visibility \cite{yamauchiPatternFormationMobile2013}.
The solution, however, uses infinite persistent memory at each robot.
Further research directions study \gathering/ and \textsc{Convergence} under crash faults or Byzantine faults \cite{agmonFaultTolerantGatheringAlgorithms2006a,augerCertifiedImpossibilityResults2013,bhagatFaultTolerantGatheringAsynchronous2015,bhagatFaulttolerantGatheringSemisynchronous2017,bouzidGatheringMobileRobots2013,bouzidByzantineConvergenceRobot2009,bouzidOptimalByzantineresilientConvergence2010,bramasWaitFreeGatheringChirality2015,defagoFaultTolerantSelfstabilizingMobile2006,defagoFaultByzantineTolerant2016,izumiBriefAnnouncementBGSimulation2011,pattanayakFaultTolerantGatheringMobile2017} or inaccurate measurement and movement sensors of the robots \cite{cohenConvergenceAutonomousMobile2008a,izumiBGsimulationByzantineMobile2011,izumiGatheringProblemTwo2012,kirkpatrickSeparatingBoundedUnbounded2021}.

\vspace*{-0.04cm}
\medskip\noindent\textbf{\textsf{Runtimes.\;}}
Considering disoriented robots with \emph{unlimited} visibility, it is known that \textsc{Convergence} can be solved in $\calO(n \cdot \log \nicefrac{\Delta}{\varepsilon})$ epochs under the \async{} scheduler, where the diameter $\Delta$ denotes the initial maximum distance of two robots \cite{cord-landwehrNewApproachAnalyzing2011} (initially a bound of  $\calO(n^2 \cdot \log \nicefrac{\Delta}{\varepsilon})$ has been proven in \cite{cohenConvergencePropertiesGravitational2005}).
When considering disoriented robots with \emph{limited} visibility and the \fsync{} scheduler, the \gtcShort/ algorithm solves \gathering/ both in two and three dimensions in $\Theta(n + \Delta^2)$ rounds \cite{DBLP:conf/sirocco/BraunCH20,DBLP:conf/spaa/DegenerKLHPW11}.
It is conjectured that the runtime  is optimal in worst-case instances, where $\Delta \in \Omega(n)$ \cite{DBLP:conf/sirocco/BraunCH20,DBLP:conf/algosensors/CastenowHJKH21}.
There is some work achieving faster runtimes for slightly different models: robots on a grid in combination with the $\mathcal{LUMI}$ model (constant sized local communication via lights) \cite{DBLP:conf/ipps/AbshoffC0JH16,DBLP:conf/spaa/Cord-Landwehr0J16}, predefined neighborhoods in a closed chain \cite{DBLP:conf/ipps/AbshoffC0JH16,DBLP:conf/algosensors/CastenowHJKH21} or agreement on one axis of the local coordinate systems \cite{DBLP:journals/information/PoudelS21}.
Also, a different time model -- the \emph{continuous} time model, where the movement of robots is defined for each \emph{real} point in time by a bounded velocity vector -- leads to faster runtimes:
There are protocols with a runtime of $\calO\left(n\right)$ \cite{DBLP:journals/tcs/BrandesDKH13,DBLP:journals/topc/DegenerKKH15}.
In \cite{DBLP:journals/tcs/LiMHP21}, a more general class of continuous protocols has been introduced, the \emph{contracting} protocols.
Contracting protocols demand that each robot part of the global convex hull of all robots' positions moves with full speed towards the inside.
Any contracting protocol gathers all robots in time $\mathcal{O}\left(n \cdot \Delta\right)$.
One such protocol also needs a runtime of $\Omega \left(n \cdot \Delta\right)$ in a specific configuration.
For instance, the continuous variant of \gtcShort/ is contracting \cite{DBLP:journals/tcs/LiMHP21} but also the protocols of \cite{DBLP:journals/tcs/BrandesDKH13,DBLP:journals/topc/DegenerKKH15}.
The class of contracting protocols also generalizes to three dimensions with an upper time bound of $\calO\bigl(n^{\nicefrac{3}{2}} \cdot \Delta\bigr)$  \cite{DBLP:conf/sirocco/BraunCH20}.

\subsection{Our Contribution \& Outline}%
\label{sec:contribution}

In the following, we provide a detailed discussion of our results and put them into context concerning the related results discussed in \cref{section:relatedWork}.
Our results assume robots located in $\mathbb{R}^d$ and the \oblot/ model for deterministic, disoriented robots with limited visibility.

\medskip\noindent\textbf{\textsf{\gathering/.\;}}
Our first main contribution is introducing a large class of \gathering/ protocols in \fsync{} that contains several natural protocols such as \gtcShort/.
We prove that \emph{every} protocol from this class gathers in $\calO(\Delta^2)$ rounds, where the diameter $\Delta$ denotes the initial maximal distance between two robots.
Note that, the bound of $\calO\left(\Delta^2\right)$ not only reflects how far a given initial swarm is from a gathering but also improves the \gtcShort/ bound from $\calO\left(n+\Delta^2\right)$ to $\calO\left(\Delta^2\right)$.
We call this class \emph{\lambdaContracting/ protocols}.
Such protocols restrict the allowed target points to a specific subset of a robot's local convex hull (formed by the positions of all visible robots, including itself) in the following way.
Let $diam$ denote the diameter of a robot's local convex hull.
Then, a target point $p$ is an allowed target point if it is the center of a line segment of length $\lambda \cdot diam$, completely contained in the local convex hull.
This guarantees that the target point lies far enough inside the local convex hull (at least along one dimension) to decrease the swarm's diameter sufficiently.
See \cref{figure:lambdaCenteredPoints} for an illustration.

\begin{figure}[htb]
	\centering
	\includegraphics[width =0.8\textwidth]{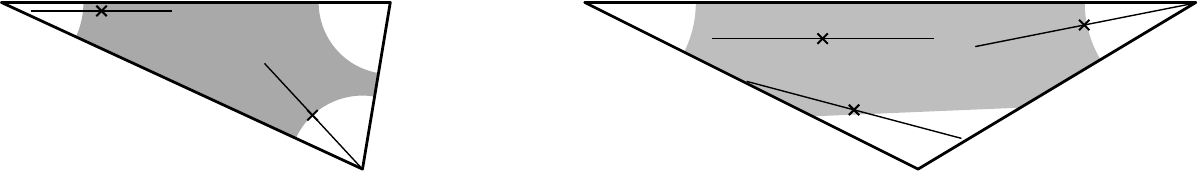}
	\caption{%
		Two local convex hulls, each formed by $3$ robots.
		The gray area marks valid target points of \lambdaContracting/ protocols.
		The exemplary line segments all have length $\lambda \cdot diam$, where $diam$ is the diameter of the respecting convex hull.
		On the left $\lambda = \nicefrac{4}{7}$, on the right $\lambda = \nicefrac{4}{11}$.
	}
	\label{figure:lambdaCenteredPoints}
\end{figure}

We believe these \lambdaContracting/ protocols encapsulate the core property of fast \gathering/ protocols.
Their analysis is comparatively clean, simple, and holds for any dimension $d$.
Thus, by proving that (the generalization of) \gtcShort/ is \lambdaContracting/ for arbitrary dimensions, we give the first protocol that provably gathers in $\calO(\Delta^2)$ rounds for any dimension.
As a strong indicator that our protocol class might be asymptotically optimal, we prove that every \gathering/ protocol for deterministic, disoriented robots whose target points lie \emph{always inside} the robots' local convex hulls requires $\Omega(\Delta^2)$ rounds.
Stay in the convex hull of visible robots is a natural property for any known protocol designed for oblivious, disoriented robots with limited visibility.
Thus, reaching a subquadratic runtime – if at all possible – would require the robots to compute target points outside of their local convex hulls sufficiently often.

\medskip\noindent\textbf{\textsf{\nearGathering/.\;}}
Our second main contribution proves that any \lambdaContracting/ protocol for \gathering/ can be transformed into a collision-free protocol that solves \nearGathering/ in $\calO(\Delta^2)$ rounds (\fsync{}) or epochs (\ssync{}).
As in previous work on the \nearGathering/ problem~\cite{DBLP:conf/sirocco/PagliPV12,DBLP:journals/dc/PagliPV15}, our transformed protocols require that the initial swarm is \emph{well-connected}, i.e., the swarm is connected with respect to the \emph{connectivity range} of $V$ and the robots have a viewing range of $V+\tau$, for a constant $\tau$.
The adapted protocols ensure that the swarm stays connected concerning the connectivity range.

The well-connectedness serves two purposes.
First, it allows a robot to compute its target point under the given \lambdaContracting/ protocol and the target points of nearby robots to prevent collisions.
Its second purpose is to enable termination:
Once there is a robot whose local convex hull has a diameter at most $\tau$, \emph{all} robots must have distance at most $\tau$, as otherwise, the swarm would not be connected concerning the connectivity range $V$.
Thus, all robots can \emph{simultaneously} decide (in the same round in \fsync{} and within one epoch in \ssync{}) whether \nearGathering/ is solved.
If the swarm is not well-connected, it is easy to see that such a simultaneous decision is impossible\footnote{%
	Consider a protocol that solves \nearGathering/ for a swarm of two robots and terminates in the \fsync{} model.
	Fix the last round before termination and add a new robot visible to only one robot (the resulting swarm is not connected concerning $V$).
	One of the original two robots still sees the same situation as before and will terminate, although \nearGathering/ is not solved.
}.
The simultaneous termination also allows us to derive the first protocol to solve \textsc{Uniform-Circle} for disoriented robots with limited visibility.
Once the robots' local diameter (and hence also the global diameter) is less than $\tau$, they essentially have a global view.
As the \textsc{Uniform Circle} protocol from~\cite{DBLP:journals/dc/FlocchiniPSV17} maintains the small diameter, it can be used after the termination of our \nearGathering/ protocol without any modification.

\medskip\noindent\textbf{\textsf{Outline.\;}}
\Cref{sec:model_and_preliminaries} introduces various notations.
\lambdaContracting/ protocols are introduced in \Cref{sec:gathering:class_definition}.
Upper and lower runtime bounds are provided in \Cref{section:alphaBetaProtocolsUpperBound}.
The section is concluded with three exemplary \lambdaContracting/ protocols, including \gtcShort/ (\Cref{section:exampleProtocols}).
\Cref{section:collisionlessProtocols} discusses the general approach to transform any \lambdaContracting/ protocol (in any dimension) into a collision-free protocol to solve \nearGathering/.
Finally, the paper is concluded, and future research questions are addressed in \Cref{section:conclusion}.
Due to space constraints, all proofs and additional information are deferred to the appendix.
\Cref{section:appendixSection3,section:appendixSection4}  contain proofs and additional material to \Cref{section:alphaBetaContractingStrategies,section:collisionlessProtocols}.

\section{Notation}%
\label{sec:model_and_preliminaries}

We consider a swarm of $n$  robots $R= \set{r_1, \dots, r_n}$ moving in a $d$-dimensional Euclidean space $\mathbb{R}^d$.
Initially, the robots are located at pairwise distinct locations.
We denote by $p_i(t)$ the position of robot $r_i$ in a global coordinate system (not known to the robots) in round~$t$.
Robots have a \emph{limited visibility}, i.e., they can observe other robots only up to a constant distance.
We distinguish the terms \emph{viewing} range and \emph{connectivity} range.
In both cases, the initial configuration is connected concerning the connectivity range.
More formally, let $V$ denote the connectivity range and $\vubg{t} = (R, E^{V}(t))$ the Unit Ball Graph with radius $V$, where $\{r_i,r_j\} \in E^{V}(t)$ if and only if $|p_i(t)-p_j(t)| \leq V$, where $|\cdot|$ represents the Euclidean norm.
The initial Unit Ball Graph \vubg{0} is always connected.
The connectivity and viewing ranges are equal when we study the \gathering/ problem.
In the context of \nearGathering/, the viewing range is larger than the connectivity range.
More formally, the viewing range is $V + \tau$, for a constant $0 < \tau \leq \nicefrac{2}{3}V$.
Thus, the robots can observe other robots at a distance of at most $V+\tau$.
The viewing range of $V+\tau$ induces $\tauUbg{t} = (R, E^{V+\tau}(t))$, the Unit Ball Graph with radius $V+\tau$, where $\{r_i,r_j\} \in E^{V+\tau}(t)$ if and only if $|p_i(t)-p_j(t)| \leq V +\tau$.
Two robots are neighbors at round $t$ if their distance is at most the viewing range ($V$ for \gathering/ and $V+\tau$ for \nearGathering/).
The set $N_i(t)$ contains all neighbors of $r_i$ in round $t$, including $r_i$.
Additionally, \chull{i} denotes the \emph{local convex hull} of all neighbors of $r_i$, i.e.,\ the smallest convex polytope that encloses the positions of all robots in $N_i(t)$, including $r_i$.
We define \globalDiameter{} as the maximum distance of any pair of robots at time $t$.
Moreover, $\Delta := \mathrm{diam}\left(0\right)$, i.e., the maximum distance of any pair of robots in the initial configuration.
Lastly, \robotDiameter{i} denotes the maximum distance of any two neighbors of $r_i$ in round $t$.

\medskip\noindent\textbf{\textsf{Discrete Protocols.\;}}
A discrete robot formation protocol $\mathcal{P}$ specifies for every round $t \in \mathbb{N}_0$ how each robot determines its target point, i.e., it is an algorithm that computes the target point \fp{i} of each robot in the \Compute/ operation based upon its snapshot taken during \Look/.
To simplify the notation, \fp{i} might express the target point of $r_i$ either in the local coordinate system of $r_i$ or in a global coordinate system (not known to $r_i$) -- the concrete meaning is always clear based on the context.
Finally, during \Move/, each robot moves to the position computed by $\mathcal{P}$, i.e., $p_i(t+1) = \fp{i}$ for all robots $r_i$.

\medskip\noindent\textbf{\textsf{Problem Statements.\;}}
The \gathering/ problem requires all robots to gather at a single, not predefined point.
More formally, \gathering/ is solved, if there exists a time $t \in \mathbb{N}_0$ such that  $\mathrm{diam}(t) = 0$.
While the \gathering/ problem clearly demands that more than one robot occupies the same position, this is prohibited in the \nearGathering/ problem.
Two robots $r_i$ and $r_j$ \emph{collide} in round $t$ if $p_i(t) = p_j(t)$.
A discrete robot formation protocol  is \emph{collisionless}, if there is no round $t' \in \mathbb{N}_0$ with a collision.
\nearGathering/  requires all robots to maintain distinct locations, become mutually visible, and be aware of this fact in the same round/epoch.
More formally, \nearGathering/ is solved if there is a time $t' \in \mathbb{N}_0$ and a constant $0 \leq \cNG/ \leq \frac{1}{2}$ such that $\mathrm{diam}(t') \leq \cNG/ \cdot V$, $p_i(t'') = p_i(t')$ for all robots $r_i$ and all rounds $t'' \geq t'$ and $p_i(t) \neq p_j(t)$ for all robots $r_i$ and $r_j$ and rounds $t$.
Moreover, all robots terminate simultaneously, i.e., know in the same round or within one epoch that $\globalDiameter \leq c_{\mathrm{ng}}$.

\section{A Class of Gathering Protocols} \label{section:alphaBetaContractingStrategies}

In this section, we describe the class of \lambdaContracting/ (gathering) protocols – a class of protocols which solve \gathering/ in $\Theta \left(\Delta^2\right)$ rounds and serves as a basis for collisionless protocols to solve \nearGathering/ (see \Cref{section:collisionlessProtocols}).
Moreover, we derive a subclass of  \lambdaContracting/ protocols, called \alphaBetaContracting/ protocols.
The class of \alphaBetaContracting/ contracting protocols is a powerful tool to determine whether a given gathering protocol (such as \gtcShort/) fulfills the property of being \lambdaContracting/.

The first intuition to define a class of protocols to solve \gathering/ would be to transfer the class of continuous contracting protocols (cf. \Cref{section:relatedWork}) to the discrete \LCM/ case.
A continuous robot formation protocol is called \emph{contracting} if robots that are part of the global convex hull move with speed $V$ towards the inside of the global convex hull.
A translation to the discrete (\LCM/) case might be to demand that each robot moves a constant distance inwards (away from the boundary) of the global convex hull, cf.\ \Cref{figure:idealProtocol}.

\begin{figure}[htbp]
	\begin{minipage}[t]{0.48\textwidth}
		\centering
		\includegraphics[width = 0.85\textwidth]{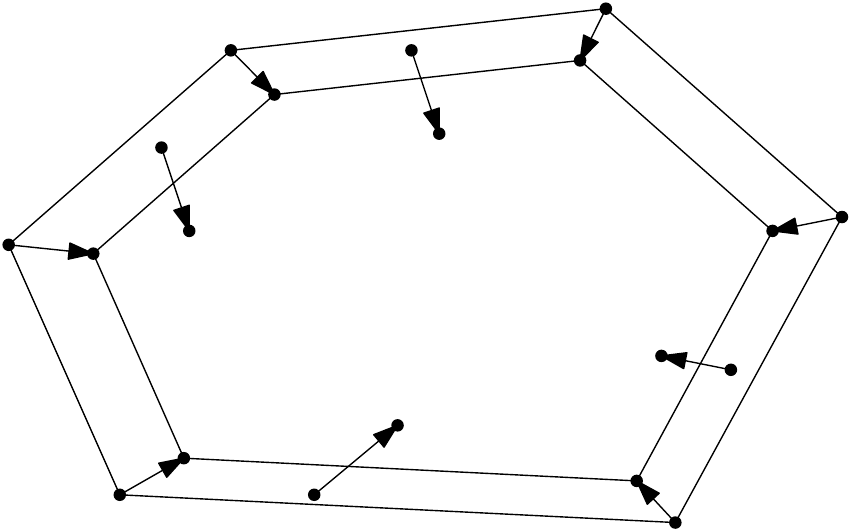}
		\caption{Ideally, every robot that is close to the boundary of the global convex hull (the surrounding convex polygon), would move a constant distance inwards.}
		\label{figure:idealProtocol}
	\end{minipage}
	\,
	\begin{minipage}[t]{0.48\textwidth}
		\centering
		\includegraphics[width =\textwidth]{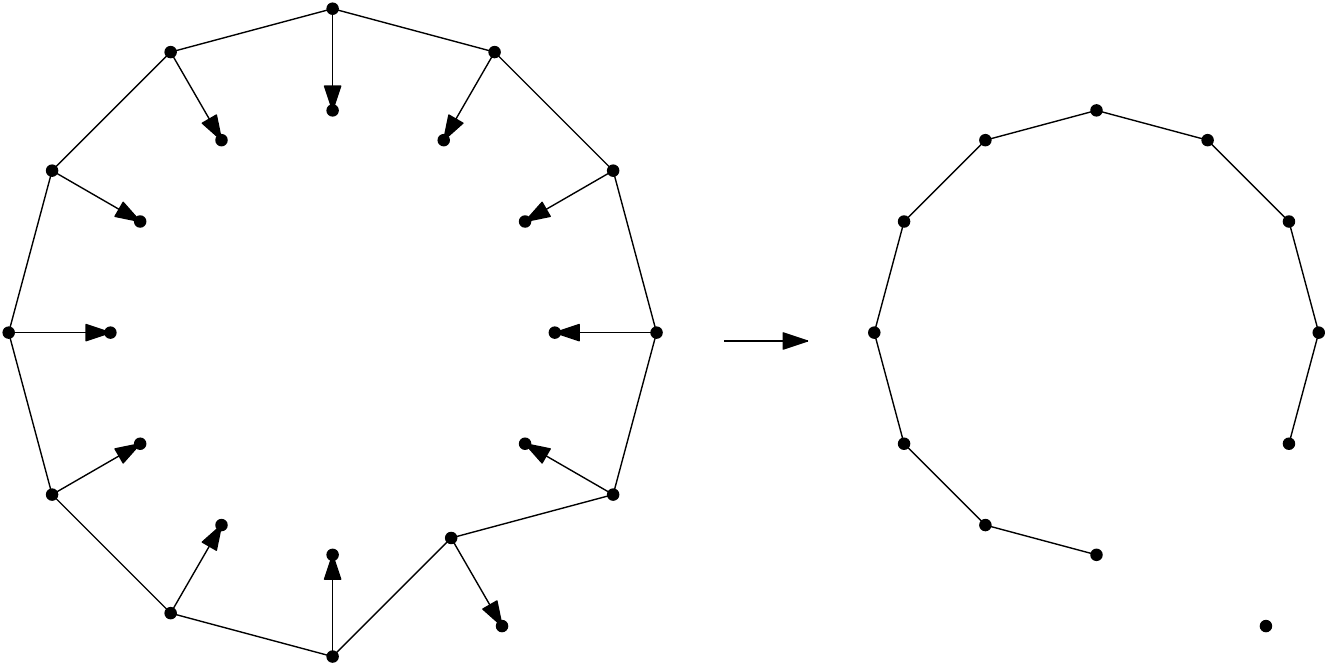}
		\caption{Visualization of the example to emphasize that continuous protocols cannot be directly translated to the \LCM/ case.}
		\label{figure:problemsOfIdealProtocols}
	\end{minipage}
\end{figure}
However, such a protocol cannot exist in the discrete \LCM/ setting.
Consider $n$ robots positioned on the vertices of a regular polygon with side length $V$.
Now take one robot and mirror its position along the line segment connecting its two neighbors (cf. \Cref{figure:problemsOfIdealProtocols}).
Next, assume that all robots would move a constant distance along the angle bisector between their direct neighbors in the given gathering protocol.
Other movements would lead to the same effect since the robots are disoriented.
In the given configuration, $n-1$ robots would move a constant distance inside the global convex hull while one robot even leaves the global convex hull.
Not only that the global convex hull does not decrease as desired, but also the connectivity of  \vubg{t} is not maintained as the robot moving outside loses connectivity to its direct neighbors.
Consequently, discrete gathering protocols have to move the robots more carefully to maintain the connectivity of \vubg{t} and to prevent disadvantageous movements caused by the disorientation of the robots.

\subsection{\lambdaContracting/ Protocols}%
\label{sec:gathering:class_definition}

Initially, we emphasize two core features of the protocols.
A discrete protocol is \emph{connectivity preserving}, if it always maintains connectivity of \vubg{t}.
Due to the limited visibility and disorientation, every protocol to solve \gathering/ and \nearGathering/ must be connectivity preserving since it is deterministically impossible to reconnect lost robots to the remaining swarm.
Moreover, we study protocols that are \emph{invariant}, i.e., the movement of a robot does not change no matter how its local coordinate system is oriented.
This is a natural assumption since the robots have variable disorientation and thus cannot rely on their local coordinate system to synchronize their movement with nearby robots.
Moreover, many known protocols under the given robot capabilities are invariant, e.g., \cite{DBLP:journals/trob/AndoOSY99,DBLP:conf/sirocco/BraunCH20,DBLP:journals/siamco/LinMA07,DBLP:journals/siamco/LinMA07a}.

\begin{definition}
	\label{def:alpha-centered}
	Let $Q$ be a convex polytope with diameter $diam$ and $0 < \lambda \leq 1$ a constant.
	A point $p \in Q$ is called to be \underline{\lambdaCentered/} if it is the midpoint of a line segment that is completely contained in $Q$ and has a length of $\lambda \cdot diam$.
\end{definition}

\begin{definition}
	\label{def:lambda-contracting}
	A connectivity preserving and invariant discrete robot formation protocol $\mathcal{P}$ is called \underline{\lambdaContracting/} if \fp{i} is a \lambdaCentered/ point of \chull{i} for every robot $r_i$ and every $t \in \mathbb{N}_0$.
\end{definition}

Observe that \Cref{def:lambda-contracting} does not necessarily enforce a final gathering of the protocols.
Consider, for instance, two robots.
A protocol that demands the two robots to move halfway towards the midpoint between themselves would be $\nicefrac{1}{4}$-contracting, but the robots would only \emph{converge} towards the same position.
The robots must be guaranteed to compute the same target point eventually to obtain a final gathering.
We demand this by requiring that there is a constant $c < 1$, such that $N_i(t) = N_j(t)$ and $\robotDiameter{i} = \robotDiameter{j} <= c$ implies that the robots compute the same target point.
Protocols that have this property are called \emph{collapsing}.
Observe that being collapsing is reasonable since \lambdaContracting/ demands that robots compute target points inside their local convex hulls and hence, the robots' local diameters are monotonically increasing in case no further robot enters their neighborhood.
%Whenever $N_i(t) = N_j(t) = N_i(t+1) = N_j(t+1)$ (no further robot becomes visible), then $\mathrm{diam}_i(t+1) = \mathrm{diam}_j(t+1) \leq \robotDiameter{i} = \robotDiameter{j}$, i.e., the local diameters of robots with the same neighborhood are monotonically decreasing.
Hence, demanding a threshold to enforce moving to the same point is necessary to ensure a final gathering.
For the ease of description, we fix $c = \nicefrac{1}{2}$ in this work.
However, $c$ could be chosen as an arbitrary constant by scaling the obtained runtime bounds with a factor of $\nicefrac{1}{c}$.

\begin{definition}
	\label{def:lambda-gathering}
	A discrete robot formation protocol $\mathcal{P}$ is called a \underline{\lambdaGathering/} if $\mathcal{P}$ is \lambdaContracting/ and collapsing.\footnote{Being \lambdaContracting/, connectivity preserving and collapsing would be sufficient to ensure \gathering/. However, these protocols are a subroutine for \nearGathering/ where robots must be able to compute nearby robot target points, which can only be done if the target points are invariant. For ease of description, we define the general protocols as invariant.}
\end{definition}

Two examples of all possible target points of \lambdaGathering/s in $\mathbb{R}^2$ are depicted in \Cref{figure:lambdaCenteredPoints} (contained in \Cref{sec:contribution}).

\subsection{Analysis of \lambdaContracting/ Protocols} \label{section:alphaBetaProtocolsUpperBound}

In the following, we state upper and lower bounds about \lambdaContracting/ protocols.
We start with a lower bound that is especially valid for \lambdaGathering/s.
The lower bound holds for all discrete gathering protocols that compute robot target points always inside local convex hulls.

\begin{restatable}{theorem}{TheoremLowerBoundProtocols}\label{theorem:LowerBoundAlphaBetaContracting}
	There exists an initial configuration such that every discrete gathering protocol $\calP$ that ensures $\fp{i} \in \chull{i}$ for all robots $r_i$ and all rounds $t \in \mathbb{N}_0$, requires $\Omega\left(\Delta^2\right)$ rounds to gather $n$ robots.
\end{restatable}

Next, we state a matching upper bound for \lambdaContracting/ protocols in two dimensions.
We first focus on robots in the Euclidean plane to make the core ideas visualizable.

\begin{restatable}{theorem}{ClassUpperBound} \label{theorem:upperBoundAlphaBetaContracting}
	Consider a swarm of robots in $\mathbb{R}^2$.
	Every \lambdaGathering/  gathers all robots in $ \frac{171 \cdot \pi \cdot \Delta^2}{\lambda^3} + 1\in \mathcal{O}\left(\Delta^2\right)$ rounds.
\end{restatable}

\medskip\noindent\textbf{\textsf{High-Level Description.\;}}
The proof aims to show that the radius of the global smallest enclosing circle (SEC), i.e., the SEC that encloses all robots' positions in a global coordinate system, decreases by $\Omega \left(\nicefrac{1}{\Delta}\right)$ every two rounds.
Since the initial radius is upper bounded by $\Delta$, the runtime of $\calO\left(\Delta^2\right)$ follows.
See \Cref{figure:intuition} for a visualization.
\begin{figure}[htb]
	\centering
	\includegraphics[width =0.95\textwidth]{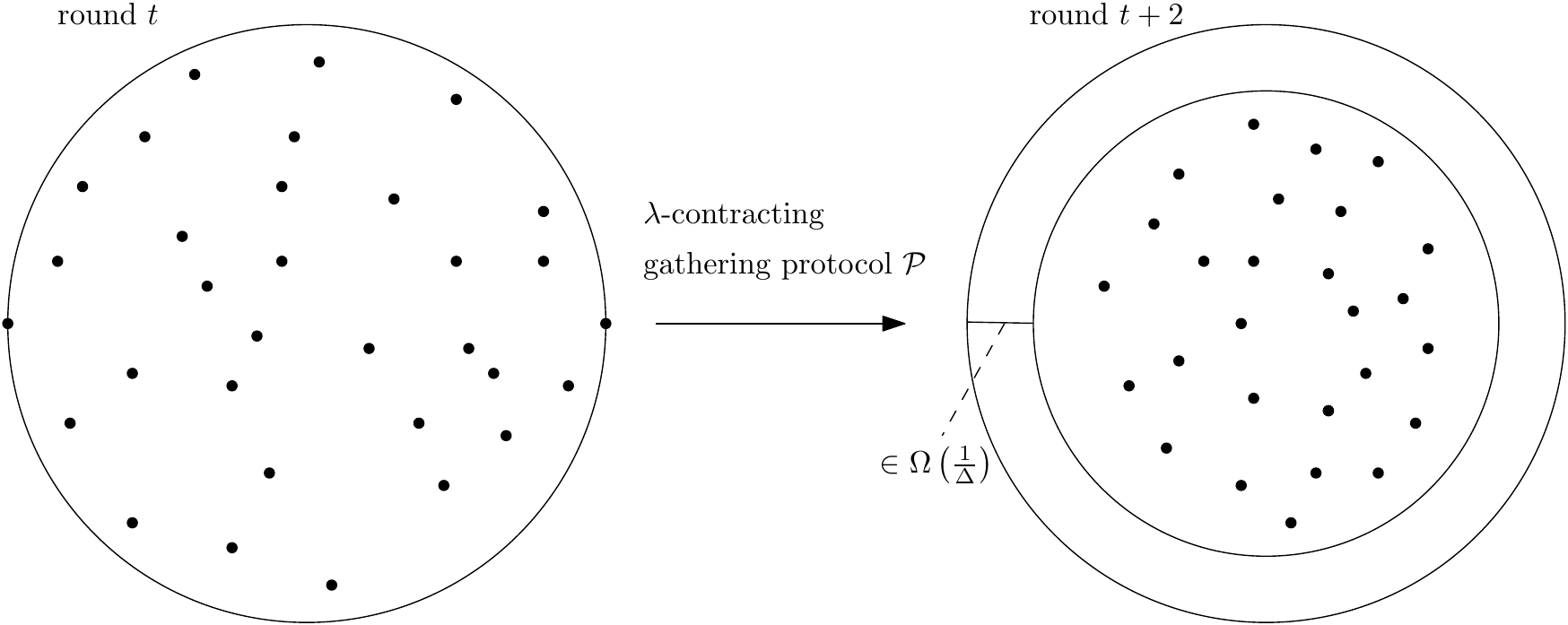}
	\caption{We show that the radius of the global SEC decreases by $\Omega \left(\nicefrac{1}{\Delta}\right)$ every two rounds.}
	\label{figure:intuition}
\end{figure}

We consider the fixed circular segment \baseSegmentLambda/ of the global SEC and analyze how the inside robots behave.
A circular segment is a region of a circle \enquote{cut off} by a chord.
The segment \baseSegmentLambda/ has a chord length of at most $\nicefrac{\lambda}{4}$ (for a formal definition, see below) and we can prove a height of \baseSegmentLambda/  in the order of $\Omega \left(\nicefrac{1}{\Delta}\right)$ (\Cref{lemma:height}).
Observe that in any circular segment, the chord's endpoints are the points that have a maximum distance within the segment, and hence, the maximum distance between any pair of points in \baseSegmentLambda/ is at most $\nicefrac{\lambda}{4}$.
Now, we split the robots inside of \baseSegmentLambda/ into two classes: the robots $r_i$ with $\robotDiameter{i} > \nicefrac{1}{4}$ and the others with $\robotDiameter{i} \leq \nicefrac{1}{4}$.
Recall that every robot $r_i$ moves to the \lambdaCentered/ point \fp{i}.
Moreover \fp{i} is the midpoint of a line segment $\ell $ of length $\lambda \cdot \robotDiameter{i}$ that is completely contained in the local convex hull of $r_i$.
For robots with $\robotDiameter{i} > \nicefrac{1}{4}$ we have that $\ell$ is larger than $\nicefrac{\lambda}{4}$ and thus, $\ell$ cannot be completely contained in \baseSegmentLambda/.
Hence, $\ell$ either connects two points outside of \baseSegmentLambda/ or one point inside and another outside.
In the former case, \fp{i} is outside of \baseSegmentLambda/, in the latter case \fp{i} is outside of  a segment with half the height $h$ of \baseSegmentLambda/.
See \Cref{lemma:largeDiameter} for a formal statement of the first case.

	It remains to argue about robots with $\robotDiameter{i} < \nicefrac{1}{4}$.
	Here, we consider a segment with even smaller height, namely $h \cdot \nicefrac{\lambda}{4}$.
	We will see that all robots which compute a target point inside this segment (which can only be robots with $\robotDiameter{i} < \nicefrac{\lambda}{4}$) will move exactly to the same position.
	Hence, in round $t+1$ there is only one position in the segment with height $h \cdot \nicefrac{\lambda}{4}$ occupied by robots.
	All other robots are located outside of the segment with height $\nicefrac{h}{2}$.
	As a consequence, for all robots $r_i$ in the segment with height $h \cdot \nicefrac{\lambda}{4}$, it must hold \fp{i} is outside of  the segment with height $h \cdot \nicefrac{\lambda}{4}$.
	See \Cref{lemma:smallDiameters} for a formal statement.
	Finally, \Cref{lemma:globalRadiusDecrease} combines the previous statements and gives a lower bound on how much the radius of the global SEC decreases.

	\medskip\noindent\textbf{\textsf{Detailed Analysis.\;}}
	First, we introduce some definitions.
	Let $N := N(t)$ be the (global) smallest enclosing circle of all robots in round $t$ and $R := R(t)$ its radius.
	Now, fix any point $b$ on the boundary of $N$.
	The two points in distance $\nicefrac{\lambda}{8}$ of $b$ on the boundary of $N$ determine the circular segment \baseSegmentLambda/ with height $h$.
	In the following, we determine by $S_{\lambda}(c)$ for $0 < c \leq 1$ the circular segment with height $c \cdot h$ that is contained in \baseSegmentLambda/.
	See \Cref{figure:globalSegments} for a depiction of the segment \baseSegmentLambda/ and the segment \betaHalfSegmentLambda/ (that is used in the proofs).
	In the following, all lemmata consider robots that move according to a \lambdaGathering/ $\mathcal{P}$.

	\begin{figure}[htb]
		\centering
		\includegraphics[width = 1.05\textwidth]{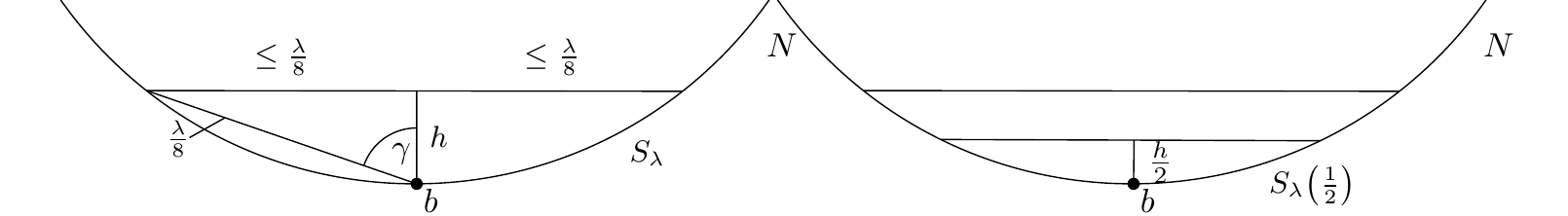}
		\caption{The segments \baseSegmentLambda/ (to the left) and \betaHalfSegmentLambda/ of the global SEC $N$ is depicted.}
		\label{figure:globalSegments}
	\end{figure}

	%	\begin{figure}[htb]
		%	\centering
		%	\includegraphics[width = 0.75\textwidth]{figures/segmentsSmaller.pdf}
		%	\caption{The segment \betaHalfSegmentLambda/ that is contained in \baseSegmentLambda/ is depicted.}
		%	\label{figure:segmentsSmaller}
		%\end{figure}

		In the following, we prove that all robots leave the segment \mainSegmentLambda/ every two rounds.
		As a consequence, the radius of $N$ decreases by at least $\nicefrac{\lambda}{4} \cdot h$.
		Initially, we give a bound on $h$.
		We use Jung's Theorem (\Cref{theorem:jungsTheorem}) to obtain a bound on $R$ and also on $h$.

		\begin{restatable}{lemma}{LemmaHeight}
			\label{lemma:height}
			$h \geq \frac{\sqrt{3} \cdot \lambda^2}{64  \pi  \Delta}$.
		\end{restatable}

		We continue to prove that all robots leave \mainSegmentLambda/ every two rounds.
		First of all, we analyze robots for which $\robotDiameter{i} > \nicefrac{1}{4}$.
		These robots even leave the larger segment \betaHalfSegmentLambda/.

		\begin{restatable}{lemma}{LemmaLargeDiameter} \label{lemma:largeDiameter}
			For any robot $r_i$ with $\robotDiameter{i} > \nicefrac{1}{4}:$ $\fp{i} \in N \setminus \betaHalfSegmentLambda/$.
		\end{restatable}

		Now, we consider the case of a single robot in \mainSegmentLambda/, and its neighbors are located outside of \betaHalfSegmentLambda/.
		We prove that this robot leaves \mainSegmentLambda/.
		Additionally, we prove that none of the robots outside of \betaHalfSegmentLambda/ that see the single robot in \mainSegmentLambda/ enters \mainSegmentLambda/.

		\begin{restatable}{lemma}{LemmaSmallDiameter} \label{lemma:smallDiameters}
			Consider a robot $r_i$ located in \mainSegmentLambda/.
			If all its neighbors are located outside of $\betaHalfSegmentLambda/$, $\fp{i} \in N \setminus \mainSegmentLambda/$.
			Similarly, for a robot $r_i$ that is located outside of \betaHalfSegmentLambda/ and that has only one neighbor located in \mainSegmentLambda/, $\fp{i} \in N \setminus \mainSegmentLambda/$.
		\end{restatable}

		Next, we derive with help of \Cref{lemma:largeDiameter,lemma:smallDiameters} that \mainSegmentLambda/ is empty after two rounds.
		Additionally, we analyze how much $R(t)$ decreases.

		\begin{restatable}{lemma}{LemmaGlobalRadius} \label{lemma:globalRadiusDecrease}
			For any round $t$ with $\globalDiameter{} \geq \nicefrac{1}{2}$, $R(t+2) \leq R(t) - \frac{\lambda^3  \cdot \sqrt{3}}{256 \cdot \pi \cdot \Delta}$.
		\end{restatable}

		Finally, we can conclude with help of \Cref{lemma:globalRadiusDecrease} the main \Cref{theorem:upperBoundAlphaBetaContracting}.

		\medskip\noindent\textbf{\textsf{Upper Bound in $d$-dimensions.\;}}
		The upper bound we derived for two dimensions can also be generalized to every dimension $d$.
		Only the constants in the runtime increase slightly.

		\begin{restatable}{theorem}{ClassUpperBoundHighDim} \label{theorem:upperBoundAlphaBetaContractingHighDim}
			Consider a team of $n$ robots located in $\mathbb{R}^{d}$.
			Every \lambdaContracting/ gathering protocol gathers all robots in $\frac{256 \cdot \pi \cdot \Delta^2}{\lambda^3 } +1 \in \mathcal{O}\left(\Delta^2\right)$ rounds.
		\end{restatable}

		\subsection{Examples of \lambdaContracting/ Gathering Protocols} \label{section:exampleProtocols}

	    Next, we present examples of  \lambdaGathering/s.
		Before introducing the concrete protocols, we describe an important subclass of \lambdaContracting/ protocols, denoted as \alphaBetaContracting/ protocols, a powerful tool to decide whether a given protocol is \lambdaContracting/.
		Afterward, we introduce the known protocol  \gtcShort/ \cite{DBLP:journals/trob/AndoOSY99} and prove it to be \lambdaContracting/.
		Additionally, we introduce two further two-dimensional \lambdaGathering/s: \gtmdShort/ and \gtcdmbShort/.

		\medskip\noindent\textbf{\textsf{\alphaBetaContracting/ Protocols.\;}}
		While the definition of \lambdaGathering/s describes the core properties of efficient protocols to solve \gathering/, it might be practically challenging to determine whether a given protocol is \lambdaContracting/.
		Concrete protocols often are designed as follows: robots compute a \emph{desired} target point and move as close as possible towards it without losing connectivity \cite{DBLP:journals/trob/AndoOSY99,DBLP:conf/sirocco/BraunCH20,DBLP:journals/siamco/LinMA07}.
		The \gtcShort/ algorithm, for instance, uses this rule.
		Since the robots do not necessarily reach the desired target point, it is hard to determine whether the resulting point is \lambdaCentered/.
		Therefore, we introduce a two-stage definition: \alphaBetaContracting/ protocols .
		The parameter $\alpha$ represents an \alphaCentered/ point (\Cref{def:alpha-centered}) and $\beta$ describes how close the robots move towards the point.
		%As an intermediate step, we define the $\beta$-scaled convex hull around an arbitrary point.

		\begin{definition}
			\label{def:beta-scaled-polygon}
			Let $c_1, \dots, c_k$ with $c_i \in \mathbb{R}^d$ be the vertices of a convex polytope $Q$, $p \in Q$ and $ 0 < \beta \leq 1$ a constant.
			\underline{$Q\left(p, \beta\right)$} is the convex polytope with vertices $p + (1-\beta) \cdot \left(c_i-p\right)$.
		\end{definition}

		Now, we are ready to define the class of \alphaBetaContracting/ protocols.
		It uses a combination of \Cref{def:alpha-centered,def:beta-scaled-polygon}: the target points of the robots must be inside of the $\beta$-scaled local convex hull around an $\alpha$-centered point.
		See also \Cref{figure:scaledHulls} for a visualization of valid target points in \alphaBetaContracting/ protocols.
		Recall that \chull{i} defines the convex hull of all neighbors of $r_i$ including $r_i$ in round $t$ and \chull{i}$\left(p,\beta\right)$ is the scaled convex hull around $p$ (\Cref{def:beta-scaled-polygon}).

		\begin{definition}
			\label{def:alpha-beta-contracting}
			A connectivity preserving and invariant discrete robot formation protocol $\mathcal{P}$ is called to be  \underline{\alphaBetaContracting/}, if there exists an \alphaCentered/ point \alphaCenterP{i} s.t.\ $\fp{i} \in \chull{i}\bigl(\alphaCenterP{i},\beta\bigr)$ for every robot $r_i$ and every $t \in \mathbb{N}_0$.
			Moreover, $\mathcal{P}$ is called an \underline{\alphaBetaContracting/ gathering protocol} if $\mathcal{P}$ is \alphaBetaContracting/ and collapsing.
		\end{definition}

		\begin{figure}[htb]
			\centering
			\includegraphics[width =\textwidth]{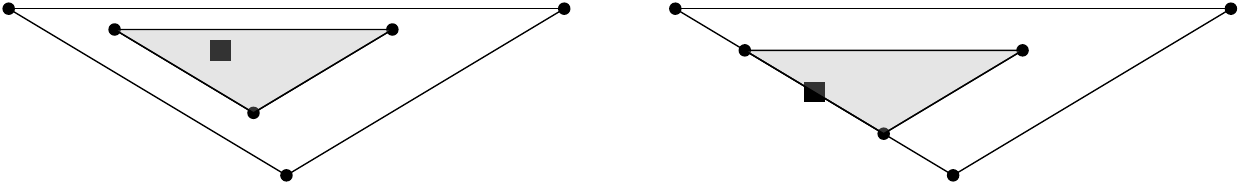}
			\caption{Two examples of valid target points of \alphaBetaContracting/ protocols. The small gray triangle represents the $\frac{1}{2}$-scaled convex hull around an $\frac{1}{4}$-centered point marked with a square.}
			\label{figure:scaledHulls}
		\end{figure}

		The following theorem describes the relation between \alphaBetaContracting/ and \lambdaContracting/ protocols.

		\begin{restatable}{theorem}{LemmaAlphaBetaIsLambda}
			\label{lem:ab-contracting-is-lambda-contracting}
			Every \alphaBetaContracting/ protocol $\calP$ is \lambdaContracting/ with $\lambda = \alpha \cdot \beta$.
		\end{restatable}

		\medskip\noindent\textbf{\textsf{\gtc/.\;}}
		As a first example, we study the two-dimensional \gtcShort/ algorithm \cite{DBLP:journals/trob/AndoOSY99}.
		It is already known that it gathers all robots in $\calO\left(n + \Delta^2\right)$ rounds \cite{DBLP:conf/spaa/DegenerKLHPW11}.
		We show that \gtcShort/ is \alphaBetaContracting/ (hence also \lambdaContracting/) and thus, obtain an improved upper runtime bound of $\calO\left(\Delta^2\right)$.
		\Cref{algorithm:gtc2D} contains the formal description of the \gtcShort/ algorithm.
		Robots always move towards the center of the smallest enclosing circle of their neighborhood.
		To maintain connectivity, \emph{limit circles} are used.
		Each robot $r_i$ always stays within the circle of radius $\nicefrac{1}{2}$ centered in the midpoint $m_j$ of every visible robot $r_j$.
		Since each robot $r_j$ does the same, it is ensured that two visible robots always stay within a circle of radius  $\nicefrac{1}{2}$ and thus, they remain connected.
		Consequently, robots move only that far towards the center of the smallest enclosing circle such that no limit circle is left.

		\begin{restatable}{theorem}{TheoremGtcAlphaBeta} \label{theorem:gtcAlphaBeta}
			\gtcShort/ is $\bigl(\nicefrac{\sqrt{3}}{8}, \nicefrac{1}{2}\bigr)$-contracting.
		\end{restatable}

		\gtcShort/ can be generalized to $d$-dimensions by moving robots towards the center of the smallest enclosing hypersphere of their neighborhood.
		We denote the resulting protocol by  $d$-\gtcShort/, a complete description is deferred to \Cref{section:appendixDGtc}.

		\begin{restatable}{theorem}{TheoremDGtc}\label{theorem:dGtC}
			$d$-\gtcShort/ is $\bigl(\nicefrac{\sqrt{2}}{8}, \nicefrac{1}{2}\bigr)$-contracting.
		\end{restatable}

		\medskip\noindent\textbf{\textsf{\gtmd/ (\gtmdShort/)\;}}
		Next, we describe a second two-dimensional protocol that is also \alphaBetaContracting/.
		The intuition is quite simple: a robot $r_i$ moves towards the midpoint of the two robots defining $\robotDiameter{i}$.
		Similar to the \gtcShort/ algorithm, connectivity is maintained with the help of limit circles.
		A robot only moves that far towards the midpoint of the diameter such that no limit circle (a circle with radius $\nicefrac{1}{2}$ around the midpoint of $r_i$ and each visible robot $r_j$) is left.
		Observe further that the midpoint of the diameter is not necessarily unique.
		To make \gtmdShort/ in cases where the midpoint of the diameter is not unique deterministic, robots move according to \gtcShort/.
		The formal description can be found in \Cref{algorithm:gtcmd}.
		We prove the following property about \gtmdShort/.
		\begin{restatable}{theorem}{TheoremGtmd} \label{theorem:gtmdContracting}
			In rounds, where the local diameter of all robots is unique, \gtmdShort/ is $\left(1,\nicefrac{1}{10}\right)$-contracting ($(\nicefrac{\sqrt{3}}{8}, \nicefrac{1}{2})$-contracting otherwise).
		\end{restatable}

		\medskip\noindent\textbf{\textsf{\gtcdmb/ (\gtcdmbShort/)\;}}
		Lastly, we derive a third algorithm for robots in $\mathbb{R}^2$ that is also \alphaBetaContracting/.
		It is based on the local \emph{diameter minbox} defined as follows.
		The local coordinate system is adjusted such that the two robots that define the diameter are located on the $y$-axis, and the midpoint of the diameter coincides with the origin.
		Afterwards, the maximal and minimal $x$-coordinates $x_{\mathrm{max}}$ and $x_{\mathrm{min}}$ of other visible robots are determined.
		Finally, the robot moves towards $\left(\nicefrac{1}{2} \cdot \left(x_{\mathrm{min}} + x_{\mathrm{max}}\right),0\right)$.
		The box boundaries with $x$-coordinates $x_{\mathrm{min}}, x_{\mathrm{max}}$ and $y$-coordinates $-\nicefrac{\robotDiameter{i}}{2}$ and $\nicefrac{\robotDiameter{i}}{2}$ is called the \emph{diameter minbox} of $r_i$.
		Note that, similar to \gtmdShort/, the diameter minbox of $r_i$ might not be unique.
		In this case, a fallback to \gtcShort/ is used.
		The complete description of \gtcdmbShort/ is contained in \Cref{algorithm:gtcdmb}.
		Also \gtcdmbShort/ is \alphaBetaContracting/.

		\begin{restatable}{theorem}{TheoremGtcdmb}\label{theorem:gtcdmbAlphaBeta}
			In rounds, where the local diameter of all robots is unique, \gtcdmbShort/ is $\bigl(\nicefrac{\sqrt{3}}{8}, \nicefrac{1}{10}\bigr)$-contracting ($(\nicefrac{\sqrt{3}}{8}, \nicefrac{1}{2})$-contracting otherwise).
		\end{restatable}

\section{Collisionless Near-Gathering Protocols} \label{section:collisionlessProtocols}

In this section, we study the \nearGathering/ problem for robots located in $\mathbb{R}^{d}$ under the  \ssync\ scheduler.
The main difference to \gathering/ is that robots never may collide (move to the same position).
We introduce a very general approach to \nearGathering/ that builds upon \lambdaGathering/s (\Cref{section:alphaBetaContractingStrategies}).
%During the execution of a \lambdaContracting/ protocol, however, two robots may collide.
%For \lambdaContracting/ gathering protocols, a collision is even required in the final configuration.
We show how to transform \emph{any} \lambdaGathering/ into a \emph{collisionless} \lambdaContracting/ protocol that solves \nearGathering/ in $\mathcal{O}\left(\Delta^2\right)$ epochs under the \ssync{} scheduler.
The only difference in the robot model (compared to \gathering/ in \Cref{section:alphaBetaContractingStrategies}) is that we need a slightly stronger assumption on the connectivity: the connectivity range must be by an additive constant smaller than the viewing range.
More formally, the connectivity range is $V$ while robots have a viewing range of $V + \tau$ for a constant $0 < \tau \leq \nicefrac{2}{3}V$.
Note that the upper bound on $\tau$ is only required because $\nicefrac{\tau}{2}$ also represents the maximum movement distance of a robot (see below).
A larger movement distance usually does not maintain the connectivity.
In general, the viewing range could also be chosen larger than $V+\tau$ without any drawbacks while keeping the maximum movement distance at $\nicefrac{\tau}{2}$.

%\begin{definition}[\clLambdaContracting/]
%	A discrete robot formation protocol $\calP$ is called \emph{\clLambdaContracting/} if it is \lambdaContracting/ (see \cref{def:alpha-beta-contracting}) and collisionless (see \cref{def:collision})
%\end{definition}

The main idea of our approach can be summarized as follows: first, robots compute a \emph{potential} target point based on a \lambdaGathering/ $\calP$ that considers only robots at a distance at most $V$.
Afterward, a robot $r_i$ uses the viewing range of $V+\tau$ to determine whether the own potential target point collides with any potential target point of a nearby neighbor.
If there might be a collision, $r_i$ does not move to its potential target point.
Instead, it only moves to a point between itself and the potential target point where no other robot moves to.
At the same time, it is also ensured that $r_i$ moves sufficiently far towards the potential target point to maintain the time bound of $\mathcal{O}\left(\Delta^2\right)$ epochs.
To realize the ideas with a viewing range of $V+\tau$, we restrict the maximum movement distance of any robot to $\nicefrac{\tau}{2}$.
More precisely, if the potential target point of any robot given by $\calP$ is at a distance of more than $\frac{\tau}{2}$, the robot moves at most $\frac{\tau}{2}$ towards it.
With this restriction, each robot could only collide with other robots at a distance of at most $\tau$.
The viewing range of $V+\tau$ allows computing the potential target point based on $\calP$ of all neighbors at a distance at most $\tau$.
By knowing all these potential target points, the own target point of the collision-free protocol can be chosen.
While this only summarizes the key ideas, we give a more technical intuition and a summary of the proof in \Cref{section:collisionlessIntuitionShort}.

\begin{theorem}
	\label{thm:collisionless-class}
	%	For every \alphaBetaGathering/ $\calP$ there exist a collisionless \alphaBetaContracting/ protocol $\pCL$ which solves \nearGathering/ in $\calO(\Delta^{2})$ rounds when the initial strong distance graph is connected.

	For every \lambdaGathering/ $\calP$, there exists a collisionless \lambdaContracting/ protocol $\pCL$ which solves \nearGathering/ in $\calO(\Delta^{2})$ epochs under the \ssync\ scheduler.
	Let $V$ be the viewing and connectivity range of $\calP$.
	$\pCL$ has a connectivity range $V$ and viewing range $V+\tau$ for a constant $0 < \tau \leq \nicefrac{2}{3}V$.

	%
	%	Let $2/3 \cdot V \geq \tau > 0$ and $0.5 \geq \varepsilon > 0$.
	%	For every \alphaBetaGathering/ $\calP$ there exist $\pCL$ with the following properties.
	%	\begin{itemize}
		%		\item $\pCL$ is a \clAlphaPrimeBetaPrimeContracting/ protocol with $\alpha' = \alphaPrimeValue$ and $\beta' = \betaPrimeValue$.
		%		\item Let $V$ be the viewing range of $\calP$. $\pCL$ has a viewing range of $V + \tau$.
		%		\item $\pCL$ results in a near-gathering with diameter $\leq \tau$ of all robots in $\runningTimeLemmaCLPrime \in \mathcal{O}(\Delta^2)$ rounds, if the unit disc graph with radius $V$ is connected in the initial configuration.\todo{introduce connectivity range in model section?}
		%	\end{itemize}

\end{theorem}

\subsection{Collisionless Protocol} \label{section:protocolPCl}

The construction of the collisionless protocol $\pCLtauEpsilon$ depends on several parameters that we briefly define.
$\calP$ is a \lambdaGathering/ (designed for robots with a viewing range of $V$).
The constant $\tau$ has two purposes.
The robots have a viewing range of $V+\tau$ ($ 0 < \tau \leq \nicefrac{2}{3}$V) and $\nicefrac{\tau}{2}$ is the maximum movement distance of any robot.
Lastly, the constant $\varepsilon \in (0,\nicefrac{1}{2})$  determines how close each robot moves towards its target point based on $\calP$.
To simplify the notation, we usually write $\pCL$ instead of  $\pCLtauEpsilon$.
Subsequently, we formally define $\pCLtauEpsilon$.
The description is split into three parts that can be found in \Cref{algorithm:collisionlessGTC,algorithm:CollisonPointsOnLine,algorithm:targetPointCollisionlessGTC}.
The main routine is contained in \Cref{algorithm:collisionlessGTC}.
The other two \Cref{algorithm:CollisonPointsOnLine,algorithm:targetPointCollisionlessGTC} are used as subroutines.

The computation of $\fpCL{i}$ is based on the movement $r_i$ would do in slightly modified version of $\calP$, denoted as $\pTau$.
The protocol \pTau{} is defined in \cref{algorithm:targetPointCollisionlessGTC} and a detailed intuition can be found in \Cref{section:collisionlessIntuitionShort}.
The position of $\fpCL{i}$ lies on the \textit{collision vector} $\ellpTau{i}$, the vector from $p_i(t)$ to $\fpTau{i}$ (the potential target point).
On $\ellpTau{i}$, there may be several \emph{collision points}.
These are either current positions or potential target points ($\fpTau{k}$) of other robots $r_k$ or
single intersection points between $\ellpTau{i}$ and another collision vector $\ellpTau{k}$.
The computation of collision points is defined in \cref{algorithm:CollisonPointsOnLine}.
Moreover, $d_i > 0$ is the minimal distance  between a collision point and $\fpTau{i}$.
The final target point $\fpCL{i}$ is exactly at distance $\targetPointDist{i}$ from $\fpTau{i}$.
\cref{fig:collisionless-algorithm-and-collision-points} gives an example of collision points and target points of  $\pCL$.

\begin{figure}[htbp]
	\includegraphics[width=\linewidth]{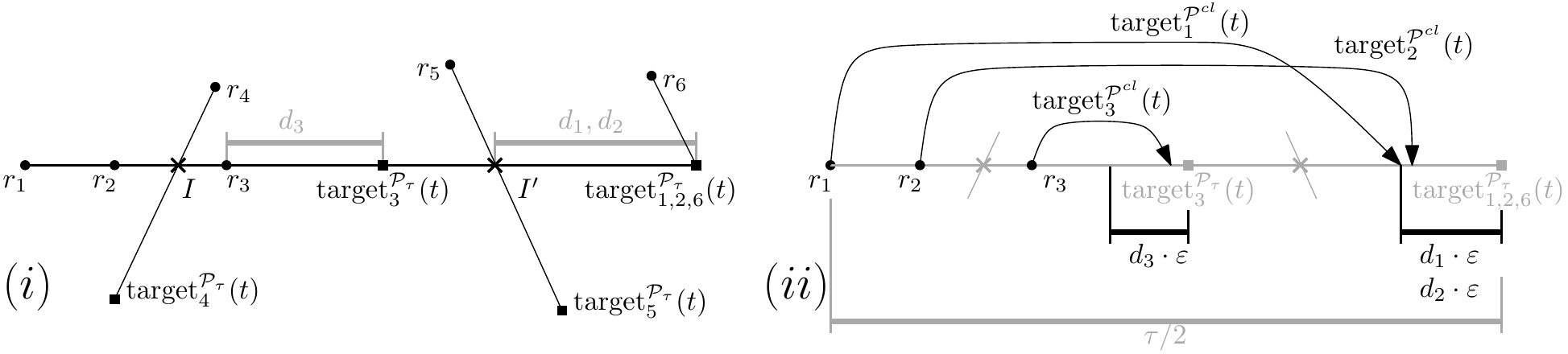}
	\caption{Example of $\fpCL{i}$ with $V = 1, \tau = 2/3$ and $\varepsilon = 0.49$. $(i)$ shows the collision points and computation of $d_1, d_2$ and $d_3$ (line \ref{line:d-i:clGTC} in \cref{algorithm:collisionlessGTC}). $(ii)$ shows the positions where $r_1, r_2$ and $r_3$ will move to in protocol $\pCL$ as returned by \cref{algorithm:collisionlessGTC}.}
	\label{fig:collisionless-algorithm-and-collision-points}
\end{figure}

%with five robots $r_1, \dots, r_5$.
%Subfigure $(i)$ shows distances $d_1, d_2$ and $d_3$ as well as all collision points on $\ellpTau{1}, \ellpTau{2}$ and $\ellpTau{3}$.
%The collision points on $\ellpTau{1}$ from left to right are $C_1 = \left\{p_1(t), p_2(t), I, p_3(t), \fpTau{3}, I', \fpTau{1} = \fpTau{2}\right\}$,\todo{umbruch} $I$ and $I'$ are the intersections between $\ellpTau{1}$ and $\ellpTau{4}$, respectively $\ellpTau{1}$ and $\ellpTau{5}$.
%$d_1$ is the distance between $\fpTau{1}$ and $I'$.
%$C_2 = \left\{p_2(t), I, p_3(t), \fpTau{3}, I', \fpTau{1} = \fpTau{2}\right\}$ and $d_2 = d_1$.
%$C_3 = \left\{p_3(t), \fpTau{3}\right\}$ and $d_3$ the distance between $p_3(t)$ and $\fpTau{3}$.
%Subfigure $(ii)$ shows the actual positions of $\fpCL{1}, \fpCL{2}$ and $\fpCL{3}$.
%They are each within a distance $d_i \cdot \epsilon$ of $\fpTau{i}$.

\begin{algorithm}[htb]
	\caption{$\fpCLtauEpsilon{i} $}
	\label{algorithm:collisionlessGTC}
	\begin{algorithmic}[1]
		%\State $P_i$ $\leftarrow$ $\fpTau{i}$ \Comment{see \cref{algorithm:targetPointCollisionlessGTC}}
		%\State $v_i$ $\leftarrow$ collision vector from $p_i(t)$ to $P_i$ \Comment{$p_i(t)$: position of $r_i$}
		\State $R_i$ $\leftarrow$ $\{r_k : |p_{k}(t) - p_i(t)| \leq \tau\}$ {\color{gray}\Comment{Robots in radius $\tau$ around $r_i$ (including $r_i$)}} \label{line:R-i:clGTC}
		\State $C_i$ $\leftarrow$ $\collisionPoints{\pTau}{i}$ {\color{gray}\Comment{Collision points on $\ellpTau{i}$, see \cref{algorithm:CollisonPointsOnLine}}} \label{line:C-i:clGTC}
		\State $d_i$ $\leftarrow$ $\min\left(\left\{\big|c - \fpTau{i}\big| : c \in C_i \setminus \{\fpTau{i}\} \right\}\right)$ {\color{gray}\Comment{min. dist. to collision point}} \label{line:d-i:clGTC}
		\State \Return point on $\ellpTau{i}$ with distance $\targetPointDist{i}$ to $\fpTau{i}$ \label{line:target-point-distance:clGTC}
	\end{algorithmic}
\end{algorithm}

\begin{algorithm}[htb]
	\caption{$\collisionPoints{\calP}{i}$}
	\label{algorithm:CollisonPointsOnLine}
	\begin{algorithmic}[1]
		\State $C_i$ $\leftarrow$ empty set
		\ForAll{$r_k \in R_i$}
		\State compute $\fp{k}$ and $\ellp{k}$ in local coordinate system of $r_i$ \label{line:all-v-in-ell:CP} %\Comment{collision vectors of $r_k$, see \cref{def:collision-vector}}
		\If{$p_k(t)  \in \ellp{k}$}
		\State add $p_{k}(t)$ to $C_i${\color{gray} \Comment{position of $r_k$}} \label{line:v-k-length-0:CP}
		\EndIf
		\If{$\fp{k} \in \ellp{i}$}
		\State add $\fp{k}$ to $C_i$  \label{line:add-fpk:CP}
		\EndIf
		\If{$\ellp{k}$ intersects $\ellp{i}$ and is not collinear to $\ellp{i}$}
		\State add intersection point between $\ellp{k}$ and $\ellp{i}$ to $C_i$
		\EndIf
		\EndFor
		\State\Return $C_i$
	\end{algorithmic}
\end{algorithm}

\begin{algorithm}[htb]
	\caption{$\fpTau{i}$}
	\label{algorithm:targetPointCollisionlessGTC}
	\begin{algorithmic}[1]
		%	\Statex \hspace*{-12pt}\textit{increase viewing range if diameter smaller $\nicefrac{\tau}{2}$}
		\State $V$ $\leftarrow$ the viewing range of protocol $\calP$
		\If{robots in range $V$ have pairwise distance $\leq \nicefrac{\tau}{2}$} \label{line:cond-1:tPclGTC}
		\State $\pVTau$ $\leftarrow$ protocol $\calP$ scaled to viewing range $V + \nicefrac{\tau}{2}$
		\State $P_i$ $\leftarrow$ $\fpVTau{i}$
		\Else
		\State $P_i$ $\leftarrow$ $\fp{i}$
		\EndIf
		%		\Statex \hspace*{-12pt}\textit{decrease movement distance to $\tau/2$ if exceeded}
		\If{distance $p_i(t)$ to $P_i > \tau/2$} \label{line:cond-2:tPclGTC}
		\State\Return point with distance $\tau/2$ to $p_i(t)$ between $p_i(t)$ and $P_i$
		\Else
		\State\Return $P_i$
		\EndIf
	\end{algorithmic}
\end{algorithm}

\subsection{Proof Summary and Intuition} \label{section:collisionlessIntuitionShort}

In the following, we give a brief overview of the proof of \cref{thm:collisionless-class}.
A more detailed proof outline where all lemmas are stated can be found in \cref{section:collisionlessIntuition} and the proofs in \Cref{section:collisionFreeAnalysis}.
For the correctness and the runtime analysis of the protocol $\pCL$, we would like to use the insights into \lambdaContracting/ protocols derived in \Cref{section:alphaBetaContractingStrategies}.
However, since the robots compute their potential target point based on a \lambdaGathering/ $\calP$ with viewing range $V$, this point is not necessarily \lambdaCentered/ concerning the viewing range of $V+\tau$.
We discuss this problem in more detail and motivate the \emph{intermediate} protocol $\pTau$ that is \lambdaContracting/ with respect to the viewing range of $V+\tau$.
Note that in $\pTau$, robots can still collide.
Afterward, we argue how to transform the intermediate protocol $\pTau$ into a collision-free protocol $\pCL$ that is still \lambdaContracting/.
Lastly, we derive a time bound for $\pCL$.
%It is \lambdaContracting/ and keeps \vubg{t} connected while having a viewing range of $V+\tau$.
%In fact, any \lambdaContracting/ protocol with this properties will contract a swarm to a diameter $\leq \tau$ in $\calO(\Delta^{2})$ epochs.

%\subsubsection{} \label{section:intuitionPTauShort}
\medskip\noindent\textbf{\textsf{The protocol $\mathbf{\pTau}$.\;}}
Recall that the main goal is to compute potential target points based on a \lambdaGathering/ $\calP$ with viewing range $V$.
%Let us ignore the collision avoidance in this section and only concentrate on the \lambdaContracting/ properties of such a protocol if applied to a scenario with a viewing range of $V+\tau$.
Unfortunately, a direct translation of the protocol loses the \lambdaContracting/ property in general.
Consider the following example which is also depicted in \Cref{fig:intuitionNotAlphaShort}.
Assume there are the robots $r_1, r_2, r_3$ and $r_4$ in one line with respective distances of $\nicefrac{1}{n}, V + \nicefrac{1}{n}$ and $V + \tau$ to $r_1$.
It can easily be seen, that the target point $\fp{1}$ (protocol $\calP$ has only a viewing range of $V$) is between $r_1$ and $r_2$.
Such a target point can never be \lambdaCentered/ with $\lambda > \nicefrac{2}{n}$ for $\pCL$ (with viewing range $V + \tau$).

Next, we argue how to transform the protocol $\calP$ with viewing range $V$ into a protocol $\pTau$ with viewing range $V+\tau$ such that $\pTau$ is \lambdaGathering/.
The example above already emphasizes the main problem: robots can have very small local diameters $\robotDiameter{i}$.
Instead of moving according to $\calP$, those robots compute a target point based on $\pVTau$, which is a  \lambdaGathering/ concerning the viewing range of $V+\nicefrac{\tau}{2}$.
Protocol $\pVTau$ is obtained by scaling $\calP$ to the larger viewing range of $V + \tau$.
More precisely, robots $r_i$ with $\robotDiameter{i} \leq \nicefrac{\tau}{2}$ compute their target points based on $\pVTau$ and all others according to $\calP$.
In addition, $\calP_{\tau}$ ensures that no robot moves more than a distance of $\nicefrac{\tau}{2}$ towards the target points computed in $\calP$ and $\pVTau$.
%This has two reasons.
The first reason is to maintain the connectivity of \vubg{t}.
While the protocol $\calP$ maintains connectivity by definition, the protocol $\pVTau$ could violate the connectivity of \vubg{t}.
Restricting the movement distance to $\nicefrac{\tau}{2}$ and upper bounding $\tau$ by $\nicefrac{2}{3} V$ resolves this issue since for all robots $r_i$ that move according to $\pVTau$, $\robotDiameter{i} \leq \nicefrac{\tau}{2}$.
%Hence, $\pTau$ has a connectivity range of $V$.
The second reason is that moving at most $\nicefrac{\tau}{2}$ ensures that collisions are only possible within a range of $\tau$.

While $\pTau$ has a viewing range of $V + \tau$, it never uses its full viewing range for computing a target point.
This is necessary for the collision avoidance such that $r_i$ can compute $\fpTau{k}$ for all robots $r_k$ in distance $\leq \tau$.
It is easy to see, that the configuration in \cref{fig:intuitionNotAlphaShort} does not violate the \lambdaContracting/ property of $\pTau$.
If $\nicefrac{1}{n} > \nicefrac{\tau}{2}$ it is trivial that $\fp{1}$ is \lambdaCentered/ in $\pTau$ with $\lambda \in \calO(\tau)$.
Else, $\fpTau{1} = \fpVTau{1}$.
$\pVTau$ also considers $r_3$ in distance $V + \nicefrac{1}{n} \leq V + \tau/2$ (note, there always exists a robot between $r_2$ and $r_4$ to ensure \vubg{t} is connected).
Hence, $\fpVTau{1}$ is $\lambda \cdot \frac{V + \nicefrac{\tau}{2}}{V + \tau}$-centered in $\pTau$.
This argument is generalized in \cref{lem:pTau-alpha-beta-gathering} to show that $\pTau$ is a \lambdaGathering/.

\medskip\noindent\textbf{\textsf{Collision Avoidance.\;}}
Next, we argue how to transform the protocol $\pTau$ into the collision-free protocol $\pCL$.
The viewing range of $V+\tau$ in $\pCL$ allows a robot $r_i$ to compute $\fpTau{k}$ (the target point in protocol $\pTau$) for all robots $r_k$ within distance at most $\tau$.
Since the maximum movement distance of a robot in $\pTau$ is $\nicefrac{\tau}{2}$, this enables $r_i$ to know movement directions of all robots $r_k$ which can collide with $r_i$.
It is easy to see, that collisions can be avoided with all robots where the collision vectors do only intersect at single points (e.g. $r_1$ and $r_4$ in \cref{fig:collisionless-algorithm-and-collision-points} can avoid a collision easily) or two robots have different target points in $\pTau$ (e.g. $r_1$ and $r_3$).
When two collision vectors overlap and the target points in $\pTau$ are the same (e.g. $r_1$ and $r_2$) our protocol $\pCL$ ensures unique positions as well.
It includes the distance to the potential target point, which naturally must be different in these cases, in the computation of the target point in $\pCL$ (line \ref{line:target-point-distance:clGTC} in \cref{algorithm:collisionlessGTC}).
Since we have designed $\pTau$ in a way that a robot $r_i$ can compute $\fpTau{k}$ for all robots $r_k$ in the distance at most $\tau$, we can execute this collision avoidance although $\pTau$ has the same viewing range as $\pCL$.

\medskip\noindent\textbf{\textsf{Time Bound.\;}}
%\subsubsection{Running Time}
%\label{section:intuitionInitialStrongDistanceShort}
Previously, we have addressed the intermediate protocol $\pTau$ that is \lambdaGathering/ with respect to the viewing range of $V+\tau$ and also keeps \vubg{t} always connected.
The same holds for $\pCL$.
Keeping \vubg{t} connected is important for the termination of a \nearGathering/ protocol.
Suppose that \vubg{t} is connected and the robots only have a viewing range of $V$.
Then, the robots can never decide if they can see all the other robots.
However, with a viewing range of $V+\tau$, it becomes possible if the swarm is brought close together ($\globalDiameter < \tau$).
For any configuration where the viewing range is $V + \tau$ and \vubg{t} is connected, we state an important observation.

\begin{restatable}{lemma}{LemmaConstantDiameterByLargerVR}
	\label{lem:constant-diameter-by-larger-vr}
	Let $\calP$ be a \lambdaContracting/ protocol with viewing range $V + \tau$ for a constant $\tau > 0$ and let \vubg{t} be connected.
	If $\globalDiameter > \tau$, then $\robotDiameter{i} > \tau$, for every robot $r_i$.
\end{restatable}

%This leads directly to another helpful insight.
Due to the  \lambdaContracting/ property, robots close to the boundary of the global smallest enclosing hypersphere (SEH) move upon activation at least $\Omega\left(\frac{\robotDiameter{i}}{\Delta}\right)$ inwards.
With $\robotDiameter{i} > \tau$, it follows that the radius of the global SEH decreases by $\Omega(\nicefrac{\tau}{\Delta})$ after each robot was active at least once (see \cref{lemma:largeDiameter-cl}).
Consequently, $\globalDiameter \leq \tau$ after $\calO(\Delta^{2})$ epochs.

\begin{restatable}{lemma}{LemmagatheringWithIncreasedVR}
	\label{lem:gathering-with-increased-vr}
	Let $\calP$ be a \lambdaContracting/ protocol with a viewing range of $V+\tau$ while \vubg{t} is always connected.
	After at most $\runningTimeLemmaCL \in \calO(\Delta^2)$ epochs executing $\calP$, $\globalDiameter \leq \tau$.
\end{restatable}

Because $\pCL$ has, regarding \lambdaContracting/, connectivity and connectivity range, the same properties as $\pTau$, this lemma can directly be applied to show the running time of \cref{thm:collisionless-class}.

\section{Conclusion \& Future Work}\label{section:conclusion}

In this work, we introduced the class of \lambdaContracting/ protocols and their collisionless extensions that solve \gathering/ and \nearGathering/ of $n$ robots located in $\mathbb{R}^{d}$ in $\Theta \left(\Delta^2\right)$ epochs.
While these results already provide several improvements over previous work, there are open questions that could be addressed by future research.
First of all, we did not aim to optimize the constants in the runtime.
Thus, the upper runtime bound of $\frac{256 \cdot \pi \cdot \Delta^2}{\lambda^3 }$ seems to be improvable.

Moreover, one major open question remains unanswered: is it possible to solve \gathering/ or \nearGathering/ of oblivious and disoriented robots with limited visibility in $\mathcal{O}\left(\Delta\right)$ rounds?
In this work, we could get a little closer to the answer:
If there is a protocol that gathers in $\mathcal{O}\left(\Delta\right)$ rounds, it must compute target points regularly outside of the convex hulls of robots' neighborhoods.
All \lambdaContracting/ protocols are slow in the configuration where the positions of the robots form a regular polygon with side length equal to the viewing range.
In \cite{DBLP:conf/algosensors/CastenowHJKH21}, it has been shown that this configuration can be gathered in time $\mathcal{O}\left(\Delta\right)$ by a protocol where each robot moves as far as possible along the angle bisector between its neighbors (leaving the local convex hull).
However, this protocol cannot perform well in general.
See \Cref{figure:alternatingStar} for the \emph{alternating star}, a configuration where this protocol is always worse compared to any protocol that computes target points inside of local convex hulls.
\Cref{figure:alternatingStar} gives a hint that every protocol that performs well for the regular polygon cannot perform equally well in the alternating star.
Thus, we conjecture that $\Omega \left(\Delta^2\right)$ is a lower bound for every protocol that considers oblivious and disoriented robots with limited visibility.

\begin{figure}[htbp]
	\begin{minipage}[t]{0.48\textwidth}
		\centering
		\includegraphics{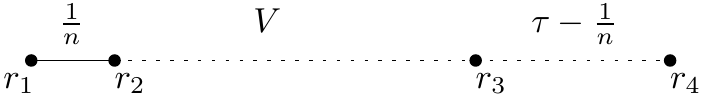}
		\caption{Example where $\fp{i}$ is not \lambdaCentered/ with respect to the viewing range $V+ \tau$.}
		\label{fig:intuitionNotAlphaShort}
	\end{minipage}
	\,
	\begin{minipage}[t]{0.48\textwidth}
		\centering
		\includegraphics[width =0.8\textwidth]{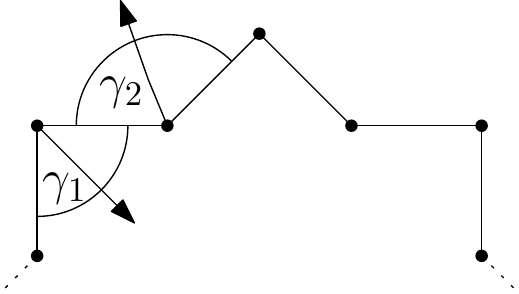}
		\caption{The robots at the angle $\gamma_1$ observe a regular square, the robots at $\gamma_2$ see a regular octagon. Given that each robot moves along the angle bisector between its neighbors and leaves its local convex hull, the radius of the global SEC decreases slower than in any \lambdaContracting/ protocol.}
		\label{figure:alternatingStar}
	\end{minipage}
\end{figure}

\newpage
\bibliography{unlimited,limited,gathering,faults}

\newpage
\appendix

\section{Proofs and Ommited Details of  \Cref{section:alphaBetaContractingStrategies} \enquote{A Class of Gathering Protocols}} \label{section:appendixSection3}

\subsection{Proof of \Cref{theorem:LowerBoundAlphaBetaContracting}}

\TheoremLowerBoundProtocols*

\begin{proof}
	In the following, we assume $n \geq 5$.
	Consider $n$ robots that are located on the vertices of a regular polygon with side length $1$.
	Observe first that due to the disorientation and because the protocols are deterministic, the local coordinate systems of the robots could be chosen such that the configuration remains a regular polygon forever (see \Cref{figure:regularPolygon} for an example).

	\begin{figure}[htb]
		\centering
		\includegraphics[width =0.35\textwidth]{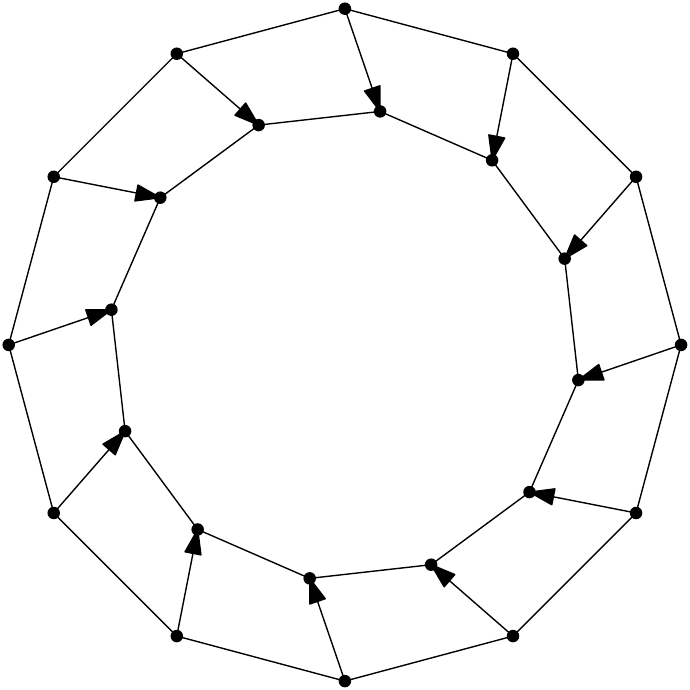}
		\caption{Initially, the robots are located on the surrounding regular polygon. The local coordinate systems of the robots can be chosen such that all robots execute the same movement in a rotated fashion such that the configuration remains a regular polygon (depicted by the inner regular polygon).}
		\label{figure:regularPolygon}
	\end{figure}

	Henceforth, we assume in the following that the robots remain on the vertices of a regular polygon.
	Let $C$ be the surrounding circle and $r_C$ its radius.
	For large $n$, the circumference $p_C$ of $C$ is $\approx n$ and $r_C \approx \frac{n}{2 \pi}$.
	Hence, $\Delta \approx \frac{n}{\pi}$.
	We show that any \lambdaContracting/ protocol (not only gathering protocols) requires $\Omega \left(\Delta^2\right)$ rounds until $p_C \leq \frac{2}{3}n$.
	As long as $p_C \geq \frac{2}{3}n$, each robot can observe exactly two neighbors at distance $\frac{2}{3} \leq s \leq 1$.

	The internal angles of a regular polygon have a size of $\gamma = \frac{\left(n-2\right) \cdot \pi }{n}$.
	Fix any robot $r_i$ and assume that $p_i(t) = \left(0,0\right)$ and the two neighbors are at $p_{i-1}(t) = \bigl(-s \cdot \sin \left(\frac{\gamma}{2}\right), s \cdot \cos \left(\frac{\gamma}{2}\right)\bigr)$ and $p_{i+1}(t) = p_{i-1}(t) = \bigl(s \cdot \sin \left(\frac{\gamma}{2}\right), s \cdot \cos \bigl(\frac{\gamma}{2}\bigr)\bigr)$.
	Now, consider the target point $\fp{i} = \bigl(x_{\fp{i}}, y_{\fp{i}} \bigr)$.
	Observe that the radius $r_C$ decreases by exactly $y_{\fp{i}}$.
	Next, we derive an upper bound on $y_{\fp{i}}$:
	$y_{\fp{i}} = s  \cdot \cos \left(\frac{\gamma}{2}\right) \leq \cos \left(\frac{\gamma}{2}\right) = \cos \left(\frac{\left(n-2\right) \cdot \pi}{2n}\right)$.

	Now, we use $\cos(x) \leq -x + \frac{\pi}{2}$ for $0 \leq x \leq \frac{\pi}{2}$.
	Hence, we obtain $\cos \left(\frac{\left(n-2\right) \cdot \pi}{2n}\right) \leq - \frac{\left(n-2\right) \cdot \pi}{2n} + \frac{\pi}{2} = - \frac{\pi}{2} + \frac{\pi}{n} + \frac{\pi}{2} = \frac{\pi}{n}$.
	Therefore, it takes at least $\frac{n^2}{3}$ rounds until $r_C$ has decreased by at least $\frac{n}{3}$.
	The same holds for the perimeter.
	All in all, it takes at least $\frac{n^2}{3} \in \Omega\left(\Delta^2\right)$ rounds until the $r_C$ decreases by at least $\frac{n}{3}$.
\end{proof}

\subsection{Proof of \Cref{theorem:upperBoundAlphaBetaContracting}}

\begin{theorem}[Jung's Theorem \cite{Jung1901,Jung1910}] \label{theorem:jungsTheorem}
	The smallest enclosing hypersphere of a point set $K \subset \mathbb{R}^{d} $ with diameter $\mathrm{diam}$ has a radius of at most $\mathrm{diam} \cdot \sqrt{\frac{d}{2 \cdot \left(d+1\right)}}$.
\end{theorem}

\LemmaHeight*

\begin{proof}
	Initially, we give an upper bound on the angle $\gamma$, see \Cref{figure:globalSegments} for its definition.
	The circumference of $N$ is $2  \pi  R$.
	We can position at most $\frac{16}{\lambda} \pi R$ points on the boundary of N that are at distance $\frac{\lambda}{8}$ from the points closest to them and form a regular convex polygon.
	The internal angle of this regular polygon is $2\gamma$.
	Hence, the sum of all internal angles is $\left(\frac{16}{\lambda} \pi  R -2 \right) \cdot \pi$.
	Thus, each individual angle has a size of at most $\frac{\left(\frac{16}{\lambda}  \pi  R-2\right) \cdot \pi }{\frac{16}{\lambda}  \pi  R} =\pi - \frac{2\pi}{\frac{16}{\lambda}  \pi  R} = \pi - \frac{\lambda}{8 R}$.
	Hence, $\gamma \leq \frac{\pi}{2} - \frac{\lambda}{16  R}$.
	Now, we are able to bound $h$.
	First of all, we derive a relation between $h$ and $\gamma$:
	$\cos \left(\gamma\right) = \frac{h}{\frac{\lambda}{8}} = \frac{8h}{\lambda} \iff h = \frac{\lambda \cdot \cos \left(\gamma\right)}{8}$.
	In the following upper bound, we make use of the fact that $\cos \left(x\right) \geq -\frac{2}{\pi}x +1 $ for $x \in [0, \frac{\pi}{2}]$.

	\begin{align*}
		h = \frac{\lambda \cdot \cos \left(\gamma\right)}{8} \geq \frac{\lambda \cdot \cos \left(\frac{\pi}{2} - \frac{\lambda}{16 R}\right)}{8}
		\geq \frac{\lambda \cdot \left(-\frac{2}{\pi} \cdot \left(\frac{\pi}{2} - \frac{\lambda}{16  R}\right) +1 \right)}{8}	= \frac{\lambda \cdot \frac{\lambda}{8  \pi  R}}{8} = \frac{\lambda^2}{64 \pi R}
	\end{align*}

	Applying \Cref{theorem:jungsTheorem} with $d=2$ yields $h \geq \frac{\sqrt{3} \cdot \lambda^2}{64  \pi  \Delta}$.

\end{proof}

\LemmaLargeDiameter*

\begin{proof}
	Since $\robotDiameter{i} >
	\nicefrac{1}{4}$ and $\mathcal{P}$ is \lambdaContracting/, \lambdaCenterP{i} is the midpoint of a line segment $\ell^\mathcal{P}_i(t)$  of length at least $\lambda \cdot \robotDiameter{i} > \nicefrac{\lambda}{4}$.
	As the maximum distance between any pair of points inside of \baseSegmentLambda/ is $\frac{\lambda}{4}$, it follows that $\ell^{i}_\mathcal{P}(t)$ either connects two points outside of \baseSegmentLambda/ or one point inside and another point outside.
	In the first case, \lambdaCenterP{i} lies outside of \baseSegmentLambda/ (since the maximum distance between any pair of points inside of  \baseSegmentLambda/ is $\frac{\lambda}{4} \leq \nicefrac{1}{4} < \robotDiameter{i}$).
	In the second case,  \lambdaCenterP{i} lies outside of \betaHalfSegmentLambda/ since, in the worst case, one endpoint of $\ell_i^{\mathcal{P}}(t)$ is the point $b$ used in the definition of $N$ (see the beginning of \Cref{section:alphaBetaProtocolsUpperBound}) and the second point lies very close above of \betaHalfSegmentLambda/.
	Since \lambdaCenterP{i} is the midpoint of $\ell_i^{P}(t)$, it lies closely above of \betaHalfSegmentLambda/.
	Every other position of the two endpoints of $\ell_i^{P}(t)$ would result in a point \lambdaCenterP{i} that lies even farther above of \betaHalfSegmentLambda/.
\end{proof}

\LemmaSmallDiameter*

\begin{proof}
	First, consider a robot $r_i$ that is located in \mainSegmentLambda/ and all its neighbors are above of \betaHalfSegmentLambda/.
	Let $p_1$ and $p_2$ be the two points of \chull{i} closest to the intersection points of \chull{i} and the boundary of \betaHalfSegmentLambda/ ($p_1$ and $p_2$ are infinitesimally above of \betaHalfSegmentLambda/).
	In case \chull{i} consists of only two robots, define $p_1$ to be the intersection point of \chull{i} and
	\betaHalfSegmentLambda/ and $p_2 = p_i(t)$.
	Clearly, $\robotDiameter{i} \geq |p_1-p_2|$.
	Thus, the maximum distance between any pair of points in $\chull{i}\cap S_{\lambda} \bigl(\frac{\lambda}{2}\bigr)$ is $\lambda \cdot |p_1-p_2|$.
	We conclude that \lambdaCenterP{i} must be located above of \intermediateSegmentLambda/ since \lambdaCenterP{i} is the midpoint of a line segment  either connecting two robots above of $\intermediateSegmentLambda/$ or one robot inside of \mainSegmentLambda/ and one robot outside of \intermediateSegmentLambda/.
	The arguments for the opposite case -- $r_i$ is located in \betaHalfSegmentLambda/, one neighbor of $r_i$ is located in \mainSegmentLambda/ and all others are also outside of \betaHalfSegmentLambda/ -- are analogous.
\end{proof}

\LemmaGlobalRadius*
\begin{proof}
	Fix any segment \baseSegmentLambda/ and consider the set of robots $R_S$ that are located in \mainSegmentLambda/ or compute a target point in \mainSegmentLambda/.
	Via \Cref{lemma:largeDiameter}, we obtain that for every robot $r_i \in R_S$ that computes a target point in  \mainSegmentLambda/, $\robotDiameter{i} \leq \nicefrac{1}{4}$.
	Since the maximum distance between any pair of points in  \mainSegmentLambda/ is $\nicefrac{1}{4}$, we conclude that, a robot which is not located in \mainSegmentLambda/ but computes its target point inside, is at distance at most $\nicefrac{1}{4}$ from \baseSegmentLambda/.
	Hence, via the triangle inequality, it is located at distance at most $\nicefrac{1}{2}$ from any other robot in $R_S$.
	Thus, all robots in $R_S$ can see each other.
	Now consider the robot $r_{\mathrm{min}} \in R_S$ which is the robot of $R_S$ with the minimal number of visible neighbors.
	Furthermore, $A_{\mathrm{min}}$ is the set of robots that have exactly the same neighborhood as $r_{\mathrm{min}}$.
	For all robots $r_j \in R_S \setminus A_{min}$, we have that $r_j$ can see $r_{\mathrm{min}}$ and at least one robot that $r_{\mathrm{min}}$ cannot see.
	Thus, $\robotDiameter{j} > 1$.
	We can conclude with help of \Cref{lemma:largeDiameter} that all robots in $R_S \setminus A_{\mathrm{min}}$ compute a target point outside of \betaHalfSegmentLambda/.
	Since all robots $r_i \in A_{\mathrm{min}}$ have the same neighborhood and $\robotDiameter{i} < \nicefrac{1}{4}$, they also compute the same target point.
	Thus, at the beginning of round $t+1$, at most one position in  \mainSegmentLambda/ is occupied.
	In round $t+1$ we have the picture that one position in  \mainSegmentLambda/ is occupied and all neighbors are located above of \betaHalfSegmentLambda/.
	\Cref{lemma:smallDiameters} yields that the robots in  \mainSegmentLambda/ compute a target point outside.
	Moreover, \Cref{lemma:smallDiameters} yields as well that no robot outside of \mainSegmentLambda/ computes a target point inside and thus,  \mainSegmentLambda/ is empty in round $t+2$.
	Since the segment \baseSegmentLambda/ has been chosen arbitrarily, the arguments hold for the entire circle $N$ and thus, $R(t+2) \leq R(t) -\nicefrac{\lambda}{4} \cdot h \leq R(t) - \frac{\lambda^3  \sqrt{3}}{256 \cdot \pi \cdot \Delta}$.
\end{proof}

\ClassUpperBound*

\begin{proof}
	First, we bound the initial radius of $N$: $R(0) \leq \nicefrac{\Delta}{\sqrt{3}}$ (\Cref{theorem:jungsTheorem}).
	\Cref{lemma:globalRadiusDecrease} yields that $R(t)$ decreases every two rounds by at least $\frac{\lambda^3  \cdot \sqrt{3}}{256 \cdot \pi \cdot \Delta}$.
	Thus, it requires $2 \cdot \frac{256 \cdot \pi \cdot \Delta}{\lambda^3 }$ rounds until $R(t)$ decreases by at least $\sqrt{3}$.
	Next, we bound how often this can happen until $R(t) \leq \frac{1}{4}$ and thus $\globalDiameter \leq \frac{1}{2}$:
	$\frac{\Delta}{\sqrt{3}} - x \cdot \sqrt{3} \leq \frac{1}{4} \iff \frac{\Delta}{3} - \frac{1}{4\cdot \sqrt{3}} \leq x$.

	All in all, it requires $x \cdot \frac{512 \cdot \pi \cdot \Delta}{\lambda^3} = \left(\frac{\Delta}{3} - \frac{1}{4 \cdot \sqrt{3}}\right) \cdot \frac{512 \cdot \pi \cdot \Delta}{\lambda^3 } \leq \frac{171 \cdot \pi \cdot \Delta^2}{\lambda^3}$ rounds until $\globalDiameter{} \leq \frac{1}{2}$.
	As soon as $\globalDiameter \leq\frac{1}{2}$, all robots can see each other, compute the same target point and will reach it in the next round.

\end{proof}

\subsection{Proof of \Cref{theorem:upperBoundAlphaBetaContractingHighDim}}

The analysis is in most parts analogous to the analysis of \lambdaContracting/ protocols in two dimensions (\Cref{section:alphaBetaProtocolsUpperBound}).
Let $N:= N(t)$ be the smallest enclosing hypersphere (SEH) of all robots in round $t$ and $R := R(t)$ its radius.
Let $B$ be an arbitrary point on the surface and define the hyperspherical cap \baseCapLambda/ with apex $B$ as follows.
Choose the height $h$ of  \baseCapLambda/ such that the inscribed hypercone has a slant height of $\nicefrac{\lambda}{8}$.
Note that this implies that the radius $a$ of the base of the cap is upper bounded by $\nicefrac{\lambda}{8}$.
Hence, the maximal distance between any pair of points in \baseCapLambda/ is $\nicefrac{\lambda}{4}$.
In the following, we denote by $\mathrm{HSC}_{\lambda}\left(c\right)$ for $0 \leq c \leq 1$ the hyperspherical cap with apex $B$ of height $c \cdot h$.

\begin{lemma} \label{lemma:heightHighDim}
	$h \geq \frac{\sqrt{2} \cdot \lambda^2}{64 \cdot \pi \cdot \Delta}$.
\end{lemma}

\begin{proof}
	Initially, we give a bound on the angle $\gamma$ which is the angle between the height $h$ and the slant height of $\nicefrac{\lambda}{8}$ of the inscribed hypercone.
	Consider a circle $K$ with radius $R$ that has the same center as $N$ and contains $B$.
	The angle $\gamma$ can now be seen as the internal angle of a regular polygon with side length $\nicefrac{\lambda}{8}$ whose vertices lie on $K$.
	The circumference of $K$ is $2 \cdot \pi \cdot R$.
	Thus, we can position at most $\frac{16}{\lambda} \cdot \pi \cdot R$ points on the boundary of $K$ that are in distance $\nicefrac{\lambda}{8}$ from the points closest to them and form a regular convex polygon.
	The internal angle of this regular polygon is $2\cdot \gamma$.
	Hence, the sum of all internal angles is $\left(\frac{16}{\lambda} \cdot \pi \cdot R -2 \right) \cdot \pi$.
	Thus, each individual angle has a size of at most $\frac{\left(\frac{16}{\lambda} \cdot \pi \cdot R-2\right) \cdot \pi }{\frac{16}{\lambda} \cdot \pi \cdot R} =\pi - \frac{2\cdot \pi}{\frac{16}{\lambda} \cdot \pi \cdot R} = \pi - \frac{\lambda}{8 \cdot R}$.
	Hence, $\gamma \leq \frac{\pi}{2} - \frac{\lambda}{16 \cdot R}$.
	Now, we are able to bound $h$.
	First of all, we derive a relation between $h$ and $\gamma$:
	$\cos \left(\gamma\right) = \frac{h}{\frac{\lambda}{8}} = \frac{8h}{\lambda} \iff h = \frac{\lambda \cdot \cos \left(\gamma\right)}{8}$.
	In the following upper bound, we make use of the fact that $\cos \left(x\right) \geq -\frac{2}{\pi}x +1 $ for $x \in [0, \frac{\pi}{2}]$.

	\begin{align*}
		h = \frac{\lambda \cdot \cos \left(\gamma\right)}{8} \geq \frac{\lambda \cdot \cos \left(\frac{\pi}{2} - \frac{\lambda}{16 \cdot R}\right)}{8}
		\geq \frac{\lambda \cdot \left(-\frac{2}{\pi} \cdot \left(\frac{\pi}{2} - \frac{\lambda}{16 \cdot R}\right) +1 \right)}{8}	= \frac{\lambda \cdot \frac{\lambda}{8 \cdot \pi \cdot R}}{8} = \frac{\lambda^2}{64 \pi R}
	\end{align*}

	Lastly observe that $R \leq \Delta \cdot \sqrt{\frac{d}{2 \cdot \left(d+1\right)}}$ (\Cref{theorem:jungsTheorem}).
	For any $d \geq 1$, it holds $\sqrt{\frac{d}{2 \cdot \left(d+1\right)}} \leq \frac{1}{\sqrt{2}}$ and thus $R \leq \frac{\Delta}{\sqrt{2}}$.
	We obtain a final lower bound on $h$:
	$ h\geq \frac{\lambda^2}{64 \pi R} \geq \frac{\sqrt{2} \cdot \lambda^2}{64\cdot \pi \cdot \Delta}$.

\end{proof}

\begin{lemma} \label{lemma:largeDiameterHighDim}
	For a robot $r_i$ with $\robotDiameter{i} > \nicefrac{1}{4}$ it holds $\fp{i} \in N \setminus \betaHalfCapLambda/$.
\end{lemma}

\begin{proof}
	Analogous to the proof of \Cref{lemma:largeDiameter}.
\end{proof}

\begin{lemma} \label{lemma:smallDiametersHighDim}
	Consider a robot $r_i$ located in \mainCapLambda/.
	If all its neighbors are located outside of $\betaHalfCapLambda/$, it holds $\fp{i} \in N \setminus$ \mainCapLambda/.
	Similarly, for a robot $r_i$ that is located outside of \betaHalfCapLambda/ and that has only one neighbor located in \mainCapLambda/, it holds $\fp{i} \in N \setminus $ \mainCapLambda/.
\end{lemma}

\begin{proof}
	Analogous to the proof of \Cref{lemma:smallDiameters}.
\end{proof}

\begin{lemma} \label{lemma:globalRadiusDecreaseHighDim}
	For any round $t$ with $\globalDiameter \geq \nicefrac{1}{2}$, it holds $R(t+2) \leq R(t) - \frac{\lambda^3 \cdot \sqrt{2}}{256 \cdot \pi \cdot \Delta}$.
\end{lemma}

\begin{proof}
	Analogous to the proof of \Cref{lemma:globalRadiusDecrease} by replacing the lower bound of $h$ by $ \frac{\sqrt{2} \cdot \lambda^2}{64\cdot \pi \cdot \Delta}$ (\Cref{lemma:heightHighDim}).
\end{proof}

\ClassUpperBoundHighDim*

\begin{proof}
	First, we bound the initial radius of $N$: $R(0) \leq \frac{\Delta}{\sqrt{2}}$ (\Cref{theorem:jungsTheorem}).
	\Cref{lemma:globalRadiusDecreaseHighDim} yields that $R(t)$ decreases every two rounds by at least $\frac{\lambda^3  \cdot \sqrt{2}}{256 \cdot \pi \cdot \Delta}$.
	Thus, it requires $2 \cdot \frac{256 \cdot \pi \cdot \Delta}{\lambda^3}$ rounds until $R(t)$ decreases by at least $\sqrt{2}$.
	Next, we bound how often this can happen until $R(t) \leq \frac{1}{4}$ and thus $\globalDiameter{} \leq \frac{1}{2}$ holds:
	$\frac{\Delta}{\sqrt{2}} - x \cdot \sqrt{2} \leq \frac{1}{4} \iff \frac{\Delta}{2} - \frac{1}{4 \cdot \sqrt{2}} \leq x$.

	All in all, it requires $x \cdot \frac{512 \cdot \pi \cdot \Delta}{\lambda^3} = \left(\frac{\Delta}{2} - \frac{1}{4 \cdot \sqrt{2}}\right) \cdot \frac{512 \cdot \pi \cdot \Delta}{\lambda^3 } \leq \frac{256 \cdot \pi \cdot \Delta^2}{\lambda^3 }$ rounds until $\globalDiameter{} \leq \frac{1}{2}$.
	As soon as $\globalDiameter{} \leq \frac{1}{2}$, all robots can see each other, compute the same target point and will reach it in the next round.

\end{proof}

\subsection{Proof of \Cref{lem:ab-contracting-is-lambda-contracting}}
\LemmaAlphaBetaIsLambda*

\begin{proof}
	From the definition of \alphaBetaContracting/, we know that for a target point $\fp{i}$, there exists a $\alphaCenterP{i}$ such that $\fp{i} \in \chull{i}(\alphaCenterP{i},\beta)$.
	We do the following geometric construction in \cref{fig:alpha-beta-contracting-to-lambda-contracting}.
	Let $p = \alphaCenterP{i}$ and $p' = \fp{i}$.
	We draw a line segment from \alphaCenterP{i} through $\fp{i}$ to the boundary of $\chull{i}$.
	Let $c$ be the endpoint of this line segment.
	Because $p$ is \alphaCentered/, there exists a line segment with length $\robotDiameter{i} \cdot \alpha$ through $p$, let this be the line segment $\overline{ab}$.
	The line segment $\overline{a'b'}$ is a parallel to $\overline{ab}$ inside the triangle $\triangle_{abc}$.
	We know that $p' \in  \chull{i}(p,\beta)$, therefore $|\overline{cp'}| \geq \beta |\overline{cp}|$.
	By the intercept theorem, it follows that $|\overline{a'b'}| \geq \beta |\overline{ab}| = \beta \cdot \alpha \cdot \robotDiameter{i}$.
	Because the points $a, b$ and $c$ are all inside $\chull{i}$, the entire triangle $\triangle_{abc}$ and $\overline{a'b'}$ are inside $\chull{i}$ as well.
	Therefore, $\fp{i}$ is a $\lambda$-centered point with $\lambda = \alpha \cdot \beta$.

	\begin{figure}[htbp]
		\centering
		\includegraphics[width = 0.65\textwidth]{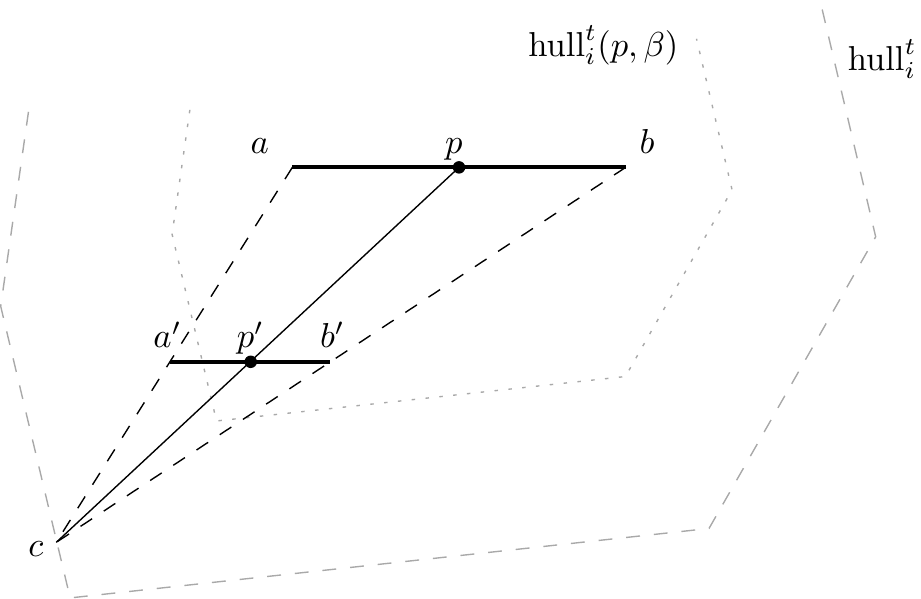}
		\caption{The construction used in the proof of \Cref{lem:ab-contracting-is-lambda-contracting}.}
		\label{fig:alpha-beta-contracting-to-lambda-contracting}
	\end{figure}
\end{proof}

\newpage
\subsection{Pseudocode of \gtcShort/}

\begin{algorithm}[htb]
	\caption{\gtc/ (view of robot $r_i$)}
	\label{algorithm:gtc2D}
	\begin{algorithmic}[1]
		\State $C_i(t) :=$ smallest enclosing circle of $N_i(t)$
		\State $c_i(t) := $ center of $C_i(t)$
		\State $\forall r_j \in N_i(t) : m_j :=$ midpoint between $r_i$ and $r_j$
		\State $D_j: $ disk with radius $\frac{1}{2}$ centered at $m_j$
		\State $\mathrm{seg} := $ line segment $\overline{p_i(t), c_i(t)}$
		\State $A := \bigcap_{r_j \in N_i(t)} D_j \cap \mathrm{seg}$
		\State $x :=$ point in $A$ closest to $c_i(t)$
		\State $\mathrm{target}_i^{\gtcShort/}(t) := x$
	\end{algorithmic}
\end{algorithm}

\subsection{Proof of \Cref{theorem:gtcAlphaBeta}}

In the following, we prove that \gtcShort/ is \alphaBetaContracting/.
First, we derive a bound on $\alpha$.

\begin{restatable}{lemma}{LemmaGtcAlpha}\label{lemma:gtcAlphaEstimation2D}
	The center of the SEC of a convex polygon $Q$ is $\frac{\sqrt{3}}{8}$-centered.
\end{restatable}

\begin{proof}
	Let $C$ denote the smallest enclosing circle (SEC) of $Q$.
	We need to distinguish two cases: either two points are located on the boundary of the SEC or (at least) $3$ that form an acute triangle.
	In the first case, the two points are located on the diameter of $C$.
	Hence, the center of the smallest enclosing circle is $1$-centered since it equals to the midpoint of the two robots that define the diameter.

	In the second case, we focus on the three points that form the acute triangle.
	We denote the three points by $ABC$ and the triangle by $\triangle_{ABC}$.
	Moreover, $a,b,c$ denote the edges of $\triangle_{ABC}$.
	Since $\triangle_{ABC}$ is acute, we have that its circumcircle equals to its SEC.
	Consequently, $C$ equals to the SEC of $\triangle_{ABC}$.
	Let $p$ denote the center of $C$ and $r$ its radius.
	The aim of the proof is to show that there exists a line segment $\ell$ with midpoint $p$ that is parallel to $a,b$ or $c
	$ and has a length of $\alpha \cdot \mathrm{diam}$, where $\mathrm{diam}$ denotes the diameter of $Q$ (we will determine the concrete value for $\alpha$ shortly).
	W.l.o.g.\, we assume that $a$ is the longest edge of $\triangle_{ABC}$.
	Since $C$ is also the SEC of $\triangle_{ABC}$, $r \leq \frac{a}{\sqrt{3}}$ and, thus, also $r \cdot \sqrt{3} \leq a$ (\Cref{theorem:jungsTheorem}).
	Additionally, we have $r \geq \frac{\mathrm{diam}}{2}$.
	Hence, $\frac{\mathrm{diam} \cdot \sqrt{3}}{2} \leq a$.
	Now, we rotate the coordinate system, such that $B = \left(0,-\frac{a}{2}\right), C = \left(0, \frac{a}{2}\right)$ and $A = \left(x_a, y_a\right)$.
	See \Cref{figure:gtcAlphaCenteredTriangle} for a visualization of the setting.
	We now consider $y_a \leq 0$.
	The arguments for $y_a > 0$ can be derived analogously with swapped roles of $B$ and $C$.

	\begin{figure}[htb]
		\begin{minipage}[t]{0.49\textwidth}
			\centering
			\includegraphics[width = \textwidth]{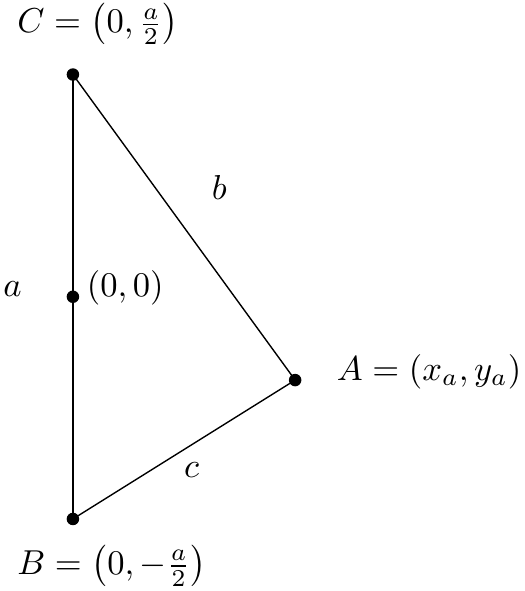}
			\caption{A visualization of $\triangle_{ABC}$, where $a$ is parallel to the $y$-axis. }
			\label{figure:gtcAlphaCenteredTriangle}
		\end{minipage}
		\hfill
		\begin{minipage}[t]{0.49\textwidth}
			\centering
			\includegraphics[width = \textwidth]{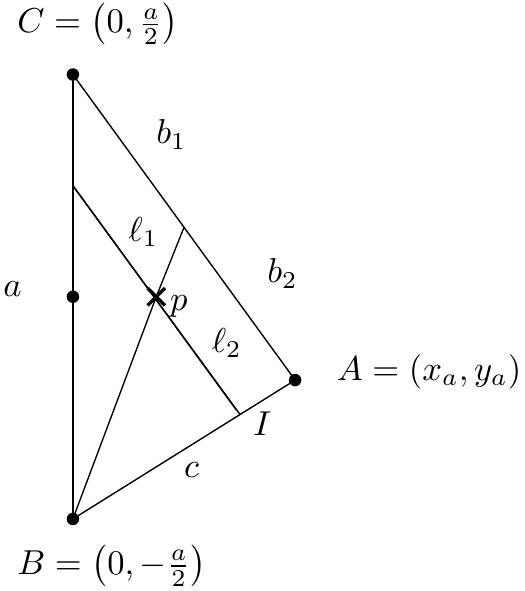}
			\caption{A visualization of the second case, where the line $\ell$ is parallel to the edge $b$. }
			\label{figure:gtcAlphaCenteredTriangleCase2}
		\end{minipage}
	\end{figure}

	Observe first, that $p = \left(x_p, 0\right)$ since $p$ is located on the intersection of the perpendicular bisectors of $a,b$ and $c$.
	Additionally, we have $x_p \leq \frac{x_a}{2}$ since both the starting points of the perpendicular bisectors of $b$ and $c$ have the $x$-coordinate $\frac{x_p}{2}$ and the bisectors have monotonically decreasing $x$-coordinates.

	Now, we distinguish two cases: $-\frac{a}{4} \leq y_a \leq 0$ and $y_a < -\frac{a}{4}$.
	In the first case, we prove that there exists a line segment parallel to $a$ with its midpoint in $p$ that as a length of at least $\frac{a}{4}$.
	We model the edge $b$ as a linear function $f(x) = - \frac{\frac{a}{2} - y_a}{x_a} \cdot x + \frac{a}{2}$.
	The value $f(x_p)$ is minimized for $x_p = \frac{x_a}{2}$ and $y_a = -\frac{a}{4}$.
	Hence,
	\begin{align*}
		f(x_p) \geq f\bigl(\frac{x_a}{2}\bigr) &= -\frac{\frac{a}{2}-y_a}{x_a} \cdot \frac{x_a}{2} + \frac{a}{2} = -\frac{a}{4} + \frac{y_a}{2} + \frac{a}{2} = \frac{a}{4} + \frac{y_a}{2} \\
		&\geq \frac{a}{4} - \frac{a}{8} = \frac{a}{8}.
	\end{align*}

	Hence, we can center a line segment of length $\frac{a}{4}$ on $p$ that is parallel to $a$ and completely contained in $\triangle_{ABC}$.

	Next, we consider $y_a < - \frac{a}{4}$.
	As long as $f(x_p) \geq \frac{a}{8}$, we can use the same arguments as for the first case.
	Thus, we assume that $f(x_p) < \frac{a}{8}$.
	Now, we show that there is a line parallel to $b$ with midpoint $p$ and with a length of at least $\frac{d \cdot \sqrt{3}}{8}$.
	Since $y_a < - \frac{a}{4}$, we conclude $b \geq \frac{3}{4}a$.
	Let $\ell$ be the line segment parallel to $y$ that is completely contained in $\triangle_{ABC}$.
	Note that $p$ does not need to be the midpoint of $\ell$, see also \Cref{figure:gtcAlphaCenteredTriangleCase2} for a depiction.
	Let $I = \left(x_i,y_i\right)$ be the intersection of $\ell$ and $c$.
	We conclude $x_i \geq \frac{x_a}{2}$ since $p$ lies on the perpendicular bisector centered in the midpoint of $c$ with $x$-coordinate $\frac{x_a}{2}$ and $\ell$ is parallel to $b$.
	Hence, $|B\,I| \geq \frac{c}{2}$.
	Applying the intercept theorem yields $\frac{|B\,I|}{c} = \frac{\ell}{b}$ and $\frac{\ell}{b} \geq \frac{1}{2} \iff \ell \geq \frac{b}{2}$.
	Since $b \geq \frac{3}{4}a$, we conclude $\ell \geq \frac{3}{8}a$.

	It remains to estimate the position of $p$ on $c$ to give a final bound for $\alpha$.
	We use the line segment starting in $B$, leading through $p$ and intersecting $b$ to split $b$ and $\ell$ into $b_1$ and $b_2$ as well as $\ell_1$ and $\ell_2$.
	See also \Cref{figure:gtcAlphaCenteredTriangleCase2} for a visualization.
	As $f(x_p) < \frac{a}{8}$, we obtain $b_1 \geq \frac{3}{8}a$ and $b_2 \leq \frac{5}{8}a$.
	Hence, also $b_1 \geq \frac{3}{8}b$.
	The intercept theorem yields $\ell_1 \geq \frac{3}{8}\ell$ and hence, we can center a line segment of length $\frac{2}{3}\ell$ in $p$ that is parallel to $b$.
	Finally, we conclude $\frac{2}{3} \ell \geq \frac{2}{3} \cdot \frac{3}{8}a = \frac{a}{4} \geq \frac{\mathrm{diam} \cdot \sqrt{3}}{8}$.
	All in all, we obtain that we can always center a line segment of length at least $\frac{\mathrm{diam} \cdot \sqrt{3}}{8}$ in $p$ that is completely contained in $Q$.
\end{proof}

Subsequently, we state general properties of SECs to derive a lower bound on the constant $\beta$ afterwards.

\begin{theorem} [\cite{chrystal1885problem}]\label{theorem:secChrystal}
	Let $C$ be the SEC if a point set S. Then, either there are two points $P,Q \in S$ on the circumference of $C$ such that the line segment $\overline{PQ}$ is a diameter of $C$ or there are three point $P,Q,R \in S$ on the circumference of $S$ such that the center $c$ of $C$ is inside the acute angled $\triangle_{PQR}$.
	Furthermore, $C$ is always unique.
\end{theorem}

\begin{restatable}{lemma}{LemmaGtcBeta} \label{lemma:gtcBetaEstimation2D}
	A robot $r_i$ that is at distance $d_{\mathrm{target}}$ from the center $c$ of its SEC moves at least a distance of $\frac{d_{\mathrm{target}}}{2}$ towards $c$.
\end{restatable}

\begin{proof}
	Let $c$ be the center of $r_i$'s SEC $C$.
	We rotate and translate the coordinate system such that $c = \left(0,0\right)$ and $r_i$ is located at $(x_i,0)$, i.e., $r_i$ is at distance $x_i$ from $c$.
	Additionally, we define $a$ to be the radius of $C$.
	Observe first that $a \leq 1$ since there must be at least one robot $r_j = (x_j,y_j)$ with $x_j \leq 0$ on the boundary of $C$ (see \Cref{theorem:secChrystal}) and $r_j$ is at distance at most $1$ from $r_i$.

	Now, let $r_k = (x_k,y_k)$ be a robot in $r_i$'s neighborhood and $m_k$ be the midpoint between $r_i$ and $r_k$.
	We will prove that $m_k$ is at distance  at most $\frac{1}{2}$ from the point $\left(\frac{x_i}{2},0\right)$.
	First of all, we calculate the coordinates of $m_k$:
	$m_k = \left(\frac{1}{2} \cdot \left(x_i+x_k\right), \frac{1}{2} \cdot y_k\right)$.
	The distance between $\left(\frac{x_i}{2}, 0\right)$ and $m_k$ is $\sqrt{\frac{1}{4} \cdot x_k^2 + \frac{1}{4} y_k^2}$.
	Basic calculus yields $\sqrt{\frac{1}{4} \cdot x_k^2 + \frac{1}{4} y_k^2} \leq \frac{1}{2} \iff -1 \leq x_k \leq 1$ and $-\sqrt{1-x_k^2} \leq y_k \leq \sqrt{1-x_k^2}$.
	Since $a \leq 1$, the inequalities for $x_k$ and $y_k$ are fulfilled.
	Hence, $r_i$ can move at least half its distance towards $c$.
\end{proof}

The combination of \Cref{lemma:gtcAlphaEstimation2D,lemma:gtcBetaEstimation2D} yields the correctness of \Cref{theorem:gtcAlphaBeta}.

\TheoremGtcAlphaBeta*

\subsection{Pseudocode of $d$-\textsc{GtC}} \label{section:appendixDGtc}

Next, we generalize the \gtcShort/ algorithm to any dimension $d$.
Instead of moving towards the smallest enclosing circle of their neighborhood, robots move towards the center of the smallest enclosing hypersphere.
The description can be found in \Cref{algorithm:gtcD}.

\begin{algorithm}[H]
	\caption{$d$-\gtc/ (view of robot $r_i$)}
	\label{algorithm:gtcD}
	\begin{algorithmic}[1]
		\State $C_i(t) :=$ smallest enclosing hypersphere of $N_i(t)$
		\State $c_i(t) := $ center of $C_i(t)$
		\State $\forall r_j \in N_i(t) : m_j :=$ midpoint between $r_i$ and $r_j$
		\State $D_j: $ hypersphere with radius $\frac{1}{2}$ centered at $m_j$
		\State $\mathrm{seg} := $ line segment $\overline{p_i(t), c_i(t)}$
		\State $A := \bigcap_{r_j \in N_i(t)} D_j \cap \mathrm{seg}$
		\State $x :=$ point in $A$ closest to $c_i(t)$
		\State $\mathrm{target}_i^{\gtcShort/}(t) := x$
	\end{algorithmic}
\end{algorithm}

\subsection{Proof of \Cref{theorem:dGtC}}

For the analysis, we first state two general properties of smallest enclosing hyperspheres.

\begin{lemma}[\cite{10.2307/2629114}]\label{lemma:hypersphereConvexCombination}
	Let $S$ be the smallest enclosing hypersphere of a set of points $P \subset \mathbb{R}^{m}$.
	The center $c$ of $S$ is a convex combination of at most $m+1$ points on the surface of $S$.
\end{lemma}

\begin{lemma}[\cite{DBLP:conf/esa/FischerGK03}]\label{lemma:hypersphereSimplexCircumsphere}
	Let $T$ be a set of points on the boundary of some hypersphere $H$ with center $c$.
	$H$ is the smallest enclosing hypersphere of $T$ if and only if $c$ is a convex combination of the points in $T$.
\end{lemma}

Next, we state that the center of the smallest enclosing hypersphere is in general $\frac{\sqrt{2}}{8}$-centered, in contrast to $\frac{\sqrt{3}}{8}$ for $d=2$.

\begin{lemma}\label{lemma:gtcAlphaEstimation}
	The center of the smallest enclosing hypersphere of a convex polytope $Q \subset R^{d}$ is $\frac{\sqrt{2}}{8}$-centered.
\end{lemma}

\begin{proof}
	Let $C$ denote the smallest enclosing hypersphere (SEH) of $Q$ and $c_i$ its center.
	We need to distinguish two cases: either two points are located on the boundary of SEH or (at least) $3$.
	It is well known, that $c_i$ is a convex combination of at most $d+1$ points on the boundary of $C$ (\Cref{lemma:hypersphereConvexCombination}).
	Those points form a simplex $S$.
	From \Cref{lemma:hypersphereSimplexCircumsphere}, it follows that $C$ is also the SEH of $S$ since $C$ is a circumsphere of $S$ and $c_i$ is inside of $S$.

	In case, there are only two points on the boundary of $C$, they must be the endpoints of a diameter of $C$.
	Hence, the center of the SEH is $1$-centered since it equals to the midpoint of the two points that define the diameter.

	Otherwise, $S$ consists of at least $3$ points.
	We take two points  of $S$ that have the maximal distance of all points in $S$ and denote those points as $B$ and $C$.
	Additionally, we take an arbitrary third point of $S$ and call it $A$.
	The points $A,B$ and $C$ form a triangle $\triangle_{ABC}$.
	Moreover, $a,b$ and $c$ denote the edges of $\triangle_{ABC}$.

	Let $r$ denote the radius of $C$.
	The aim of the proof is to show that there exists a line segment $\ell$ with midpoint $c_i$ that is parallel to $a,b$ or $c
	$ and has a length of $\alpha \cdot d$, where $d$ denotes the diameter of $Q$ (we will determine the concrete value for $\alpha$ shortly).
	Recall that $a$ is the longest edge of $\triangle_{ABC}$.
	Since $C$ is also the SEH of $S$, it holds $r \leq a \cdot \sqrt{\frac{d}{2 \cdot \left(d+1\right)}}$  and thus also $\frac{r}{\sqrt{\frac{d}{2 \cdot \left(d+1\right)}}} \leq a$ (\Cref{theorem:jungsTheorem}).
	Additionally, it holds $r \geq \frac{d}{2}$.
	Hence, $\frac{d}{2\cdot \sqrt{\frac{d}{2 \cdot \left(d+1\right)}}}  \leq a$.
	Now, we rotate the coordinate system, such that $B = \left(0,-\frac{a}{2}\right), C = \left(0, \frac{a}{2}\right)$ and $A = \left(x_a, y_a\right)$.
	See \Cref{figure:gtcAlphaCenteredTriangle} for a visualization of the setting.
	We now consider $y_a \leq 0$, the arguments for $y_a > 0$ can be derived analogously with swapped roles of $B$ and $C$.

	Observe first, that $c_i = \left(x_{c_i}, 0\right)$ since $x_i$ is located on the intersection of the perpendicular bisector hyplerplanes of $a,b$ and $c$.
	Additionally, it holds $x_{c_i} \leq \frac{x_a}{2}$ since both the midpoints of $b$ and $c$ have the $x$-coordinate $\frac{x_p}{2}$ and the parts of bisector hyperplanes inside of $S$ have monotonically decreasing $x$-coordinates.

	Now, we distinguish two cases: $-\frac{a}{4} \leq y_a \leq 0$ and $y_a < -\frac{a}{4}$.
	In the first case, we prove that there exists a line segment parallel to $a$ with its midpoint in $c_i$ that as a length of at least $\frac{a}{4}$.
	We model the edge $b$ as a linear function $f(x) = - \frac{\frac{a}{2} - y_a}{x_a} \cdot x + \frac{a}{2}$.
	The value $f(x_{c_i})$ is minimized for $x_p = \frac{x_a}{2}$ and $y_a = -\frac{a}{4}$.
	Hence,
	\begin{align*}
		f(x_{c_i}) \geq f\left(\frac{x_a}{2}\right) &= -\frac{\frac{a}{2}-y_a}{x_a} \cdot \frac{x_a}{2} + \frac{a}{2} = -\frac{a}{4} + \frac{y_a}{2} + \frac{a}{2} = \frac{a}{4} + \frac{y_a}{2} \\
		&\geq \frac{a}{4} - \frac{a}{8} \\
		&= \frac{a}{8}.
	\end{align*}

	Hence, we can center a line segment of length $\frac{a}{4}$ on $c_i$ that is parallel to $a$ and completely contained in $\triangle_{ABC}$.

	Next, we consider $y_a < - \frac{a}{4}$.
	As long as $f(x_{c_i}) \geq \frac{a}{8}$ holds, we can use the same arguments as for the first case.
	Thus, we assume that $f(x_{c_i}) < \frac{a}{8}$.
	Now, we show that there is a line parallel to $b$ with midpoint $c_i$ and a length of at least $\frac{d \cdot \sqrt{2}}{8}$.
	Since $y_a < - \frac{a}{4}$, it holds $b \geq \frac{3}{4}a$.
	Let $\ell$ be the line segment parallel to $y$ that is completely contained in $\triangle_{ABC}$.
	Note that $c_i$ does not need to be the midpoint of $\ell$, see also \Cref{figure:gtcAlphaCenteredTriangleCase2} for a depiction.
	Let $I = \left(x_i,y_i\right)$ be the intersection of $\ell$ and $c$.
	We conclude $x_i \geq \frac{x_a}{2}$ since $p$ lies on the perpendicular bisector centered in the midpoint of $c$ with $x$-coordinate $\frac{x_a}{2}$ and $\ell$ is parallel to $b$.
	Hence, $|B\,I| \geq \frac{c}{2}$.
	Applying the intercept theorem yields $\frac{|B\,I|}{c} = \frac{\ell}{b}$ and $\frac{\ell}{b} \geq \frac{1}{2} \iff \ell \geq \frac{b}{2}$.
	Since $b \geq \frac{3}{4}a$, we conclude $\ell \geq \frac{3}{8}a$.

	It remains to estimate the position of $c_i$ on $c$ to give a final bound for $\alpha$.
	We use the line segment starting in $B$, leading through $c_i$ and intersecting $b$ to split $b$ and $\ell$ into $b_1$ and $b_2$ as well as $\ell_1$ and $\ell_2$.
	See also \Cref{figure:gtcAlphaCenteredTriangleCase2} for a visualization.
	As $f(x_{c_i}) < \frac{a}{8}$, we obtain $b_1 \geq \frac{3}{8}a$ and $b_2 \leq \frac{5}{8}a$.
	Hence, also $b_1 \geq \frac{3}{8}b$.
	The intercept theorem yields $\ell_1 \geq \frac{3}{8}\ell$ and hence, we can center a line segment of length $\frac{2}{3}\ell$ in $p$ that is parallel to $b$.
	Finally, we conclude $\frac{2}{3} \ell \geq \frac{2}{3} \cdot \frac{3}{8}a = \frac{a}{4} \geq \frac{d}{\sqrt{\frac{d}{2 \cdot \left(d+1\right)}} \cdot 8} \geq \frac{d \cdot \sqrt{2}}{8}$ since $\lim\limits_{d \rightarrow \infty} \sqrt{\frac{d}{2 \cdot \left(d+1\right)}} = \frac{1}{\sqrt{2}}$.
	All in all, we obtain that we can always center a line segment of length at least $\frac{d \cdot \sqrt{2}}{8}$ in $c_i$ that is completely contained in $Q$.
\end{proof}

\begin{lemma} \label{lemma:gtcBetaEstimation}
	A robot $r_i$ that is at distance $d_{\mathrm{target}}$ from the center $c$ of its SEC moves at least a distance of $\frac{d_{\mathrm{target}}}{2}$ towards $c$.
\end{lemma}

\begin{proof}
	Let $c$ be the center of $r_i$'s SEC $C$.
	We rotate and translate the coordinate system such that $c = \left(0,0\right)$ and $r_i$ is located at $(x_i,0)$, i.e., $r_i$ is in distance $x_i$ of $c$.
	Additionally, we define $a$ to be the radius of $C$.
	Observe first that $a \leq 1$ since there must be at least one robot $r_j = (x_j,y_j)$ with $x_j \leq 0$ on the boundary of $C$ and $r_j$ is in distance at most $1$ of $r_i$.

	Now, let $r_k = (x_k,y_k)$ be a robot in $r_i$'s neighborhood and $m_k$ be the midpoint between $r_i$ and $r_k$.
	We will prove that $m_k$ is in distance of at most $\frac{1}{2}$ of the point $\left(\frac{x_i}{2},0\right)$.
	First of all, we calculate the coordinates of $m_k$:
	$m_k = \left(\frac{1}{2} \cdot \left(x_i+x_k\right), \frac{1}{2} \cdot y_k\right)$.
	The distance between $\left(\frac{x_i}{2}, 0\right)$ and $m_k$ is $\sqrt{\frac{1}{4} \cdot x_k^2 + \frac{1}{4} y_k^2}$.
	Basic calculus yields $\sqrt{\frac{1}{4} \cdot x_k^2 + \frac{1}{4} y_k^2} \leq \frac{1}{2} \iff -1 \leq x_k \leq 1$ and $-\sqrt{1-x_k^2} \leq y_k \leq \sqrt{1-x_k^2}$.
	Since $a \leq 1$ holds, the inequalities for $x_k$ and $y_k$ are fulfilled.
	Hence, $r_i$ can move at least half its distance towards $c$.
\end{proof}

\TheoremDGtc*

\begin{proof}
	A conclusion of \Cref{lemma:gtcAlphaEstimation,lemma:gtcBetaEstimation}.
\end{proof}

\subsection{Pseudocode of \gtmdShort/}

\begin{algorithm}[htb]
	\caption{\gtmd/ (view of robot $r_i$)}
	\label{algorithm:gtcmd}
	\begin{algorithmic}[1]
		\State $p_{d_1}(t), p_{d_2}(t) :=$ positions of robots in $N_i(t)$ that have the maximal distance
		\If{$p_{d_1}(t)$ and $p_{d_2}(t)$ are unique}
		\State $m_d(t) := \frac{1}{2} \cdot \left(p_{d_1}(t) + p_{d_2}(t)\right)$
		\State $\forall r_j \in N_i(t) : m_j :=$ midpoint between $r_i$ and $r_j$
		\State $D_j: $ disk with radius $\frac{1}{2}$ centered at $m_j$
		\State $\mathrm{seg} := $ line segment $\overline{p_i(t), m_d(t)}$
		\State $A := \bigcap_{r_j \in N_i(t)} D_j \cap \mathrm{seg}$
		\State $x :=$ point in $A$ closest to $m_d(t)$
		\State $\mathrm{target}_i^{\gtmdShort/}(t) := x$
		\Else
		\State Compute $\mathrm{target}^{\gtmdShort/}_i(t)$ with \gtcShort/
		\EndIf
	\end{algorithmic}
\end{algorithm}

\subsection{Proof of \Cref{theorem:gtmdContracting}}

\TheoremGtmd*

\begin{proof}
	Obviously, the midpoint of the diameter is $1$-centered and thus, we obtain $\alpha = 1$.
	In the following, $\mathrm{md}_i(t)$ denotes the midpoint of the diameter of $r_i$ in round $t$.
	The value for $\beta$ is more difficult to determine.
	The main difference between \gtcShort/ and \gtmdShort/ is that each robot that is not located on the SEC of its neighborhood moves towards in \gtcShort/ whereas some robots are not allowed to move towards the midpoint of the diameter in \gtmdShort/.
	However, we will show that those robots lie already inside of $\chull{i}\left(\mathrm{md}_i(t), \frac{1}{10}\right)$ such that \gtmdShort/ is still $\left(1,\frac{1}{10}\right)$-contracting.

	W.l.o.g., we assume that $V=1$.
	Analogously, results for arbitrary $V$ can be derived.
	Consider any robot $r_i$.
	We first show that $\robotDiameter{i} \leq \frac{2}{\sqrt{3}}$ implies that $r_i$ is allowed to move half its distance towards $\mathrm{md}_i(t)$.
	Let $r_{d_1}$ and $r_{d_2}$ be the two robots that define the local diameter of $r_i$.
	We rotate and translate the global coordinate system such that $\mathrm{md}_i(t) = \left(0,0\right), p_{d_1} = (0, \frac{\robotDiameter{i}}{2})$, $p_{d_2} = \left(0,-\frac{\robotDiameter{i}}{2}\right)$ and $p_i(t) = \left(x_i(t), y_i(t)\right)$ with $x_i(t) \geq 0$ and $y_i(t) \geq 0$.
	Next, we assume that $\robotDiameter{i} \leq \frac{2}{\sqrt{3}}$.
	Consider any robot $r_j $ located at $p_j(t) = \left(x_j(t), y_j(t)\right)$.
	By definition, $m_j(t) = \left(\frac{1}{2} \cdot \left(x_i(t) + x_j(t)\right), \frac{1}{2} \cdot \left(y_i(t) + y_j(t)\right) \right)$.
	The distance between $m_j(t)$ and $\left(\frac{1}{2}x_i(t), \frac{1}{2}y_i(t)\right)$ is $\sqrt{\frac{1}{4}x_j(t)^2 + \frac{1}{4}y_j(t)^2}$.
	Next, observe that the maximum $x_j(t)$ is reached when $r_j$ is located on the $x$-axis as far as possible to the right.
	More formally, we can upper bound the maximal $x_j$ as follows: $\sqrt{x_j(t)^2 + \frac{\robotDiameter{i}^2}{4}} = \robotDiameter{i} \iff x_j(t) = \frac{\sqrt{3}}{2} \cdot \robotDiameter{i}$.
	For $\robotDiameter{i} \leq \frac{2}{\sqrt{3}}$, we obtain $-1 \leq x_j(t) \leq 1$.
	Hence, all robots $r_j$ that are at a distance of at most $1$ of $r_i$ fulfill this criterion and hence, $r_i$ is able to move half its distance towards $\mathrm{md}_i(t)$.

	Next, we assume that $\frac{2}{\sqrt{3}} < \robotDiameter{i} \leq 2$.
	As mentioned earlier, we only consider robots that are not already located in $\chull{i}\left(\mathrm{md}_i(t), \frac{1}{10}\right)$.
	For simplicity, we assume that $r_i$ is located on the boundary of $\chull{i}$, the arguments for the case where $r_i$ lies inside $\chull{i} \setminus \chull{i}\left(\mathrm{md}_i(t), \frac{1}{10}\right)$ are analogous.
	To determine the edges of \chull{i}, we only can take the positions of three robots for granted: the positions of $r_{\mathrm{diam,1}}, r_{\mathrm{diam,2}}$ and of $r_i$ (which even might be either $r_{d_1}$ or $r_{d_2}$).
	We model the edges of \chull{i} as linear function.
	The first function is $f(x) = \left(\frac{\left(y_i(t) - \frac{\robotDiameter{i}}{2}\right)}{x_i(t)}\right) \cdot x+ \frac{\robotDiameter{i}}{2} =   \left(\frac{\left(2y_i(t) - \robotDiameter{i}\right)}{2x_i(t)}\right) \cdot x+ \frac{\robotDiameter{i}}{2}$ and describes the edge through $r_{d_1}$.
	The second function $g(x) = \left(\frac{\left(2y_i(t) +2\robotDiameter{i}\right)}{2x_i(t)}\right) \cdot x+ \frac{\robotDiameter{i}}{2}$ describes the edge through $r_{d_2}$.
	The setup is depicted in \Cref{figure:gtmdProofSetup}.

	\begin{figure}[htbp]
		\centering
		\includegraphics[width = 0.5\textwidth]{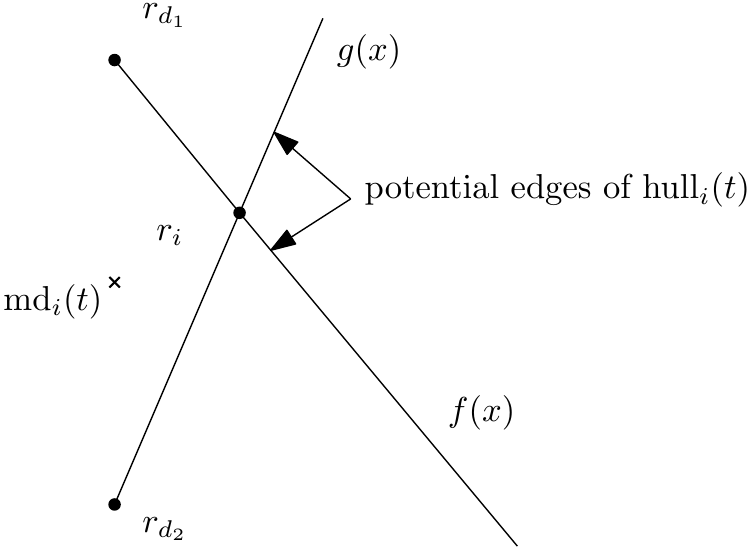}
		\caption{The robot $r_i$ wants to move towards $\mathrm{md}_i(t)$. Since we assume that $r_i$ is an edge of $\chull{i}$, there cannot be any robot beyond the lines described by the functions $f(x)$ and $g(x)$. }
		\label{figure:gtmdProofSetup}
	\end{figure}
	We focus on the function $f(x)$, the arguments for $g(x)$ are analogous.
	Next, we determine the largest $x$-coordinate a robot $r_j$ might have.
	Note that it must hold $|p_{d_1}(t), p_j(t)| < \robotDiameter{i}$, otherwise, the two robots that define $\robotDiameter{i}$ are not be unique and $r_i$ would move according to \gtcShort/.
	Hence,

	\begin{align*}
		\sqrt{x_j(t)^2 + \left(\frac{\left(2y_i(t) - \robotDiameter{i}\right)}{2x_i(t)}\right)^2 \cdot x_j(t)^2} < \robotDiameter{i} \iff x_j(t) < \frac{\robotDiameter{i}}{\sqrt{1+\left(\frac{\left(2y_i(t) - \robotDiameter{i}\right)}{2x_i(t)}\right)^2}}
	\end{align*}

	Next, we determine the maximal value of $y_i(t)$ given a fixed $x_i(t)$.
	\begin{align*}
		\sqrt{x_i(t)^2 + \left(y_i(t) + \frac{\robotDiameter{i}}{2}\right)^2} \leq 1 \iff y_i(t) \leq \frac{1}{2} \cdot \left(2\cdot \sqrt{1-x_i(t)^2} - \robotDiameter{i}\right)
	\end{align*}

	Thus, we obtain a more precise upper bound on $x_j(t)$:

	\begin{align*}
		x_j(t) < \frac{\robotDiameter{i}}{\sqrt{1+\left(\frac{\left( 2 \cdot \sqrt{1-x_i(t)^2}- 2\robotDiameter{i}\right)}{2x_i(t)}\right)^2}}
	\end{align*}

	One can verify $\frac{\robotDiameter{i}}{\sqrt{1+\left(\frac{\left( 2 \cdot \sqrt{1-x_i(t)^2}- 2\robotDiameter{i}\right)}{2x_i(t)}\right)^2}} \leq x_i(t) + \frac{2}{3} \robotDiameter{i}$ for $\frac{2}{\sqrt{3}} \leq \robotDiameter{i} \leq 2$.
	Hence, the maximal $x_j(t)$ is less than $x_i(t) + \frac{2}{3}\robotDiameter{i}$.
	Now, we consider the point $p_{\mathrm{dest}} = \left(\frac{9}{10}x_i(t), \frac{9}{10}y_i(t)\right)$ and bound the distance between $p_{\mathrm{dest}}$ and $p_j(t)$.
	It is easy to verify that $p_{\mathrm{dest}}$ is in distance of at most $\frac{1}{2}$ of any such robot $r_j$.
	The same arguments apply to the $r_j$ with minimal $x$-coordinate and thus, also for all in between.
	Hence, we obtain $\beta = \frac{1}{10}$.
\end{proof}
\newpage
\subsection{Pseudocode of \gtcdmbShort/}

\begin{algorithm}[htb]
	\caption{\gtcdmb/ (view of robot $r_i$)}
	\label{algorithm:gtcdmb}
	\begin{algorithmic}[1]
		\State $p_{d_1}(t), p_{d_2}(t) :=$ positions of robots in $N_i(t)$ that have the maximal distance

		\If{$p_{d_1}(t)$ and $p_{d_2}(t)$ are unique}
		\State rotate local coordinate system such that $p_{d_1}(t) = \left(0, \frac{\robotDiameter{i}}{2}\right)$ and  $p_{d_2}(t) = - p_{d_1}(t)$
		\State $x_{max} :=$ maximal $x$-coordinate of a robot $r_j \in N_i(t)$
		\State $x_{min} :=$ minimal $x$-coordinate of a robot $r_j \in N_i(t)$
		\State $p_{\mathrm{box}}(t) := \left(\frac{1}{2} \cdot \left(x_{min} + x_{max}\right),0\right) $
		\State $\forall r_j \in N_i(t) : m_j :=$ midpoint between $r_i$ and $r_j$
		\State $D_j: $ disk with radius $\frac{1}{2}$ centered at $m_j$
		\State $\mathrm{seg} := $ line segment $\overline{p_i(t), p_{\mathrm{box}}(t)}$
		\State $A := \bigcap_{r_j \in N_i(t)} D_j \cap \mathrm{seg}$
		\State $x :=$ point in $A$ closest to $p_{\mathrm{box}}(t)$
		\State $\mathrm{target}_i^{\gtcdmbShort/}(t) := x$
		\Else
		\State Compute $\mathrm{target}^{\gtcdmbShort/}_i(t)$ with \gtcShort/
		\EndIf
	\end{algorithmic}
\end{algorithm}

\subsection{Proof of \Cref{theorem:gtcdmbAlphaBeta}}

\TheoremGtcdmb*

\begin{proof}
	The constants can be proven analogously to the previous algorithm.
	The same proof as in \Cref{lemma:gtcAlphaEstimation2D} can be used to derive that the center of the minbox is $\frac{\sqrt{3}}{8}$-centered.
	The longest of the triangle simply needs to be replaced by the diameter, i.e., by the line segment connecting $p_{d_1}(t)$ and $p_{d_2}(t)$.
	Moreover, the constant $\beta = \frac{1}{10}$ can be proven with the same setup as used in the proof of  \Cref{theorem:gtmdContracting}.
\end{proof}

\newpage

\section{Details and Proofs of \Cref{section:collisionlessProtocols} \enquote{Collisionless Near-Gathering Protocols}} \label{section:appendixSection4}

\subsection{Detailed Intuition and Proof Outline} \label{section:collisionlessIntuition}

In the following, we describe the technical intuitions behind the protocol $\pCL$.
Since the intuition is closely interconnected with the formal analysis, we also give a proof outline here.
The proofs of all stated lemmas and theorems can be found in \Cref{section:collisionFreeAnalysis}.
The entire protocol $\pCL$ is described in \Cref{section:protocolPCl}.
As mentioned in the introduction of this section, the main intuition of the protocol $\pCL$ is the following:
robots compute a potential target point based on a \lambdaGathering/ $\calP$ (that uses a viewing range of $V$), restrict the maximum movement distance to $\nicefrac{\tau}{2}$ and use the viewing range of $V+\tau$ to avoid collisions with robots in the distance at most $\tau$.
However, there are several technical details we want to emphasize in this section.

For the correctness and the runtime analysis of the protocol $\pCL,$ we would like to use the insights into \lambdaContracting/ protocols derived in \Cref{section:alphaBetaContractingStrategies}.
However, since the robots compute their potential target point based on a \lambdaGathering/ $\calP$ with viewing range $V$, this point must not necessarily be \lambdaCentered/concerning the viewing range of $V+\tau$.
We discuss this problem in more detail in \Cref{section:intuitionPTau} and motivate the \emph{intermediate} protocol $\pTau$ that is \lambdaContracting/ with respect to the viewing range of $V+\tau$.
$\pTau$ is only an intermediate protocol since robots still may collide.
Afterward, we emphasize the importance of keeping \vubg{t} always connected.
We derive some general properties of \emph{not}-collisionfree \lambdaContracting/ protocols with a viewing range of $V+\tau$ while $\vubg{t}$ is always connected in \Cref{section:intuitionInitialStrongDistance}.
These properties especially hold for the intermediate protocol $\calP_{\tau}$.
Lastly, we argue how to transform the intermediate protocol $\calP_{\tau}$ into a collisionfree protocol $\pCL$ that is still \lambdaContracting/ in \Cref{section:intuitionCollisionAvoidance}.

\subsubsection{The protocol $\calP_{\tau}$} \label{section:intuitionPTau}

The main goal is to compute potential target points based on a \lambdaGathering/ protocol $\calP$ with viewing range $V$.
Let us ignore the collision avoidance in this section and only concentrate on the \lambdaContracting/ properties of such a protocol applied to a scenario with a viewing range of $V+\tau$.
Unfortunately, a direct translation of the protocol loses the \lambdaContracting/ property in general.
Consider the following example, which is also depicted in \Cref{fig:intuitionNotAlpha}.
Assume there is a robot $r_i$ that can only observe one other robot $r_j$ in distance $\nicefrac{2}{n}$.
Hence, $\robotDiameter{i} = \nicefrac{2}{n}$.
Now assume that under the protocol $\calP$ with viewing range $V$, it moves to the midpoint between itself and the other robot (every other valid \lambdaCentered/ point results in the same argumentation).
By definition, the midpoint is $1$-centered (\cref{def:alpha-centered}).
By increasing $r_i$'s viewing range to $V+\tau$, $r_i$ might observe a robot it could not see before.
In our example, it observes the robot $r_k$ which is exactly at a distance of $V+\tau$ from $r_i$ in opposite direction to $r_j$.
Thus, \robotDiameter{i} increases significantly to $\nicefrac{1}{n} + V + \tau$.
Now assume $r_i$ still moves to the midpoint between $r_i$ and $r_j$.
The maximal line segment with the target point as the center has length $\nicefrac{2}{n}$.
However, there is no \emph{constant} $\lambda$ such that $\nicefrac{2}{n} \geq \lambda \cdot \left(\nicefrac{1}{n} + V + \tau\right)$ for any $n$ ($\lambda$ depends on $n$).
Hence, in the example, the protocol $\calP$ is not \lambdaContracting/ (\cref{def:alpha-beta-contracting}) anymore because of the viewing range of $V+\tau$.

Next, we argue how to transform the protocol $\calP$ with viewing range $V$ into a protocol $\calP_{\tau}$ with viewing range $V+\tau$ such that $\calP_{\tau}$ is \lambdaGathering/.
The example above already emphasizes the main problem: robots can have very small local diameters $\robotDiameter{i}$.
Instead of moving according to the protocol $P$, those robots compute a target point based on a protocol $\pVTau$ that is \lambdaGathering/concerning the viewing range of $V+\nicefrac{\tau}{2}$ (the difference to $\pTau$'s viewing range is intended).
Usually, the protocol $\pVTau$ simulates $\calP$ for the larger instance by scaling the local coordinate system by factor $\frac{V}{V + \nicefrac{\tau}{2}}$.
For instance, $\calP$ can be \gtcShort/ that moves robots towards the smallest enclosing hypersphere (SEH) of all robots at a distance of at most $V$ and $\pVTau$ does the same with all robots at a distance of at most $V+\nicefrac{\tau}{2}$.
In general, however, $\pVTau$ could also be a different protocol (known to the robots to ensure collision avoidance later).
More precisely, robots $r_i$ with $\robotDiameter{i} \leq \nicefrac{\tau}{2}$ compute their target points based on $\pVTau$ and all others according to $\calP$.
In addition, $\calP_{\tau}$ ensures that no robot moves more than a distance of $\nicefrac{\tau}{2}$ towards the target points computed in $\calP$ and $\pVTau$.
This has two reasons.
The first reason is to maintain the connectivity of \vubg{t}.
While the protocol $\calP$ maintains connectivity by definition, the protocol $\pVTau$ could violate the connectivity of \vubg{t}.
However, restricting the movement distance to $\nicefrac{\tau}{2}$ and upper bounding $\tau$ by $\nicefrac{2}{3}$V resolves this issue since for all robots $r_i$ that move according to $\pVTau$, $\robotDiameter{i} \leq \nicefrac{\tau}{2}$.
Hence, after moving according to $\pVTau$, the distance to any neighbor is at most $\nicefrac{3}{2}\cdot \tau$.
Since $\tau$ is upper bounded by $\nicefrac{2}{3}V$, the distance is at most $V$ afterward.

\begin{restatable}{lemma}{LemmaConnectivityRange}
	\label{lem:connectivityrange}
	Let $\calP$ be a \lambdaGathering/ with viewing range of $V$.
	\vubg{t} stays connected while executing $\pTau$.
\end{restatable}

The second reason is that moving at most $\nicefrac{\tau}{2}$ makes sure that collisions are only possible within a range of $\tau$.
This is crucial for our collision avoidance which is addressed in the following section.

While $\pTau$ has a viewing range of $V + \tau$, it never uses its full viewing range for computing a target point.
Either, it simulates $\calP$ with a viewing range of $V$, or $\pVTau$ with one of $V +\nicefrac{\tau}{2}$.
Technically, the same problem as described above can still happen:
The robots in range of $V$ or $V + \nicefrac{\tau}{2}$ have a relatively small diameter while $\robotDiameter{i} > V +\nicefrac{\tau}{2}$.
Nevertheless, contrary to the example above, the robots in the smaller range cannot have an arbitrary small diameter in such a configuration.
$\pVTau$ is simulated if the robots in $\calP$ have a diameter $\leq \nicefrac{\tau}{2}$.
The above discussed \vubg{t} is connected.
It is observable that, if $\robotDiameter{i} \geq \nicefrac{\tau}{2}$, the $V +\nicefrac{\tau}{2}$ surrounding must have a diameter $\geq \nicefrac{\tau}{2}$.
The diameter of robots used for the simulation of $\calP$ or $\pVTau$ cannot be less than $\robotDiameter{i} \cdot \Omega(\nicefrac{\tau}{V})$.
The constant $\lambda$ mentioned above can be chosen accordingly.

\begin{restatable}{lemma}{LemmaPTauABGathering}
	\label{lem:pTau-alpha-beta-gathering}
	Let $\calP$ be a \lambdaGathering/.
	$\pTau$ is a \lambdaPrimeGathering/ with $\lambda' = \lambdaPrimeValuePTau$.
\end{restatable}

To conclude, the protocol $\calP_{\tau}$ has two main properties: it restricts the movement distance of any robot to at most $\nicefrac{\tau}{2}$ and robots $r_i$ with $\robotDiameter{i} \leq \nicefrac{\tau}{2}$ compute their target points based on protocol $\pVTau$ with viewing range $V+\nicefrac{\tau}{2}$.

\begin{figure}[htbp]
	\includegraphics[width =\textwidth]{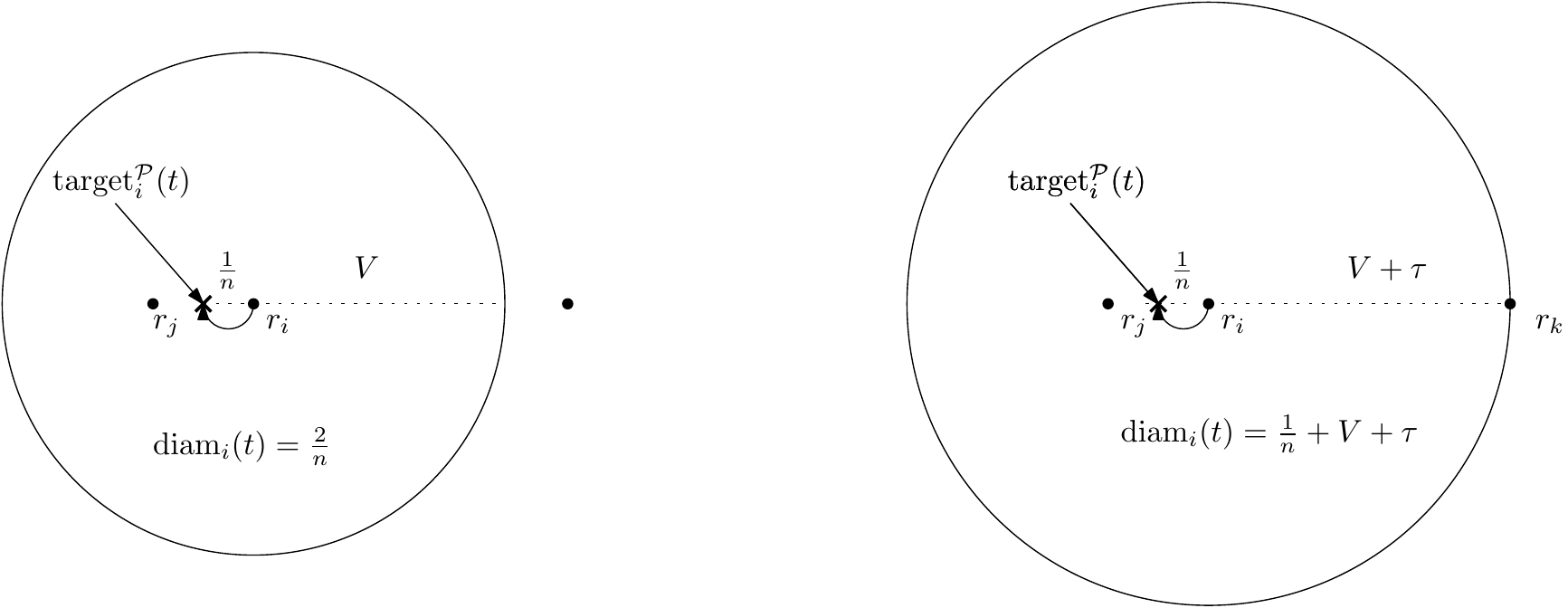}
	\caption{To the left, the local diameter of $r_i$ is $\nicefrac{2}{n}$ and it has a viewing range of $1$. The robot moves to \fp{i} which is the midpoint of its local diameter marked by a cross. Per definition, this is a $1$-centered point. The robot $r_k$ (outside the circle of radius $V$) is not visible to $r_i$. To the right, the same setup with a viewing range of $V+\tau$ is depicted, Now, $r_i$ can see also $r_k$. Hence, $\robotDiameter{i}$ increases to $\nicefrac{1}{n} + V+\tau$. However, if the target point remains unchanged, the point is not $\lambda$-centered anymore since the longest line segment that can be centered in \fp{i} still has a length of $\nicefrac{2}{n}$ while $\robotDiameter{i}$ has increased to $\nicefrac{1}{n} + V + \tau$. Thus there exists no \emph{constant} $\lambda$ anymore such that $\nicefrac{2}{n} >  \lambda \cdot \left(\nicefrac{1}{n} + V + \tau \right)$ since $\lambda$ would depend on $n$.}
	\label{fig:intuitionNotAlpha}
\end{figure}

\subsubsection{Implications of Smaller Connectivity Range} \label{section:intuitionInitialStrongDistance}

In the previous \Cref{section:intuitionPTau}, we have addressed the intermediate protocol $\pTau$ that is \lambdaGathering/concerning the viewing range of $V+\tau$ and also keeps \vubg{t} always connected.
Keeping \vubg{t} connected is important for the termination of a \nearGathering/ protocol.
Suppose that \vubg{t} is connected and the robots only have a viewing range of $V$.
Then, the robots can never decide if they can see all the other robots.
With a viewing range of $V+\tau$, however, it becomes possible.
The robots must be brought as close together such that $\globalDiameter < \tau$.
Now, each robot can see that all other robots are at a distance of at most $\tau$ and no other robot is visible, especially no robot at a distance $\mathrm{dist} \in (V, V+\tau]$ is visible.
Since \vubg{t} is connected, the robots can now decide that they can see all other robots, and \nearGathering/ is solved.
For any configuration where the viewing range is $V + \tau$ and \vubg{t} is connected, we can make an important observation.

\LemmaConstantDiameterByLargerVR*
%\begin{restatable}{lemma}{LemmaConstantDiameterByLargerVR}
%	\label{lem:constant-diameter-by-larger-vr}
%	Let $\calP$ be a \lambdaContracting/ protocol with viewing range $V + \tau$ for any constant $\tau > 0$ and let \vubg{t} be connected.
%	If $\globalDiameter > \tau$ (maximal distance of any pair of robots in round $t$), then also $\robotDiameter{i} > \tau$  (maximal distance of any pair of robots in $r_i$'s neighborhood), for every robot $r_i$.
%\end{restatable}

This leads directly to another helpful insight.
The \lambdaContracting/ property is defined in a way, that robots close to the boundary of the global SEH always move at least $\Omega\left(\frac{\robotDiameter{i}}{\Delta}\right)$ inside the SEH when they are active.
With $\robotDiameter{i} > \tau$ it follows that the radius of the global SEH decreases by $\Omega(\nicefrac{\tau}{\Delta})$ after each robot was active at least once (see \cref{lemma:largeDiameter-cl}).
Consequently, $\globalDiameter \leq \tau$ after $\calO(\Delta^{2})$ epochs.

\LemmagatheringWithIncreasedVR*
%\begin{restatable}{lemma}{LemmagatheringWithIncreasedVR}
%	\label{lem:gathering-with-increased-vr}
%	Let $\calP$ be a \lambdaContracting/ protocol with a viewing range of $V+\tau$ while \vubg{t} is always connected.
%	After at most $\runningTimeLemmaCL \in \calO(\Delta^2)$ epochs executing $\calP$, $\globalDiameter \leq \tau$.
%\end{restatable}

\subsubsection{Collision Avoidance} \label{section:intuitionCollisionAvoidance}

Next, we argue how to transform the protocol $\pTau$ (based on the protocol $\calP$) into the collisionfree protocol $\pCL$.
Recall that $\calP$ uses a viewing range of $V$, and $\pCL$ has a viewing range of $V+\tau$.

%For a moment, let us assume $\pTau$ has a viewing range of $V$.
The larger viewing range in $\pCL$ allows a robot $r_i$ to compute $\fpTau{k}$ (the target point in protocol $\pTau$) for all robots $r_k$ within distance at most $\tau$.
Since the maximum movement distance of a robot in $\calP$ is $\nicefrac{\tau}{2}$, this enables $r_i$ to know the movement directions of all robots $r_k$ which can collide with $r_i$.
%\begin{definition}[Collision vector $\ellp{i}$]
%	\label{def:collision-vector}
%	Let $\calP$ be a \lambdaGathering/.
%	The vector $\ellp{i}$ pointing from $p_i(t)$ to $\fp{i}$ is the \emph{collision vector} of $r_i$.
%	%$\fp{i}$ is deterministic but not necessary invariant under rotation and reflection (see Paragraph \ref{par:target-function-not-invariant}).
%	%The set of collision vectors is defined as follows.\\
%	%$ :=  \left\{ \textrm{vector from } p_i(t) \textrm{ to } x : \textrm{ there exist rotation and reflection such that } \fp{i} = x  \right\} $.
%\end{definition}
We will ensure that each robot $r_i$ moves to some position on $\ellpTau{i}$ and avoids positions of all other $\ellpTau{k}$.
Henceforth, no collision can happen.
While this is the basic idea of our collision avoidance, there are some details to add.

First of all, $\pTau$ has the same viewing range as $\pCL$ of $V + \tau$.
However, it never uses the full viewing range to compute the target position $\fpTau{i}$.
We consider two robots $r_i$ and $r_k$ with distance $\leq \tau$.
If $r_k$ simulates $\calP$ to compute $\fpTau{k}$, $r_i$ can compute $\fpTau{k}$ as well since $r_i$ is able to observe all robots in distance $V$ around $r_k$.
If $r_k$ simulates $\pVTau$, the condition in $\pTau$ makes sure that $r_i$ and $r_k$ have a distance of $\leq \nicefrac{\tau}{2}$.
Similarly, $r_i$ is able to observe all robot in distance $V + \nicefrac{\tau}{2}$ around $r_k$ and can compute $\fpTau{k}$ as well.

\begin{restatable}{lemma}{LemmaViewingRange}
	\label{lem:viewingrange}
	Let $\calP$ be a \lambdaGathering/ with a viewing range of $V$.
	A viewing range of $V + \tau$ is sufficient to compute $\fpTau{k}$ for all robots $r_k$ within a radius of $\tau$.
\end{restatable}

Secondly, $r_i$ can't avoid positions on all other $\ellpTau{k}$ in some cases.
For instance, $\ellpTau{i}$ may be completely contained in $\ellpTau{k}$ (e.g., $\ellpTau{2} \in \ellpTau{1}$ in the example depicted in \cref{fig:collisionless-algorithm-and-collision-points}).
In case $\ellpTau{i}$ and $\ellpTau{k}$ are not collinear and intersect in a single point, both robots simply avoid the intersection point (e.g. $r_1$ and $r_4$ in the example).

\begin{restatable}{lemma}{LemmaIntersectionNotCollinear}
	\label{lem:intersection-not-co-linear}
	%Let $v_i$ be the collision vector between $p_i(t)$ and $\fpTau{i}$.
	No robot moves to a point that is the intersection of two collision vectors that are not collinear.
\end{restatable}

If $\ellpTau{i}$ and $\ellpTau{k}$ are collinear, both robots move to a point closer to their own target point than to the other one (e.g., $r_1$ and $r_3$ in the example).

\begin{restatable}{lemma}{LemmaFPTauNeq}
	\label{lem:fpTau-neq}
	If the target points of robots are different in $\pTau$ they are different in $\pCL$.
\end{restatable}

But there are cases, in which robots have the same target point in $\pTau$ (e.g. $r_1, r_2$ and $r_6$ in the example).
Because robots stay in the same direction towards the target point, collisions can only happen if one robot is currently on the collision vector of another one (e.g., $r_2$ is on $\ellpTau{1}$).
Their movement is scaled by the distance to the target point, which must be different.
Therefore, their target points in $\pCL$ must be different as well.
\begin{restatable}{lemma}{LemmaCollinearCollisionVector}
	\label{lem:co-linear-collision-vector}
	If the target points of robots are the same in $\pTau$ they are different in $\pCL$.
\end{restatable}

In \ssync\, robots may be inactive in one round.
Nevertheless, in the same way, single intersection points between collision vectors and the positions of other robots are avoided as well.

\begin{restatable}{lemma}{LemCollisionlessInactive}
	\label{lem:collisionless-inactive}
	No robot moves to the position of an inactive robot.
\end{restatable}

The following lemma follows immediately from \cref{lem:fpTau-neq}, \ref{lem:co-linear-collision-vector} and  \ref{lem:collisionless-inactive}.
\begin{restatable}{lemma}{LemCollisionless}
	\label{lem:no-early-collison}
	The protocol $\pCL$ is collisionless.
\end{restatable}

%\begin{figure}
%
%	\includegraphics[width=\linewidth]{figures/collisionless-algorithm-and-collision-points}
%	\caption{Example of $\fpCL{i}$ with $V = 1, \tau = 2/3$ and $\varepsilon = 0.99$. $(i)$ shows the collision points and computation of $d_1, d_2$ and $d_3$ (line \ref{line:d-i:clGTC} in \cref{algorithm:collisionlessGTC}). $(ii)$ shows the positions where $r_1, r_2$ and $r_3$ will move to in protocol $\pCL$ as returned by \cref{algorithm:collisionlessGTC}.}
%	\label{fig:collisionless-algorithm-and-collision-points}
%\end{figure}

Because $\pTau$ is \lambdaContracting/ and preserves the connectivity of \vubg{t}, also $\pCL$ has these properties. (\cref{lem:movement-1-over-epsilon,lem:connectivityrange-cl}).
This allows us to follow from \cref{lem:gathering-with-increased-vr} that \nearGathering/ is solved after $\calO(\Delta^{2})$ epochs.

\subsection{Analysis} \label{section:collisionFreeAnalysis}

\subsubsection{Time Bound}

\LemmaConstantDiameterByLargerVR*

\begin{proof}
	We prove the claim by contradiction.
	Let the initial UDG be connected with a radius $V$ and  $\globalDiameter > \tau$.
	To derive the contradition, we assume that there is a robot $r_i$ with $\robotDiameter{i} \leq \tau$.
	By definition, $|p_i(t) - p_k(t)| \leq \tau$ for all $r_k \in N_i(t)$ (the neighborhood of $r_i$).
	Consequently, there exists at least one robot $r_j \notin N_i(t)$.
	For all robots $r_j \notin N_i(t)$, $|p_i(t) - p_j(t)| > V + \tau$.
	For all $r_k \in N_i(t)$ and $r_j \notin N_i(t)$, we have $|p_k(t) - p_j(t)| > V$.
	Hence, none of $r_i$'s neighbors is at a distance of at most $V$ from a robot that $r_i$ cannot see.
	This is a contradiction since we have assumed that the UDG with radius $V$ is connected.
	Hence, $\robotDiameter{i} > \tau$.
\end{proof}

To derive the runtime of $\calO\left(\Delta^2\right)$, we use the course of the analysis of \lambdaContracting/ protocols presented in \Cref{section:alphaBetaContractingStrategies}.
The analysis can be found in \Cref{section:alphaBetaProtocolsUpperBound}.
Most of the notation is identical. However, the main circular segment of the analysis is defined slightly differently.
Let $N := N(t)$ (we omit the time parameter for readability) the (global) SEH of all robots in round $t$ and $R := R(t)$ its radius.
Now, fix any point $b$ on the boundary of $N$.
Let $\tau > 0$ be any constant.
The two points in distance $\tau \cdot \nicefrac{\lambda}{8}$ of $b$ on the boundary of $N$ determine the circular segment $S_{\lambda \cdot \tau}$ with height $h$.
In the following, we determine by $S_{\lambda \cdot \tau}(c)$ for $0 < c\leq 1$ the circular segment with height $c \cdot h$ that is contained in $S_{\lambda \cdot \tau}$.

\begin{lemma} \label{lemma:largeDiameter-cl}
	Let $\calP$ be a \lambdaContracting/ Protocol.
	For a robot $r_i$ with $\robotDiameter{i} > \tau$, $\fp{i} \in N \setminus \baseSegmentC/$.
\end{lemma}

\begin{proof}
	Since $\robotDiameter{i} > c$ and $\mathcal{P}$ is \lambdaContracting/, $\fp{i}$ is the midpoint of a line segment $\ell^\mathcal{P}_i(t)$ of length at least $\lambda \cdot \robotDiameter{i} > \lambda \cdot \tau$.
	Observe that the maximum distance between any pair of points in \baseSegmentC/ is $\tau \cdot \nicefrac{\lambda}{4}$.
	It follows that $\ell^\mathcal{P}_i(t)$ either connects two points outside of \baseSegmentC/ or one point inside and another point outside.
	In the first case, \fp{i} (the midpoint of $\ell^{i}_\mathcal{P}(t)$) lies outside of \baseSegmentC/.
	In the second case,  \fp{i} lies outside of \baseSegmentC/ as well, because it is the midpoint of $\ell_i^{\mathcal{P}}(t)$ and one half of this line segment is longer than $\frac{\lambda \cdot \robotDiameter{i}}{2} > \lambda \cdot \nicefrac{\tau}{2}$ and the maximum distance between any pair of points in \baseSegmentC/ is $\tau \cdot \nicefrac{\lambda}{4}$.
	%
	%	The rest of the proof is analogous to the proof of \cref{lemma:largeDiameter} but stated here for the sake of completeness.
	%	Now, consider any vertex $v$ of \chull{i}.
	%	Rotate the global coordinate system, such that $h$ is on the $y$-axis.
	%	Let $h_v$ be the vertical distance of $v$ and \alphaCenterP{i}.
	%	Furthermore, define $v_{\beta}$ to be the scaled point $v$ in $\chull{i}\bigl(\alphaCenterP{i},\beta\bigr)$.
	%	By definition, it follows that the vertical distance between \alphaCenterP{i} and $v_{\beta}$ is $(1-\beta) \cdot h_v$.
	%	Hence, every vertex of \chull{i} located in  \betaCSegment/, is scaled outside of   \betaCSegment/.
	%	Therefore, $\chull{i}\bigl(\alphaCenterP{_i},\beta\bigr)$ lies outside of  \betaCSegment/ and, thus, also \fp{i}.
\end{proof}

\LemmagatheringWithIncreasedVR*
\begin{proof}
	As long $\globalDiameter > \tau$, we have by the preliminaries that $\robotDiameter{i} > \tau$.
	From \Cref{lemma:largeDiameter-cl}, it follows that a robot leaves $\baseSegmentC/$ when in becomes active.
	This happens for all robots at most once per epoch.
	Hence, $R(t)$, the radius of the (global) SEH of all robots in round $t$, decreases by $h$ in one epoch, where $h$ denotes the height of $\baseSegmentC/$.
	In \Cref{lemma:heightHighDim}, we have analyzed the height $h'$ of the \baseCapLambda/.
	We proved $h' \geq \frac{\sqrt{2} \cdot \lambda^2}{64 \cdot \pi \cdot \Delta}$.
	Here, we analyze the $HSC$ \baseSegmentC/ whose slant height of the inscribed hypercone is reduced by a factor of $\tau$.
	Hence, also $h = \tau \cdot h'$.
	%As a consequence, the height of \betaCSegment/ is $\frac{\beta}{2} \cdot \tau \cdot h'$.
	Thus, $R(t)$ decreases by at least $\frac{\tau \cdot \sqrt{2} \cdot \lambda^2}{64 \cdot \pi \cdot \Delta}$ in one epoch.
	By \cref{theorem:jungsTheorem} we know that the initial global SEH has a radius of at most $\nicefrac{\Delta}{\sqrt{2}}$.
	After $\frac{\nicefrac{\Delta}{\sqrt{2}}}{h} = \frac{\Delta \cdot 64 \cdot \pi \cdot \Delta}{\sqrt{2} \cdot \tau \cdot \sqrt{2} \cdot \lambda^2} = \runningTimeLemmaCL$ epochs the global SEH has a radius $\leq \nicefrac{\tau}{2}$ and $\globalDiameter \leq \tau$.

\end{proof}

\subsubsection{Analysis of $\pTau$}

\LemmaConnectivityRange*
\begin{proof}
	Any protocol $\calP$ must hold the connectivity concerning $V$.
	However, with the larger viewing range of $V+\nicefrac{\tau}{2}$, the protocol $\pVTau$ only guarantees a connectivity of \tauUbg{t}.
	Now suppose that a robot $r_i$ moves according to $\pVTau$ in $\pTau$.
	This only happens, if there is no robot in distance $\mathrm{dist} \in (\nicefrac{\tau}{2}, V]$ around $r_i$, all connected robots $r_k$ have a distance $\leq \nicefrac{\tau}{2}$ before the movement.
	$r_i$ and $r_k$ can both move at most a distance of $\nicefrac{\tau}{2}$ in one round.
	It follows that in the next round their distance is $\leq 3\cdot \nicefrac{\tau}{2} \leq V$ because $\tau \leq \nicefrac{2}{3} V$ by definition.

\end{proof}

\LemmaPTauABGathering*
\begin{proof}

	By \cref{lem:connectivityrange,lem:constant-diameter-by-larger-vr}, it follows that $\robotDiameter{i} \geq \nicefrac{\tau}{2}$ for all $i$.
	The protocol $\pTau$ has a viewing range of $V + \tau$.
	$\fpTau{i}$ either equals to $\fp{i}$ ($\calP$ with viewing range $V$) or $\fpVTau{i}$ ($\pVTau$ with viewing range $V+\nicefrac{\tau}{2}$), dependent on the local diameter $\robotDiameter{i}$ of a robot.
	If  $\fpTau{i} = \fp{i}$, we know that a line segment with length $\lambda \cdot \nicefrac{\tau}{2}$ exists with $\fp{i}$ as midpoint, because the robots used to simulate have at least a diameter of $\nicefrac{\tau}{2}$ (see condition in line \ref{line:cond-1:tPclGTC} of \cref{algorithm:targetPointCollisionlessGTC}) and $\calP$ is \lambdaContracting/.

	By \cref{lem:connectivityrange} we know that \vubg{t} stays connected.
	$\pVTau$ has a viewing range of $V + \nicefrac{\tau}{2}$
	By \cref{lem:constant-diameter-by-larger-vr}, it follows that the diameter of robots used for simulating $\pVTau$ is $\geq \nicefrac{\tau}{2}$ if $\globalDiameter \geq \nicefrac{\tau}{2}$.
	Because $\pVTau$ is \lambdaContracting/, there exists a line segment with length $\lambda \cdot \nicefrac{\tau}{2}$ trough $\fpVTau{i}$.

	The local diameter is naturally bounded by $\robotDiameter{i} \leq 2(V + \tau)$.
	The length of the above described line segment with $\pTau$ as midpoint is  $\lambda \cdot \nicefrac{\tau}{2} = \lambda \cdot \frac{\nicefrac{\tau}{2}}{2(V + \tau)} \cdot 2(V + \tau) \geq \lambda \cdot \frac{\tau}{4(V + \tau)} \cdot \robotDiameter{i}$.
	Therefore, $\pTau$ is \lambdaPrimeGathering/ with $\lambda' = \lambdaPrimeValuePTau$.

\end{proof}

\subsubsection{Analysis of $\pCL$}

\LemCollisionless*

The lemma follows directly from \cref{lem:fpTau-neq,lem:co-linear-collision-vector,lem:collisionless-inactive} which are introduced in the following.

\LemmaViewingRange*

\begin{proof}

	%Robot $r_i$ uses function $\collisionPoints{\pTau}{i}$ with $R_i = \{r_k : |p_i^{k}(t) - p_i^{i}(t)| \leq \tau\}$ the robots in radius $\tau$ around $r_i$ (line \ref{line:R-i:clGTC} \cref{algorithm:collisionlessGTC}).
	%For all $r_k \in R_i$ is $\ellpTau{k}$ computed in $r_i$'s local coordinate system (line \ref{line:all-v-in-ell:CP} \cref{algorithm:CollisonPointsOnLine}).
	%$\ellpTau{k}$ computes all possible collision vectors of $r_k$ in round $t$.
	%This are the vectors $\ellpTau{k} = \left\{ \textrm{vector from } p_k(t) \textrm{ to } x : \textrm{ there exist rotation and reflection such that } \fpTau{k} = x  \right\}$
	The computation requires that the entire neighborhood of $r_k$ relevant to compute $\fpTau{k}$ is also in $r_i$'s neighborhood.
	Depending on whether there exists a robot in distance $\mathrm{dist} \in (\nicefrac{\tau}{2}, V]$ around $r_k$, is this the neighborhood relevant for $\fp{k}$ or $\fpVTau{k}$ (first if/else block in \cref{algorithm:targetPointCollisionlessGTC}).
	%The distance between robots is invariant under reflection and rotation; therefore is it sufficient to compute $\ellpTau{k}$ when there exist no robot in distance $d \in (\nicefrac{\tau}{2}, V]$ and $\ellp{k}$ else.

	$\fp{k}$ needs a viewing range of $V$ around $r_k$.
	The distance between $r_k$ and $r_i$ is at most $\tau$ and $r_i$ has a viewing range of $V + \tau$, therefore all robots relevant for computing $\fp{k}$ are in $r_i$'s neighborhood.
	$\fpVTau{k}$ needs a viewing range of $V+\nicefrac{\tau}{2}$ around $r_k$.
	The condition that the pairwise distance between robots in range $V$ around $r_k$ is $\leq \nicefrac{\tau}{2}$ makes sure that $|p_i(t) - p_k(t)| \leq \nicefrac{\tau}{2}$.
	$r_i$ has a viewing range of $V + \tau$, therefore are all robots relevant for computing $\fpVTau{k}$ are in $r_i$'s neighborhood.
\end{proof}

\LemmaIntersectionNotCollinear*

\begin{proof}
	More formally, we prove the following statement:
	Let $\ellpTau{i}$ and $\ellpTau{k}$ be collision vectors which intersect in a single point $I$.
	$\fpCL{i} \neq I$.
	$r_i$ and $r_k$ have a distance of at most $\tau$, because the movement distance of $\nicefrac{\tau}{2}$ is an upper bound for the length of $\ellpTau{i}$, respectively $\ellpTau{k}$.
	Hence, $r_k$ is in $R_i$ as computed in line \ref{line:R-i:clGTC} of \cref{algorithm:collisionlessGTC}.
	In line \ref{line:all-v-in-ell:CP} of \cref{algorithm:targetPointCollisionlessGTC}, the collision vector $\ellpTau{k}$ is checked for intersections with $\ellpTau{i}$.
	By \cref{lem:viewingrange}, we know, that $\ellpTau{k}$ is computable by $r_i$ with the available viewing range of $V + \tau$.
	It follows that $I$ is in $C_i$ (l. \ref{line:C-i:clGTC} Algorithm \ref{algorithm:collisionlessGTC}).
	$\fpCL{i}$ is some point in between the nearest points in $C_i \setminus \{p_i\}$ and $\fpTau{i}$.
	This can never be $I$.
\end{proof}

\LemmaFPTauNeq*

\begin{proof}
	More formally, we prove the following statement:
	Let $r_i$ and $r_k$ be two robots with $\fpTau{i} \neq \fpTau{k}$.
	It follows that $\fpCL{i} \neq \fpCL{k}$.
	If $\ellpTau{i}$ and $\ellpTau{k}$ are not collinear, the statement follows directly from \cref{lem:intersection-not-co-linear}.
	We consider both robots with collinear $\ellpTau{i}$ and $\ellpTau{k}$.
	Let $P_i = \fpTau{i}$, $P_k = \fpTau{k}$.
	We distinguish all three cases how $r_i, P_i$ and $P_k$ can be arranged: $P_k$ is between $r_i$ and $P_i$; $r_i$ is between $P_i$ and $P_k$; $P_i$ is between $r_i$ and $P_k$.
	\begin{itemize}
		\item Case $P_k$ is between $r_i$ and $P_i$:
		Analogous to the arguments in \cref{lem:intersection-not-co-linear}, $r_k \in R_i$ (line \ref{line:R-i:clGTC} \cref{algorithm:collisionlessGTC}) and $P_k$ is added to $C_i$ (line \ref{line:add-fpk:CP} \cref{algorithm:CollisonPointsOnLine}).
		$d_i$ (line \ref{line:d-i:clGTC} \cref{algorithm:collisionlessGTC}) is at most the distance between $P_i$ and $P_k$.
		$r_i$ stops a distance of $\targetPointDist{i}$ away from $P_i$.
		By definition is $\varepsilon < 0.5$ and $\big|\ellpTau{i}\big| \leq \nicefrac{\tau}{2}$.
		It follows $\targetPointDist{i} \leq d_i \cdot \varepsilon < \nicefrac{d_i}{2}$.
		$r_i$ will move onto a point closer to $P_i$ than to $P_k$.
		\item Case $r_i$ is between $P_i$ and $P_k$:
		$r_i$ will move onto a point closer to $P_i$ than to $P_k$ because $d_i \leq \big|\ellpTau{i}\big|$ which is in this case less than the distance between $P_i$ and $P_k$.
		\item Case $P_i$ is between $R_i$ and $P_k$:
		$r_i$ will move to a point between its current position and $P_i$, this is naturally closer to $P_i$ than to $P_k$.
	\end{itemize}
	In all cases, $\fpCL{i}$ is closer to $P_i$ than to $P_k$ and analogously, $\fpCL{k}$ is closer to $P_k$ than to $P_i$.
	Hence, $\fpCL{i} \neq \fpCL{k}$.

	%	A third case is where one collision vector, w.l.o.g. $\ellpTau{k}$, has a length of 0 because the target point is equivalent to the current position.
	%	The robot does not move at all in such a round.
	%	Robots start a round with unique positions when no collision happened earlier.
	%	A collision with $r_k$ can be avoided by not moving to current positions of other robots.
	%	When $r_k$ is on $\ellpTau{i}$, analog to the arguments in \cref{lem:intersection-not-co-linear} is $r_k \in R_i$ (line \ref{line:R-i:clGTC} in \cref{algorithm:collisionlessGTC}), $p_k(t)$ is added to $C_i$ (line \ref{line:v-k-length-0:CP} in \cref{algorithm:CollisonPointsOnLine}) and $r_i$ will not move to it.
\end{proof}

\LemmaCollinearCollisionVector*

\begin{proof}
	More formally, we prove the following statement:
	Let $r_i$ and $r_k$ be two robots with $\fpTau{i} = \fpTau{k}$.
	It follows that $\fpCL{i} \neq \fpCL{k}$.

	Let $P = \fpTau{i} = \fpTau{k}$.
	A robot moving towards $P$ will stay on the same side of $P$, and none will reach $P$.
	So collisions can solely happen if $\ellpTau{i}$ and $\ellpTau{k}$ are collinear pointing from the same side to $P$.
	We consider this case.
	W.l.o.g., let $r_i$ be closer to $P$ than to $r_k$.
	$d_i$, respectively $d_k$, is computed by the point in $C_i \setminus \{P\}$, respectively $C_k \setminus \{P\}$, with minimal distance to $P$ (line \ref{line:d-i:clGTC} in \cref{algorithm:collisionlessGTC}).
	Let this be $c_i \in C_i$, respectively $c_k \in C_k$.
	We assume $c_i \neq c_k$.
	From $c_i \neq c_k$, it follows directly $c_i \notin C_k$ or $c_k \notin C_i$.
	\begin{itemize}
		\item Case $c_i \notin C_k$:
		$r_i$ and $r_k$ are chosen in a way that $\ellpTau{i} \subset \ellpTau{k}$, it follows $c_i$ is also on $\ellpTau{k}$.
		$c_i$ is a point on $\ellpTau{j}$ the collision vector of some robot $r_j$ (see \cref{algorithm:CollisonPointsOnLine}).
		$\big|\ellpTau{j}\big| + \big|\ellpTau{k}\big| \leq \tau$ is an upper bound for the distance between $r_k$ and $r_j$.
		$r_j$ must be in $R_k$ as computed in line \ref{line:R-i:clGTC} of \cref{algorithm:collisionlessGTC}.
		Hence, $\ellpTau{j}$ is checked for collisions and $c_i$ must be in $C_k$.

		\item Case $c_k \notin C_i$:
		Similar to the arguments above, $p_i(t)$ is the position of robot $r_i$, in $C_k$.
		The distance of $c_k$ to $P$ is therefore is not larger than the distance $\big|\ellpTau{i}\big|$ (otherwise would $p_i(t)$ be nearer to $P$ than the chosen collision point $c_k$ with minimal distance to $P$).
		It follows that $c_k$ is also on $\ellpTau{i}$.
		Analogous to the case above, $c_k \in C_i$.
	\end{itemize}
	$c_i = c_k$ and $d_i = d_k$, accordingly.
	$\big|\ellpTau{i}\big| \neq \big|\ellpTau{k}\big|$, otherwise would $r_i$ and $r_k$ be at the same position and a collision has happened earlier.
	It follows that $\targetPointDist{i} \neq \targetPointDist{k}$  in every case such that $r_i$ and $r_j$ move to different positions.
\end{proof}

\LemCollisionlessInactive*

\begin{proof}
	A robot cannot know which robots are active or inactive.
	However, the algorithm is designed so that no robot moves to the current position of any robot.
	This can be proven analogously to \cref{lem:intersection-not-co-linear} because in line \ref{line:v-k-length-0:CP} of \cref{algorithm:CollisonPointsOnLine} the positions of robots on the collision vector are added to the set of collision points.
\end{proof}

%\LemCollisionless*
%
%\begin{proof}
%	Let $r_1, \cdots, r_n$ be the robots.
%	Robot $r_i$ moves on a straight line towards $\fpTau{i}$ in one round $t$.
%	Due to collision avoidance it may not reach $\fpTau{i}$, but its new position will always be on the collision vector $\ellpTau{i}$.
%	$r_i$ can only collide with $r_k$ on the intersection between $\ellpTau{i}$ and $\ellpTau{k}$.
%	When $\ellpTau{i}$ and $\ellpTau{k}$ are not co-linear, this is only one point.
%	$r_i$ knows $\ellpTau{i}$ and can compute $\fpTau{k}$ and $\ellpTau{k}$ because it is independent of $r_k$'s local coordinate system.
%	%But it can compute the set $\ellpTau{k}$ which contains $v_k$.
%	In \cref{lem:intersection-not-co-linear} it is stated, that $r_i$ does not move to the intersection point between $\ellpTau{i}$ and any $\ellpTau{k}$ which is not co-linear to $\ellpTau{i}$.
%
%	When $\ellpTau{i}$ and $\ellpTau{k}$ are collinear, it is not sufficient to avoid certain points.
%	Each robot moving on the same line $l$ must find a unique position on $l$.
%	In \cref{lem:co-linear-collision-vector} it is stated, that this happens for $\pCL$.
%
%
%\end{proof}

\begin{restatable}{lemma}{LemmaPclABContracting}
	\label{lem:movement-1-over-epsilon}
	If $\calP$ is a \lambdaGathering/, $\pCL$ is \lambdaPrimeContracting/ with $\lambda' = \lambdaPrimeValuePcl$.
\end{restatable}

\begin{proof}

	From \cref{lem:pTau-alpha-beta-gathering}, we obtain that $\pTau$ is \lambdaPrimePrimeContracting/ with $\lambda'' = \lambdaPrimeValuePTau$.
	Because $\pTau$ and $\pCL$ have the same viewing range, \fpTau{i} is always \lambdaCentered/ for $\pCL$.

	$\fpCL{i}$ is chosen such that a robot moves in direction $\fpTau{i}$ and $\frac{\fpCL{i}}{\fpTau{i}} \geq (1-\varepsilon)$.
	Analogous to the arguments in \cref{lem:ab-contracting-is-lambda-contracting}, we can follow with the intercept theorem, that $\lambda' = \lambda'' \cdot (1-\varepsilon) = \lambdaPrimeValuePcl$.

\end{proof}

\begin{restatable}{lemma}{LemmaConnectivityRangeCL}
	\label{lem:connectivityrange-cl}
	Let $\calP$ be a \lambdaGathering/ with viewing range of $V$.
	\vubg{t} stays connected while executing $\pCL$.
\end{restatable}

\begin{proof}
	If the \vubg{t} is connected in the initial configuration, it will stay connected while executing $\pTau$ (\cref{lem:connectivityrange}).
	Because of the semi-synchronous environment, we know that $r_i, r_k$ with $|p_i(t)-p_k(t)| \leq V$ implies $\big|\fpTau{i}-\fpTau{k}\big| \leq V$ and $\big|\fpTau{i}-p_k(t)\big| \leq V$, respectively $\big|p_i(t)-\fpTau{k}\big| \leq V$ (in case $r_k$, respectively $r_i$, is inactive in round $t$ and $p_k(t) = p_k(t+1)$).
	If both endpoints of both collision vectors $\ellpTau{i}$ and $\ellpTau{k}$ are pairwise at a distance $\leq V$, all points on both vectors are pairwise at a distance $\leq V$.
	For all robots $\fpCL{i} \in \ellp{i}$ and therefore, \vubg{t} stays connected while executing $\pCL$.
\end{proof}

\subsubsection{Proof of \cref{thm:collisionless-class}}
\label{subsec:proof-collisionless-thm}
The theorem is stated with a more precise formulation but the same meaning as this section's beginning.
\begin{customthm}{\ref{thm:collisionless-class}}

	Let $\nicefrac{2}{3} \cdot V \geq \tau > 0$ and $0.5 > \varepsilon > 0$.
	For every \lambdaGathering/ $\calP$ there exists the protocol $\pCLtauEpsilon$ with the following properties.
	\begin{itemize}
		\item $\pCL$ is a \clLambdaPrimeContracting/ protocol with $\lambda' = \lambdaPrimeValuePcl$.
		\item Let $V$ be the viewing range of $\calP$. $\pCL$ has a viewing range of $V + \tau$.
		\item $\pCL$ results in a near-gathering with diameter $\leq \tau$ of all robots in $\runningTimeLemmaCLPrime \in \mathcal{O}(\Delta^2)$ epochs, if \vubg{0} is connected.
	\end{itemize}
\end{customthm}

\begin{proof}
	\cref{lem:no-early-collison} states that $\pCL$ is collisionless and \cref{lem:movement-1-over-epsilon} that $\pCL$ as \lambdaPrimeContracting/ with $\lambda' = \lambdaPrimeValuePcl$.
	%The viewing range is shown in \cref{lem:viewingrange}
	If the \vubg{t} is connected in the initial configuration, it will stay connected while executing $\pCL$ (\cref{lem:connectivityrange-cl}).
	By \cref{lem:constant-diameter-by-larger-vr}, it follows that the $\robotDiameter{i} > \tau$ if $\globalDiameter > \tau$.
	This is the preliminary for \cref{lem:gathering-with-increased-vr}, which can therefore be applied to $\pCL$ to show that after $\runningTimeLemmaCLPrime$ epochs a near-gathering happened.

\end{proof}

\end{document}